\newif\ifjsl
\newif\ifarxiv
\arxivtrue 

\ifjsl
\documentclass{asl}
\makeatletter\@namedef{ver@amsmath.sty}{}\makeatother \usepackage[utf8]{inputenc}
\usepackage[linkcolor=black,linkstyle=none,theorems=false,mnsymbol=true,tikz,newdocumentcommand]{generic} 
\usepackage[xcolor,hyperref,notion,quotation,electronic]{knowledge}
\usepackage{titletoc}
\makeindex
\newcommand\jslversion
\fi

\ifarxiv
\documentclass{article}
\usepackage[utf8]{inputenc}
\usepackage[linkcolor=black,linkstyle=none,mnsymbol=true,tikz,newdocumentcommand]{generic} 
\usepackage[left=4.5cm,right=4.5cm]{geometry}
\usepackage[xcolor,hyperref,notion,quotation,electronic]{knowledge}
\usepackage{titletoc}
\makeindex
\fi

\knowledgestyle*{intro notion}{color=black, emphasize}
\knowledgestyle*{notion}{color=black} \knowledgestyle*{unknown}{color=orange}

\knowledge{Olivier Carton}[OLIVIER CARTON]{url={https://www.irif.fr/~carton}}
\knowledge{Gabriele Puppis}[GABRIELE PUPPIS]{url={http://www.labri.fr/perso/gpuppis}}
\knowledge{Thomas Colcombet}[THOMAS COLCOMBET]{url={https://www.irif.fr/~colcombe}}

\knowledge{MSO}[monadic second-order]{notion}

\knowledge{linear ordering}[linear orderings|order|orders|ordering|linear order|linear orders]{notion}
\knowledge{reverse ordering}[\rev]{notion}
\knowledge{order type}[order types]{notion}
\knowledge{\omega}{notion}
\knowledge{\omegaop}{notion}
\knowledge{\zeta}{notion}
\knowledge{\eta}{notion}
\knowledge{countable}[Countable]{notion}

\knowledge{\suborder}[induced subordering|induced suborderings|induced suborder|induced suborders|
                      subordering|suborderings|suborder|suborders]{notion}
\knowledge{convex}[convexes]{notion}

\knowledge{sum}[\sum|+]{notion}
\knowledge{dense in}[dense in itself]{notion}
\knowledge{dense}{notion}
\knowledge{scattered}{notion}

\knowledge{condensation}[condensations|Condensations|\condensation]{notion}
\knowledge{quotient}[quotients|\quotient|condensed ordering|condensed orderings]{notion}

\knowledge{word}[words|labelled linear ordering|labelled linear orderings]{notion}
\knowledge{empty word}[non-empty word|non-empty words|\emptystr]{notion}
\knowledge{\dom}[domain|domains]{notion}
\knowledge{\countable}{notion}
\knowledge{\scountable}{notion}
\knowledge{subword}[subwords|Subwords]{notion}
\knowledge{factor}[factors|Factors]{notion}

\knowledge{$\eta$-shuffle of letters}[$\eta$-shuffles]{notion}
\knowledge{concatenation}[concatenations|Concatenation|Concatenations|\prod]{notion}
\knowledge{$\omega$-power}[\omegapow]{notion}
\knowledge{$\omegaop$-power}[\omegaoppow]{notion}
\knowledge{$\eta$-shuffle}[\etapow]{notion}

\knowledge{$\countable$-language}[$\countable$-languages|language|languages|Languages]{notion}
\knowledge{$\scountable$-language}[$\scountable$-languages]{notion}

\knowledge{product}[products|generalized product|generalized products|(generalized) product|\pi]{notion}
\knowledge{associativity}[generalized associativity|associative]{notion}
\knowledge{semigroup}[semigroups|$\scountable$-semigroup|$\scountable$-semigroups]{notion}
\knowledge{monoid}[monoids|$\countable$-monoid|$\countable$-monoids]{notion}
\knowledge{free}[(A^\scountable,\prod)|(S^\scountable,\prod)|
                 (A^\countable,\prod)|(S^\countable,\prod)]{notion}
\knowledge{neutral element}[neutral elements|1]{notion}
\knowledge{morphism}[morphisms|morphism of $\countable$-monoids]{notion}
\knowledge{recognizable}[recognizing|recognizes|recognized|
                         recognizable by a $\scountable$-semigroup|
                         recognizable by a finite $\scountable$-semigroup|
                         recognizable by $\scountable$-semigroups|
                         recognizable by finite $\scountable$-semigroups|
                         recognizability|Recognizability|
                         recognizability by $\scountable$-semigroups|
                         recognizability by finite $\scountable$-semigroups|
                         Recognizability by $\scountable$-semigroups|
                         Recognizability by finite $\scountable$-semigroups]{notion}
\knowledge{recognizable by a $\countable$-monoid}          [recognized by the $\countable$-monoid|
           recognizable by $\countable$-monoids|
           recognizable by finite $\countable$-monoids|
           recognizability by $\countable$-monoids|
           recognizability by finite $\countable$-monoids|
           Recognizability by $\countable$-monoids|
           Recognizability by finite $\countable$-monoids]{notion}
\knowledge{recognized by a $\countable$-algebra}[recognized by finite $\countable$-algebras]{notion}

\knowledge{binary product}[\cdot]{notion}
\knowledge{$\tau$-iteration}[\tau]{notion}
\knowledge{$\tauop$-iteration}[\tauop]{notion}
\knowledge{$\kappa$-iteration}[\kappa]{notion}
\knowledge{$\scountable$-algebra}[$\scountable$-algebras|
                                  axiom|axioms|axioms of $\scountable$-algebras|axioms of $\countable$-algebras]{notion}
\knowledge{algebra}[algebras|$\countable$-algebra|$\countable$-algebras]{notion}
\knowledge{induced}[induced by $\pi$|inducing|induces]{notion}
\knowledge{compatible to the right}{notion}
\knowledge{compatible to the left}{notion}
\knowledge{compatible with shuffles}{notion}

\knowledge{condensation tree}[condensation trees|Condensation trees]{notion}
\knowledge{node}[nodes|Nodes]{notion}
\knowledge{root}[roots|Roots]{notion}
\knowledge{leaf}[leaves|Leaves]{notion}
\knowledge{internal node}[internal nodes|Internal nodes]{notion}
\knowledge{descendant}[descendants]{notion}
\knowledge{ancestor}[ancestors]{notion}
\knowledge{child}[children|childhood|Children|\children]{notion}
\knowledge{parent}[parents|Parents]{notion}
\knowledge{proper descendant}[proper descendants|Proper descendants]{notion}
\knowledge{subtree}[subtrees]{notion}

\knowledge{\pis}{notion}
\knowledge{idempotent}[idempotents|Idempotents]{notion}
\knowledge{evaluation tree}[evaluation trees|Evaluation trees]{notion}
\knowledge{value}[values|Values|induced value|induced values|\pit]{notion}

\knowledge{generated}{notion}

\knowledge{additive labelling}[additive labellings|Additive labelling|multiplicative labelling]{notion}
\knowledge{definable}[definable convex|definable convexes|
                      definable convex subset|definable convex subsets]{notion}
\knowledge{strongly definable}[strongly definable convex|strongly definable convexes|
                               strongly definable convex subset|strongly definable convex subsets]{notion}
\knowledge{limit condensation}{notion}

\knowledge{rank}[height|\rank]{notion}
\knowledge{extend the partial function $\pis$}[\epis]{notion}
\knowledge{\subtree}[generalized subtree|generalized subtrees]{notion}

\knowledge{projection}[projections]{notion}
\knowledge{$\forall$-fragment}{notion}
\knowledge{$\exists$-fragment}{notion}
\knowledge{$\forall\exists$-fragment}[$\exists\forall$-fragment]{notion}
\knowledge{split}[splits]{notion}
\knowledge{($k$-)neighbours}[$k$-neighbours|$k$-neighbour|neighbours|neighbour|neighbouring|neighbourhood]{notion}
\knowledge{Ramseian}{notion}
\knowledge{Ramseian pair}[Ramseian pairs]{notion}
\knowledge{\additive}{notion}
\knowledge{first-order definable}[definable by a first-order formula]{notion}
\knowledge{\eval}[\eval_a]{notion}
\knowledge{\level}[level]{notion}
\knowledge{\eval^k}[\eval^{k-1}]{notion}
\knowledge{u[I,\sim]}{notion}
\knowledge{\consistency}{notion}
\knowledge{\varphival}{notion}

\knowledge{Dedekind cut}[Dedekind cuts|(Dedekind) cut|cut|cuts|Cuts]{notion}
\knowledge{extremal}[extremal cut|extremal cuts]{notion}
\knowledge{natural}[natural cut|natural cuts]{notion}
\knowledge{completion}[completions|\completion]{notion}
\knowledge{word completion}[word completions|\wcompletion]{notion}
\knowledge{\cutlabel}{notion}
\knowledge{variant of the product}[product variant|\eprod]{notion}
\knowledge{MSO definable with the cuts at the background}[with the cuts at the background|
                              MSO logic with the cuts at the background]{notion}
\knowledge{quantifier rank of $\varphi$}[quantifier rank|quantifier ranks]{notion}
\knowledge{$k$-type}[\type|$k$-types|logical type|logical types]{notion}
\knowledge{$\cA$-type}[\atype|$\cA$-types]{notion}
\knowledge{\hattype}{notion}
\knowledge{\word}{notion}

\knowledge{tree}[trees|labelled binary tree|labelled binary trees]{notion}
\knowledge{yield}[yields|\yield]{notion}
\knowledge{\invyield}{notion}
\knowledge{yield-invariant}[yield-invariance|yield-invariant language|yield-invariant languages]{notion}
\knowledge{\placeholder}{notion}
\knowledge{hole}{notion}
\knowledge{inhabits}[inhabit|inhabiting|inhabited|inhabitant|inhabitants]{notion}
\knowledge{regular inhabitant}[regular inhabitants]{notion}
\knowledge{\congT}{notion}
\knowledge{congruence}{notion}
\knowledge{parity tree automaton}[parity tree automata|tree automaton|tree automata|automaton|automata]{notion}
\knowledge{successful run of $\cA$ on a tree $t$}[successful run|successful runs|run|runs]{notion}
\knowledge{\lang}[recognized by an automaton]{notion}
\knowledge{emptiness problem}[emptiness problems|emptiness]{notion}
\knowledge{containment problem}[containment problems|containment]{notion}
\knowledge{equivalence problem}[equivalence problems|equivalence]{notion}
\knowledge{membership problem}[membership problems|membership]{notion}
\knowledge{regular tree}[regular trees|regular]{notion}

\begin{document}

\ifarxiv
\title{\bfseries An algebraic approach to MSO-definability\\ on countable linear orderings}
\author{Olivier Carton\\ IRIF, Université Paris Diderot \and
        Thomas Colcombet\\ CNRS -- IRIF, Université Paris Diderot \and
        Gabriele Puppis\\ CNRS -- LaBRI, Université Bordeaux}
\date{}
\fi

\maketitle

\begin{abstract}
We develop an algebraic notion of recognizability for languages of words indexed 
by countable linear orderings. We prove that this notion is effectively equivalent 
to definability in monadic second-order (MSO) logic.
We also provide three logical applications. First, we establish the first 
known collapse result for the quantifier alternation of MSO logic over countable 
linear orderings. Second, we solve an open problem posed by Gurevich and Rabinovich,
concerning the MSO-definability of sets of rational numbers using the reals in the
background. Third, we establish the MSO-definability of the set of yields induced by
an MSO-definable set of trees, confirming a conjecture posed by Bruy\`ere, Carton,
and S\'enizergues.
 \end{abstract}

\dottedcontents{section}[2.3em]{}{2.3em}{5pt}
\dottedcontents{subsection}[5.5em]{}{3.2em}{5pt}
\tableofcontents

\sloppy

\let\oldomega\omega
\let\oldzeta\zeta
\let\oldeta\eta
\def\rev{{"*@\rev"}}
\def\omega{{"\oldomega@\omega"}}
\def\zeta{{"\oldzeta@\zeta"}}
\def\eta{{"\oldeta@\eta"}}
\def\omegaop{{"\oldomega\mspace{-1mu}^\rev\mspace{-1mu}@\omegaop"}}

\newcommand\suborder[1]{{\mathop{"|_{#1}@\suborder"}}}
\newcommand\quotient[1]{{\mathop{"/_{#1}@\quotient"}}}

\let\oldemptystr\emptystr
\def\emptystr{{"\oldemptystr@\emptystr"}}

\newcommand{\dom}{"\mathsf{dom}@\dom"}
\newcommand{\countable}{{"\ostar@\countable"}}
\newcommand{\scountable}{{"\oplus@\scountable"}}

\newcommand{\omegapow}{"^\omega@\omegapow"}
\newcommand{\omegaoppow}{"^\omegaop@\omegaoppow"}
\newcommand{\etapow}{"^\eta@\etapow"}

\let\oldcdot\cdot
\let\oldtau\tau
\let\oldkappa\kappa
\def\cdot{\mathbin{"\oldcdot@\cdot"}}
\def\tau{{"\oldtau@\tau"}}
\def\tauop{{"\oldtau\mspace{-1mu}^\rev\mspace{-1mu}@\tauop"}}
\def\kappa{{"\oldkappa@\kappa"}}

\newcommand{\pis}{"\pi_0@\pis"}
\newcommand{\epis}{"\pi_0@\epis"}
\newcommand{\pit}{"\gamma@\pit"}

\newcommand{\range}{"\mathsf{range}"}

\let\oldsum\sum
\let\oldprod\prod
\def\sum{\mathop{"\oldsum@\sum"}}
\def\prod{\mathop{"\oldprod@\prod"}}
\def\eprod{\mathop{"\reallywidehat\prod@\eprod"}}

\newcommand{\additive}[1]{{"\sigma_{#1}@\additive"}}

\newcommand{\children}{"\mathsf{children}@\children"}
\renewcommand{\min}{{\mathsf{min}}}
\renewcommand{\max}{{\mathsf{max}}}
\newcommand{\rank}{{"\mathsf{rank}@\rank"}}
\newcommand{\level}{{"\mathsf{level}@\level"}}
\newcommand{\varphival}[1]{{"\varphi^{\mathsf{value}}_{#1}@\varphival"}}
\newcommand{\hatvarphival}[1]{{"\hat\varphi^{\mathsf{value}}_{#1}@\varphival"}}
\newcommand{\eval}{{"\mathsf{eval}@\eval"}}
\newcommand{\consistency}{{"\mathsf{consistency}@\consistency"}}
\newcommand{\cover}{{"\mathsf{cover}@\cover"}}
\newcommand{\type}[1]{{"\mathsf{type}_{#1}@\type"}}
\newcommand{\atype}[1]{{"\mathsf{type}_{#1}@\atype"}}
\newcommand{\hattype}[1]{{"\mathsf{type}^{\wedge}_{#1}@\hattype"}}
\newcommand{\word}{{"\mathsf{word}@\word"}}
\newcommand{\yield}{{"\mathsf{yield}@\yield"}}
\newcommand{\invyield}{{"\mathsf{yield}^{-1}@\invyield"}}

\newcommand\subtree[1]{{\mathop{"|_{#1}@\subtree"}}}

\def\cutlabel{{"c@\cutlabel"}}
\def\placeholder{{"c@\placeholder"}}

\newcommand\completion[1]{{"\hat{#1}@\completion"}}
\newcommand\wcompletion[1]{{"\hat{#1}@\wcompletion"}}

\newcommand\lang{{"\sL@\lang"}}
\newcommand\congT{\mathbin{"\cong_T@\congT"}}

\newcommand{\overdashed}[1]{\mathchoice                            {\overset{\smash{-~-~-}}{\raisebox{0pt}[\height-1pt][0pt]{                               \ensuremath{\displaystyle#1}}}}                            {\overset{\smash{-~-~-}}{\raisebox{0pt}[\height-1pt][0pt]{                               \ensuremath{#1}}}}                            {\overset{\smash{-\;-\;-}}{\raisebox{0pt}[\height-0.75pt][0pt]{                               \ensuremath{\scriptstyle#1}}}}                            {\overset{\smash{-\:-\:-}}{\raisebox{0pt}[\height-0.5pt][0pt]{                               \ensuremath{\scriptscriptstyle#1}}}}}

\newcommand{\reallywidehat}[1]{"\stackrel{\raisebox{-2pt}{\resizebox{\widthof{\ensuremath{#1}}}                                                                    {\heightof{\ensuremath{\wedge}}}                                                                    {\ensuremath{\wedge}}}}                                         {\ensuremath{#1}}@word completion"}

\section{Introduction}\label{sec:introduction}

\AP
The paper continues a long line of research aiming at understanding the
notions of regularity for languages of infinite objects, e.g., infinite
words and trees. 
The central objects in this paper are \emph{"words"} indexed by countable 
"linear orderings", i.e., total orders over finite or "countable" sets 
paired with functions mapping elements to letters in some finite alphabet.
Accordingly, "languages" here are just sets of "countable" "words".
We use \emph{"monadic second-order" ("MSO") logic} as a 
formalism for describing such "languages". In particular, an "MSO" 
formula may involve quantifications over positions of a word, as
well as quantifications over sets of positions. A sentence naturally 
defines the "language" of all "words" that make the sentence true.

\AP
This paper provides a fine comprehension of the expressive power
of "MSO" logic over "countable" "linear orderings" by proving
a correspondence between definability in "MSO" and "recognizability"
by suitable algebraic structures.
More precisely, we introduce a generalization of the classical notion of  
finite monoid (i.e., a finite set equipped with an associative product), 
that we call "$\countable$-monoid", and we extend accordingly the notion of
"recognizability" by monoid "morphism" to capture a large class of 
"languages" of "countable" "words". Differently from the classical 
setting, "$\countable$-monoids" are not finite objects, as the product 
mapping is defined over "countable" sequences of elements
and a priori it is not clear how to represent this mapping by a 
finite table.
To obtain finite presentations of the "recognized@recognizability" 
"languages", we follow an approach similar to \cite{algebraic_theory_for_regular_languages},
namely, we associate with each "$\countable$-monoid" a finite number of
operators with finite domain. We prove that, under natural conditions,
the associated algebraic structure, called "$\countable$-algebra", 
uniquely determines a "$\countable$-monoid".
The correspondence between "$\countable$-monoids" and "$\countable$-algebras",
together with the proposed notion of "recognizability", gives a natural 
framework where "languages" of "countable" "words" can be represented and 
manipulated algorithmically.
Our main contribution consists in proving that "recognizability" by 
"$\countable$-monoids"/"algebras@$\countable$-algebras" corresponds 
effectively to definability in "MSO" logic, exactly as it happens
for regular languages of finite words and $\omega$-words:

\begin{quote}
\AP
\em 
The "languages" "recognized@recognizability" by $\countable$-monoids
are the same as the "languages" definable in "MSO" logic.
\end{quote}

\AP
Prior results (see related work below) also focused on "MSO" logic
over "countable" "linear orderings" and similar correspondences with algebraic structures, 
but mostly from the point of view of decidability of the logical theory. 
Our study gives a deeper insight on the expressive power of "MSO" logic
on these structures.  
For example, as a by-product of our results we obtain that the quantifier
hierarchy of "MSO" logic collapses to its second level:

\begin{quote}
\AP
\em 
Every "language" of "countable" "words" defined in "MSO" logic 
can be equally defined in the "$\exists\forall$-fragment".
\end{quote}

\noindent
The above result is reminiscent of the collapse of "MSO" 
to its existential fragment when interpreted over $\omega$, 
as shown by B\"uchi in \cite{s1s}. We also show that our collapse
result is optimal, in the sense that the first level of the quantifier
hierarchy does not capture the full expressive power of "MSO" logic
on "countable" "linear orderings". 
This situation is also very similar to the setting of regular languages 
of infinite trees, where a collapse of "MSO" at the second level holds \cite{s2s}.
Despite this similarity and the fact that "recognizable" "languages"
of "countable" "words" are "MSO"-interpretable from regular languages 
of infinite trees, our collapse result does not follow immediately from 
Rabin's result. Indeed, an "MSO"-interpretation may exploit second-order 
quantifications to define "linear orderings" inside infinite trees. 

\AP
Our investigation on "recognizability" by "$\countable$-monoids" 
provides also new insights on the type of properties that can be 
expressed in "MSO" logic over {\sl "uncountable@countable"} linear orderings. 
For example, we consider the following question that was raised and left open 
by Gurevich and Rabinovich in~\cite{definability_with_reals_in_the_background}: 

\begin{quote}
\AP
\em Given a property for sets of rational numbers that is "MSO"-definable in the real line, 
is it possible to define it directly in the rational line? In other words, is it true 
that the presence of reals `at the background' does not increase the expressive power 
of "MSO" logic?
\end{quote}

\noindent
We answer positively the above question by building up on the correspondence 
between "MSO"-definability and "recognizability" by "$\countable$-monoids". 
The latter expressiveness result is inherently non-effective since the "MSO" 
theory of the real line is undecidable \cite{composition_method_shelah}, while 
that of the rational line is decidable.

\AP
Finally, we establish an interesting correspondence between "MSO"-definability 
of languages of (possibly infinite) trees and "MSO"-definability of their "yields":

\begin{quote}
\AP
\em Define the "yield" of a "tree" as the set of "leaves" ordered by the infix relation.
Consider an "MSO"-definable "tree" language $L$ that is "yield-invariant", namely, such 
that for all "trees" $t,t'$ with the same "yield", $t\in L$ iff $t'\in L$. 
The set of "yields" of "trees" in $L$ is effectively "MSO"-definable.
\end{quote}

\noindent
In \cite{yields_of_regular_tree_languages} a similar result was shown
in the restricted setting of {\sl finite} "trees".

\paragraph{Related work}
\AP
B\"uchi initiated the study of "MSO" logic using the tools of language theory. 
He established that every language of $\omega$-words (i.e., the particular 
case of words indexed by the ordinal~$"\omega"$) definable in "MSO" logic
is effectively recognized by a suitable form of automaton \cite{s1s}.
A major advance was obtained by Rabin, who extended this result to 
infinite trees \cite{s2s}. One consequence of Rabin's result is that 
"MSO" logic is decidable over the class of all "countable" "linear orderings". 
Indeed, every "linear ordering" can be seen as a set of nodes of the infinite 
tree, with the order corresponding to the infix ordering on nodes.
Another proof of the decidability of the "MSO" theory of "countable" "linear orderings" 
has been given by Shelah using the composition method \cite{composition_method_shelah}. 
This automaton-free approach to logic is based on syntactic operations 
on formulas and is inspired from Feferman and Vaught \cite{feferman_vaught_theorem}. 
The same paper of Shelah is also important for another result it contains: 
the undecidability of the "MSO" theory of the real line (the reals with order).
However, for infinite words as for infinite trees, the theory is much
richer than simply the decidability of "MSO" logic. In particular, "MSO" logic
is known to be equivalent to a number of different formalisms, such as automata,
some forms of algebras, and, in the $\omega$-word case, regular expressions.
"MSO" logic is also known to collapse to its existential fragment when
interpreted on the linear order $"\omega"$, that is, every formula is 
equivalent to a formula consisting of a block of existential quantifiers 
followed by a first-order formula.

\AP
Another branch of research has been pursued to raise the equivalence
between logic, automata, and algebra to infinite words beyond
$\omega$-words.  In \cite{Buchi64}, B\"uchi introduced $\omega_1$-automata
on transfinite words to prove the decidability of "MSO" logic for ordinals 
less than $\omega_1$.  Besides the usual transitions, $\omega_1$-automata 
are equipped with limit transitions of the form $P\rightarrow q$, with 
$P$ set of states, which are used in a Muller-like way to process words 
indexed over ordinals.  
B\"uchi proved that these automata have the same expressive power as 
"MSO" logic over ordinals less than $\omega_1$. The key ingredient 
is the closure under complementation of $\omega_1$-automata.
In \cite{CartonBruyere07}, $\omega_1$-automata have been extended to
$\diamond$-automata by introducing limit transitions of the form $q \rightarrow P$
to process words over "linear orderings".  In \cite{CartonRispal05}, $\diamond$-automata 
are proven to be closed under complementation with respect to "countable" and 
"scattered" "linear orderings" (a "linear ordering" is "scattered" if it is 
nowhere "dense", namely, if none of its "suborders" is isomorphic to the rational line).
More precisely, $\diamond$-automata have the same expressive power as "MSO" logic
over "countable" and "scattered" "linear orderings" \cite{BedonBesCartonRispal10}.  
However, it was already noticed in \cite{BedonBesCartonRispal10} that $\diamond$-automata 
are strictly weaker than "MSO" logic over "countable" (possibly non-"scattered") 
"linear orderings": indeed, the closure under complementation fails as there is 
an automaton that accepts all words with non-"scattered" domains, whereas there 
is none for "scattered" words.

\AP
Some of the results presented here appeared in preliminary form in 
the conference papers \cite{mso_over_countable_orderings} and 
\cite{MSO_with_cuts_in_the_background}.

\paragraph{Structure of the paper}
\AP After the preliminaries in Section~\ref{sec:preliminaries}, we introduce 
in Section~\ref{sec:semigroups} the notions of "$\countable$-monoids" 
and "$\countable$-algebras", and present the corresponding tools and results. 
In Section~\ref{sec:logic-to-algebra} we translate "MSO" formulas to "$\countable$-algebras"
and in Section~\ref{sec:algebra-to-logic} we establish the converse. 
In Section~\ref{sec:applications} we exploit the developed algebraic framework
to solve three open problems that we discussed earlier, namely: 
(i) the collapse of the quantifier hierarchy of "MSO" logic, 
(ii) the correspondence between classical "MSO"-definability and 
definability with the reals `at the background', and 
(iii) the "MSO"-definability of the set of yields induced 
by a regular yield-invariant tree language.

\paragraph{Acknowledgements}
\AP
We are grateful to Achim Blumensath for his numerous comments on this work
and to Alexander Rabinovich for the discussions on the subject and for 
introducing us to the question of definability with the reals at the background.
 
\smallskip
\section{Preliminaries}\label{sec:preliminaries}

In this section we recall some definitions for 
"linear orderings", "condensations", "words", and "languages".

\medskip
\subsection{Linear orderings}\label{subsec:linear-orderings}

A ""linear ordering"" $\alpha=(X,<)$ is a set $X$ equipped with 
a total order $<$. 
By a slight abuse of terminology, we call a "linear ordering" 
""countable"" when its domain is finite or countable.
We write $\alpha^{""\rev""}$ to denote the reverse linear ordering $(X,>)$. 
Two "linear orderings" have same ""order type"" if there is an "order"-preserving 
bijection between their domains. 
We denote by $""\omega""$, $""\omegaop""$, $""\zeta""$, $""\eta""$
the "order types" of $(\bbN,<)$, $(-\bbN,<)$, $(\bbZ,<)$, $(\bbQ,<)$, respectively. 
Unless strictly necessary, we do not distinguish between a "linear ordering" and 
its "order type". 

\AP
Given a subset $I$ of a "linear ordering" $\alpha$, we denote by $\alpha\suborder{I}$ 
the ""induced subordering"". Given two subsets $I,J$ of $\alpha$, we 
write $I<J$ iff $x<y$ for all $x\in I$ and all $y\in J$. A subset $I$ of 
$\alpha$ is said to be ""convex"" if for all $x,y\in I$ and all $z\in\alpha$, 
$x<z<y$ implies $z\in I$. 

\AP
The ""sum"" $\alpha_1 + \alpha_2$ of two "linear orderings"
$\alpha_1=(X_1,<_1)$ and $\alpha_2=(X_2,<_2)$ (up to renaming, assume that 
$X_1$ and $X_2$ are disjoint) is the "linear ordering" $(X_1\uplus X_2,<)$, 
where $<$ coincides with $<_1$ on $X_1$, with $<_2$ on $X_2$, and, furthermore, it 
satisfies $X_1<X_2$. More generally, given a "linear ordering" $\alpha=(X,<)$ and, 
for each $i\in X$, a "linear ordering" $\beta_i=(Y_i,<_i)$ (assume that the sets $Y_i$ 
are pairwise disjoint), we define the "sum" $\sum_{i\in\alpha}\beta_i$ to be the 
"linear ordering" $(Y,<')$, where $Y=\biguplus_{i\in X}Y_i$ and, for every $i,j\in X$, every 
$x\in Y_i$, and every $y\in Y_j$, $x<'y$ iff either $i=j$ and $x<_i y$ hold or $i<j$ holds.

\AP
A subset $I$ of a "linear ordering" $\alpha$ is ""dense in"" $\alpha$
if for every $x<y\in\alpha$, there exists $z\in I$ such that $x<z<y$. 
For example, $(\bbQ,<)$ is "dense in" $(\bbR,<)$ and $(\bbR,<)$ is "dense in itself". 
If a "linear ordering" $\alpha$ is "dense in itself", then 
we simply say that $\alpha$ is ""dense"". 
A "linear ordering" $\alpha$ is ""scattered"" 
if all its "dense" "suborderings" are empty or singletons. 
For example, $(\bbN,<)$, $(\bbZ,<)$, and all the ordinals are "scattered". 
Being "scattered" is preserved under taking a "subordering". 
A "scattered" "sum" of "scattered" "linear orderings" 
also yields a "scattered" "linear ordering". 

\AP
Additional material on "linear orderings" can be found in \cite{linear_orderings}.
 
\medskip
\subsection{Condensations}\label{subsec:condensations}

A standard way to prove properties of "linear orderings" is to decompose them into basic objects
(e.g., finite sequences, $\omega$-sequences, $\omegaop$-sequences, and $\eta$-orderings).
This can be done by exploiting the notion of "condensation". 

\AP
Precisely, a ""condensation"" of a "linear ordering" $\alpha$ is an equivalence relation $\sim$ 
over $\alpha$ such that for all $x<y<z$, $x\sim z$ implies $x\sim y\sim z$. Equivalently, a 
"condensation" of $\alpha$ can be seen as a partition of $\alpha$ into "convex" subsets. 

\AP
The order on $\alpha$ induces a corresponding "order" on the 
quotient $\alpha\quotient{\sim}$, which is called the ""condensed ordering"". 
This "condensed ordering" $\alpha\quotient{\sim}$ inherits some properties 
from $\alpha$: if $\alpha$ is "countable" (resp., "scattered"), then 
$\alpha\quotient{\sim}$ is "countable" (resp., "scattered").

\medskip
\subsection{Words and languages}\label{subsec:words}

We use a generalized notion of "word", which coincides 
with the notion of labelled "linear ordering".
Given a "linear ordering" $\alpha$ and a finite alphabet $A$, 
a ""word over $A$ with domain $\alpha$@word"" 
is a mapping of the form $w:\alpha\then A$. 
The ""domain"" of a word $w$ is denoted $\dom(w)$.
Unless specifically required, we shall always consider "words"
of "countable" "domain", and up to isomorphism. 
The set of all "words" (of "countable" "domain") over an alphabet~$A$ 
is denoted $A^{""\countable""}$.
The set of all "words" of non-empty ("countable") "domain" over an alphabet $A$ 
is denoted $A^{""\scountable""}$.
Given a "word" $w$ and a subset $I$ of $\dom(w)$, 
we denote by $w\suborder{I}$ the ""subword"" resulting 
from restricting the domain of $w$ to $I$. 
If in addition $I$ is "convex", then $w\suborder{I}$ is said to be a ""factor"" of~$w$.

\AP
Certain "words" will play a crucial role in the sequel,
so we introduce specific notation for them.
For example, we denote the ""empty word"" by $\emptystr$.
A "word" $w$ is said to be 
an ""$\eta$-shuffle@$\eta$-shuffle of letters"" of set $A$ of letters 
if (i) the domain $\dom(w)$ has "order type" $\eta$ and 
(ii) for every symbol $a\in A$, the set 
$w^{-1}(a)=\{x\in\dom(w) \mid w(x)=a\}$ is "dense in" $\dom(w)$.
Recall that $"\eta"$ is -- up to isomorphism -- the unique "countable" 
"dense" linear ordering with no end-points.
Likewise, for every finite set $A$, there is a unique, up to isomorphism, 
"$\eta$-shuffle@$\eta$-shuffle of letters" of $A$.

\AP
Given two "words" $u,v$, we denote by $u v$ the ""concatenation"" 
of $u$ and $v$, namely, the "word" with domain $\dom(u) + \dom(v)$, 
where each position $x\in\dom(u)$ (resp., $x\in\dom(v)$) 
is labelled by $u(x)$ (resp., $v(x)$). 
This is readily generalized to infinite "concatenations" of the form 
$\prod_{i\in\alpha}w_i$, for any "linear ordering" $\alpha$ 
and any sequences of words $(w_i)_{i\in\alpha}$, the resulting word having 
"domain" $\sum_{i\in\alpha} \dom(w_i)$.
The ""$\omega$-power"" of a word $w$ is defined as $w\omegapow=\prod_{i\in\omega}w$. 
Similarly, we define the ""$\omegaop$-power"" $w\omegaoppow=\prod_{i\in\omegaop}w$. 
By a slight abuse of terminology, we also define the ""$\eta$-shuffle"" of a tuple 
of words $w_1,\ldots,w_k$ as the word
$$
  \{w_1,\ldots,w_k\}\etapow ~\eqdef~ \prod_{i\in\eta}w_{f(i)}
$$
where $f$ is the unique "$\eta$-shuffle@$\eta$-shuffle of letters" 
of the set of letters $I=\{1,\ldots,k\}$. 

\AP
A ""$\countable$-language"" (resp., ""$\scountable$-language"") is any 
set of "words" (resp., non-empty words) over a fixed finite alphabet. 
The operations of "concatenation", "$\omega$-power", "$\omegaop$-power", 
"$\eta$-shuffle", etc.~are extended to "languages" in the obvious way.
  
\smallskip
\section{Algebras for "countable" "words"}\label{sec:semigroups}

\AP
In this section we present the algebraic objects that are suited for 
deriving a notion of "recognizability" for "languages" of "countable" "words".
As it was already the case for "words" with domain $\omega$,
\cite{infinite_words,algebraic_theory_for_regular_languages}, 
our definitions come in two flavors, "$\countable$-monoids" (corresponding to $\omega$-monoids)
and "$\countable$-algebras" (corresponding to Wilke's algebras). 
We prove the equivalence of the two notions when the supports are finite.

\medskip
\subsection{"Countable" products}\label{subsec:countableproducts}

We introduce below a notion of product indexed by "countable"
"linear orderings" that satisfies a "generalized associativity" 
property. 

\begin{definition}\label{def:semigroup}
A ""(generalized) product"" over a set~$S$ is a function~$\pi$ from~$S^\scountable$
to~$S$ such that, for every $a\in S$, $\pi(a)=a$ and, for every family of 
words~$(u_i)_{i\in\alpha}\in (S^\scountable)^\scountable$,
\ifarxiv
\phantomintro{generalized associativity}
\begin{align*}
  \pi\Big( \prod\nolimits_{i\in\alpha}\pi(u_i) \Big)
  ~&=~
  \pi\Big( \prod\nolimits_{i\in\alpha}u_i \Big)\ . 
  \tag{generalized associativity}
\end{align*}
\fi
The pair $(S,\pi)$ is called a ""$\scountable$-semigroup"".

\noindent
If the same definition holds, with $\pi$ function from $S^\countable$ to $S$
and $(u_i)_{i\in\alpha}\in (S^\countable)^\countable$, 
then $(S,\pi)$ is called a ""$\countable$-monoid"".
\end{definition}

\AP
As an example, the function~$\prod$ that maps any "countable" sequence 
of "non-empty words" to their "concatenation" is a "generalized product" 
over $A^\scountable$.
Hence, $(A^\scountable,\prod)$ is a "$\scountable$-semigroup"; 
it is indeed the ""free"" $\scountable$-semigroup generated by $A$.
Similarly, $(A^\countable,"\prod")$ is the "free" $\countable$-monoid generated by $A$.

\AP
Given a "$\scountable$-semigroup" $(S,\pi)$, we call 
""neutral element"" an element $1\in S$ such that, for every word $w\in S^\scountable$, 
if $w|_{\neq1}$ is the "subword" of $w$ obtained by removing every occurrence of the 
element $1$ and $w|_{\neq1}$ is non-empty, then $\pi(w)=\pi(w|_{\neq1})$. Note that the
"neutral element", if exists, is unique: given two "neutral elements" $1,1'\in S$, we 
have $1 = \pi(1) = \pi(11'|_{\neq1'}) = \pi(11') = \pi(11'|_{\neq 1}) = \pi(1') = 1'$. 

\AP
At some places in the proofs it will be necessary to use "$\scountable$-semigroups" 
rather than "$\countable$-monoids". The two notions are however very close. On the
one hand, any "$\countable$-monoid" $(S,\pi)$ can be seen as a "$\scountable$-semigroup"
by simply restricting its "generalized product" $\pi$ to $S^\scountable$. On the other
hand, any "$\scountable$-semigroup" $(S,\pi)$ can be extended to a "$\countable$-monoid"
either by letting $\pi(\emptystr)=1$, where $\emptystr$ is the "empty word" and 
$1$ is the (unique) "neutral element" of $(S,\pi)$, or, if $(S,\pi)$ has no "neutral element", 
by introducing a fresh element $1\nin S$ and by letting $\pi(\emptystr)=1$ and 
$\pi(w)=\pi(w|_{\neq1})$ for all words $w$ over $S\uplus\{1\}$.

\medskip
\AP
A ""morphism"" from a $\scountable$-semigroup $(S,\pi)$ to 
another $\scountable$-semigroup $(S',\pi')$ is a mapping $h:S\then S'$ such that, 
for every word $(w_i)_{i\in\alpha} \in S^\scountable$,
\begin{align*}
  h(\pi(w))=\pi'(\bar h(w))\ ,
\end{align*}
where $\bar h$ is the pointwise extension of~$h$ to "words".
A \emph{morphism of $\countable$-monoids} is defined similarly,
this time with $(w_i)_{i\in\alpha} \in S^\countable$.

\AP
A "$\scountable$-language" $L\subseteq A^\scountable$ is 
""recognizable by a $\scountable$-semigroup"" if there 
exists a "morphism" $h$ from $"(A^\scountable,\prod)"$
to some finite "$\scountable$-semigroup"~$(S,\pi)$ 
(here finite means that~$S$ is finite) 
such that~$L=h^{-1}(F)$ for some~$F\subseteq S$
(equivalently, $h^{-1}(h(L))=L$).
Similarly, a "$\countable$-language" $L\subseteq A^\countable$
is ""recognizable by a $\countable$-monoid""
if there exists a "morphism" $h$ from $"(A^\countable,\prod)"$
to some finite "$\countable$-monoid"~$(M,\pi)$ 
(here finite means that~$M$ is finite)
such that~$L=h^{-1}(F)$ for some~$F\subseteq M$
(equivalently, $h^{-1}(h(L))=L$).

\AP
We are mainly interested in languages "recognizable by finite $\countable$-monoids".
However, it is worth noticing that, with respect to membership of non-empty "words", 
this notion is the same as "recognizability by finite $\scountable$-semigroups": 
indeed, a language $L$ is 
"recognizable by finite $\countable$-monoids" iff 
$L\setminus\{\emptystr\}$ is "recognizable by finite $\scountable$-semigroups".
 
\medskip
\subsection{From "countable" products to algebras}\label{subsec:algebras-from-semigroups}

The notion of "recognizability" for $\countable$-languages makes use of a "product" 
function $\pi$ that needs to be represented, a priori, by an infinite table. 
This is a not usable as it stands for finite presentations of languages,
nor for decision procedures. 
That is why, given a finite "$\scountable$-semigroup"
$(S,\pi)$, we define the following (finitely presentable) algebraic operators:
\begin{itemize}
  \item the ""binary product"" 
        $\cdot: S^2\then S$, mapping any pair of elements $a,b\in S$ to the element $\pi(ab)$,
  \item the ""$\tau$-iteration"" $\tau: S\then S$, mapping any element $a\in S$ to the element $\pi(a\omegapow)$
        (thus, $\tau$ is the analogous of the "$\omega$-power" inside $S$),
  \item the ""$\tauop$-iteration"" $\tauop: S\then S$, mapping any element $a\in S$ to the element $\pi(a\omegaoppow)$
        (thus, $\tauop$ is the analogous of the "$\omegaop$-power" inside $S$),
  \item the ""$\kappa$-iteration"" $\kappa: \sP(S)\setminus\{\emptyset\}\then S$, 
        mapping any non-empty subset $\{a_1,\ldots,a_k\}$ of $S$
        to the element $\pi\big(\{a_1,\ldots,a_k\}\etapow\big)$
        ($\kappa$ is the analogous of the "$\eta$-shuffle" inside $S$).
\end{itemize}
Furthermore, if $(S,\pi)$ is a "$\countable$-monoid", we see the "neutral element" 
as a nullary operator induced by $\pi$, namely, as $1=\pi(\emptystr)$. 
One says that $\cdot$, $\tau$, $\tauop$, $\kappa$ (and possibly $1$) 
are ""induced by $\pi$"".  
From now on, we shall use the operator $\cdot$ with infix notation (e.g., 
$a\cdot b$) and the operators $\tau$, $\tauop$, and $\kappa$ with 
superscript notation (e.g., $a^\tau$, $\{a_1,\ldots,a_k\}^\kappa$).  
As shown below, the resulting structures $(S,\cdot,\tau,\tauop,\kappa)$ 
$(S,1,\cdot,\tau,\tauop,\kappa)$ have the property of being, respectively,
a "$\scountable$-algebra" and a "$\countable$-algebra".

\begin{definition}\label{def:o-algebra}
A structure $(S,\cdot,\tau,\tauop,\kappa)$, with
$\cdot:S^2\then S$, $\tau,\tauop:S\then S$, and $\kappa:\sP(S)\setminus\{\emptyset\}\then S$,
is called a ""$\scountable$-algebra"" if:
\begin{axiomlist}
  \item $(S,\cdot)$ is a semigroup, namely, for every 
        $a,b,c\in S$, $a\cdot(b\cdot c)=(a\cdot b)\cdot c$,
        \label{axiom:concatenation}
  \item $\tau$ is ""compatible to the right"", namely, for every $a,b\in S$ and every $n>0$, 
        $(a\cdot b)^\tau=a\cdot(b\cdot a)^\tau$ and $(a^n)^\tau=a^\tau$,
        \label{axiom:omega}
  \item $\tauop$ is ""compatible to the left"", namely, for every $a,b\in S$ and every $n>0$, 
        $(b\cdot a)^\tauop=(a\cdot b)^\tauop\cdot a$ and $(a^n)^\tauop=a^\tauop$,
        \label{axiom:omegaop}
  \item $\kappa$ is ""compatible with shuffles"", namely, for every non-empty subset $P$ of $S$,
        every element $c$ in $P$, every subset $P'$ of $P$, and every non-empty subset $P''$ of 
        $\{P^\kappa,\, a\cdot P^\kappa,\, P^\kappa\cdot b,\, a\cdot P^\kappa\cdot b ~\mid~a,b\in P\}$,
        we have 
        $$
        \begin{array}{rl}
          P^\kappa &=~ P^\kappa\cdot P^\kappa ~=~ P^\kappa\cdot c\cdot P^\kappa   \\[1ex]
                   &=~ (P^\kappa)^\tau        ~=~ (P^\kappa\cdot c)^\tau          \\[1ex]
                   &=~ (P^\kappa)^\tauop      ~=~ (c\cdot P^\kappa)^\tauop        \\[1ex]
                   &=~ (P'\cup P'')^\kappa\ .
        \end{array}
        $$
        \label{axiom:shuffle}
\end{axiomlist}
A ""$\countable$-algebra"" $(M,1,\cdot,\tau,\tauop,\kappa)$ 
is a "$\scountable$-algebra" $(M,\cdot,\tau,\tauop,\kappa)$
with a distinguished element $1\in M$ such that 
\begin{axiomlist}   \addtocounter{enumi}{4}
  \item $x\cdot 1=1\cdot x=x$, \ 
        $1^\tau=1^\tauop=\{1\}^\kappa=1$, \ and \ 
        $P^\kappa=(P\cup\{1\})^\kappa$, \ 
        for all $x\in M$ and all non-empty $P\subseteq M$.
        \label{axiom:identity}
\end{axiomlist}
\end{definition}

\AP
The typical "$\scountable$-algebras" and "$\countable$-algebras" are:

\begin{lemma}\label{lemma:word-algebra}
For every alphabet~$A$, $(A^\scountable,\cdot,\omega,\omegaop,\eta)$ is a
"$\scountable$-algebra" and $(A^\countable,\varepsilon,\cdot,\omega,\omegaop,\eta)$ is a 
"$\countable$-algebra"\footnote{Similarly to what happens for Wilke's algebras
                                \cite{algebraic_theory_for_regular_languages}, 
                                $(A^\scountable,\cdot,\omega,\omegaop,\eta)$
                                is not the "free" $\scountable$-algebra generated by~$A$, 
                                as the "free" algebra generated by a finite set is by 
                                definition countable, while $A^\scountable$ has the 
                                cardinality of the continuum.}.
\end{lemma}

\begin{proof}
By a systematic analysis of Axioms~\refaxiom{axiom:concatenation}-\refaxiom{axiom:identity}. 
\end{proof}

\smallskip
\AP
Furthermore, as we mentioned above, every "$\scountable$-semigroup" induces a 
"$\scountable$-algebra" and every "$\countable$-monoid" induces a "$\countable$-algebra":

\begin{lemma}\label{lemma:semigroup-to-algebra}
For every "$\scountable$-semigroup" $(S,\pi)$, $(S,\cdot,\tau,\tauop,\kappa)$ is a 
"$\scountable$-algebra", where the operators $\cdot$, $\tau$, $\tauop$, and $\kappa$ 
are those "induced by $\pi$". Similarly every "$\countable$-monoid" $(S,\pi)$, 
$(S,1,\cdot,\tau,\tauop,\kappa)$ is a "$\countable$-algebra", where the operators 
$\cdot$, $\tau$, $\tauop$, $\kappa$, and $1$ are those "induced by $\pi$".
\end{lemma}

\begin{proof}
The results are simply inherited from Lemma~\ref{lemma:word-algebra} by morphism.
Let~$(S,\pi)$ be a "$\scountable$-semigroup" "inducing" the 
operators~$\cdot,\tau,\tauop,\kappa$.
The structure~$"(S^\scountable,\prod)"$ is also a "$\scountable$-semigroup", 
which induces the operations of "concatenation", "$\omega$-power", "$\omegaop$-power", 
and "$\eta$-shuffle".
Furthermore, the product~$\pi$ can be seen as a surjective morphism from $"(S^\scountable,\prod)"$
to~$(S,\pi)$ (just a morphism of abstract algebras, not of "$\scountable$-algebras").
By definition of~$\cdot,\tau,\tauop,\kappa$, this morphism maps 
"concatenation" to $\cdot$, "$\omega$-power" to "$\tau$-iteration", 
"$\omegaop$-power" to "$\tauop$-iteration", and "$\eta$-shuffle" to "$\kappa$-iteration". 
It follows that any equality involving concatenation, $\omega$-power, $\omegaop$-power, and~$\eta$-shuffle 
is also satisfied by the analogous operations $\cdot$, $\tau$, $\tauop$, and $\kappa$.
In particular, the "axioms" that, thanks to Lemma~\ref{lemma:word-algebra}, are satisfied 
by the "$\scountable$-algebra" $(S^\scountable,\cdot,\omega,\omegaop,\eta)$ are directly
transferred to $(S,\cdot,\tau,\tauop,\kappa)$. The case of a "$\countable$-monoid" is similar.
\end{proof}
 
\medskip
\subsection{From algebras to "countable" products}\label{subsec:semigroups-from-algebras}

\AP
Here, we aim at proving a converse to Lemma~\ref{lemma:semigroup-to-algebra}, namely,
that every {\sl finite} "$\scountable$-algebra" $(S,\cdot,\tau,\tauop,\kappa)$ 
can be uniquely extended to a "$\scountable$-semigroup" $(S,\pi)$, and similarly for
"$\countable$-algebras" and "$\countable$-monoids" (Theorem~\ref{th:algebra-to-semigroup}
and Corollary~\ref{cor:algebra-to-monoid}).

\AP
Let us fix a finite "$\scountable$-algebra" $(S,\cdot,\tau,\tauop,\kappa)$.
In this section, we assume that all "words" are over the alphabet $S$. The objective 
of the construction is to attach to each "word" $u$ (over the alphabet $S$) a `value' 
in $S$. Furthermore, this value needs to be shown unique.  

\AP
The key ingredient for associating a unique value in $S$ to each "word" 
$u\in S^\scountable$ is the notion of \emph{"evaluation tree"}. Intuitively, 
this is an infinite tree describing a strategy for evaluating larger and larger 
factors of the "word" $u$. 
To define these objects, we need to first introduce the concept of 
\emph{"condensation tree"}, which is a convenient representation of 
nested "condensations" of a "linear ordering". This will provide the
underlying structure of an "evaluation tree".
The nodes of a "condensation tree" are "convex" subsets of the "linear ordering" 
and the descendant relation is given by inclusion.  
The set of children of each node defines a "condensation". 
Furthermore, in order to provide an induction parameter, 
we require that the branches of a "condensation tree" are finite (but their
length may not be uniformly bounded).

\begin{definition}\label{def:condensation-tree}
A ""condensation tree"" over a "linear ordering" $\alpha$ is a set $T$ of non-empty 
"convex" subsets of $\alpha$ such that:
\begin{itemize}
  \item $\alpha\in T$,
  \item for all $I,J$ in $T$, either $I\subseteq J$ or $J\subseteq I$ or $I\cap J=\emptyset$,
  \item for all $I\in T$, the union of all $J\in T$ such that $J\subsetneq I$ 
        is either $I$ or $\emptyset$,
  \item every subset of $T$ totally ordered by inclusion is finite.
\end{itemize}
\end{definition}

\AP
Elements in $T$ are called ""nodes"". The "node" $\alpha$ is 
called the ""root"" of the tree. Nodes minimal for $\subseteq$ 
are called ""leaves""; the other "nodes", including the "root", 
are called ""internal nodes"".
A "node" $I\in T$ is a ""descendant"" of a "node" $J\in T$
(and accordingly $J$ is an ""ancestor"" of $I$) if $I\subseteq J$.
If in addition we have $I\neq J$, then we say that $I$ is a
""proper descendant"" of $J$.
Similarly, $I$ is a ""child"" of a "node" $J$ 
(and accordingly $J$ is the ""parent"" of $I$) if $I\subsetneq J$ 
and, for all $K\in T$, $I\subsetneq K$ implies $J\subseteq K$.
According to the definition, if $I$ is an "internal node" of a "condensation tree" 
$T$ over $\alpha$, then it has a set of "children" that forms a partition 
of $I$ into "convex" subsets. 
We denote this partition by $\children_T(I)$, and we observe 
that it naturally corresponds to a "condensation" of $\alpha\suborder{I}$. 
When the tree $T$ is clear from the context, we will denote by 
$\children(I)$ the set of all children of $I$ in $T$ and, by extension, 
the corresponding "condensation" and the corresponding "condensed ordering".
Finally, we define the ""subtree"" of $T$ rooted at some of node $I$ of it 
as the "condensation tree" obtained by restricting $T$ to the "descendants" 
of $I$ (including $I$ itself).

\smallskip
\AP
We now introduce "evaluation trees". Intuitively, these are "condensation trees" 
where each "internal node" has an associated value in $S$ that can be `easily computed' 
from the values of its "children". Here it comes natural to consider a "word" $u$ 
`easy to compute' if it is isomorphic to either $a b$, $a\omegapow$, $a\omegaoppow$, 
or $P\etapow$, for some elements $a,b\in S$ and some non-empty set $P\subseteq S$.
Indeed, in each of these cases, the value of $u$ can be computed by a single application 
of the operations of the "$\scountable$-algebra" $(S,\cdot,\tau,\tauop,\kappa)$. 
Formally, the "words" that are easily computable are precisely those that belong
to the domain of the partial function $\pis$, defined just below:

\newcounter{definitionref}
\setcounter{definitionref}{\value{theorem}} \begin{definition}\label{def:base-mapping}
Let $""\pis""$ be the partial function from $S^\scountable$ to $S$ such that:
\begin{itemize}
  \item $\pis(a b)=a\cdot b$ for all $a,b\in S$,
  \item $\pis(e\omegapow)=e^\tau$ for all ""idempotents"" $e\in S$ 
        (i.e., all $e\in S$ such that $e\cdot e=e$),
  \item $\pis(e\omegaoppow)=e^{\tauop}$ for all "idempotents" $e\in S$,
  \item $\pis(P\etapow)= P^\kappa$ for all non-empty sets $P\subseteq S$,
  \item in all remaining cases, $\pis$ is undefined.
\end{itemize}
\end{definition}

\begin{definition}\label{def:evaluation-tree}
An ""evaluation tree"" over a "word" $u$ is a pair $\cT=(T,\pit)$, where $T$ is a 
"condensation tree" over the domain of $u$ and $""\pit"" $ is a 
function from $T$ to $S$ such that:
\begin{itemize}
  \item every "leaf" of $T$ is a singleton of the form $\{x\}$ and $\pit(\{x\})=u(x)$,
  \item for every "internal node" $I$ of $T$, 
        the partial function $\pis$ is defined on the "word" $\pit(\children(I))$
        that has domain $\children(I)$ and labels each 
        position $J \in \children(I)$ with $\pit(J)$;
        in addition, we have $\pit(I)=\pis(\pit(\children(I))$.
\end{itemize}
The \emph{value} of $(T,\pit)$ is defined to be $\pit(\alpha)$, 
i.e., the value of the "root".
\end{definition}

\AP
Let us turn back to the problem of associating a unique value in $S$ 
to each "word" $u\in S^\countable$. Based on the previous definitions, 
we can solve this problem in two steps.
First, we show that every "word" $u$ has an "evaluation tree", 
and thus a possible "value" that can be associated with it.
Then, we show that the associated "value" in fact does not depend 
on the choice of the "evaluation tree" over $u$, namely, that
"evaluation trees" over the same "word" induce the same "value".
The next two propositions formalize precisely these two steps.

\begin{proposition}\label{prop:existence-evaluation}
For every "word" $u$, there exists an "evaluation tree" over $u$.
\end{proposition}

\begin{proposition}\label{prop:equivalence-evaluation}
"Evaluation trees" over the same "word" have the same "value".
\end{proposition}

\noindent
The proofs of the two propositions are quite technical and 
deferred to Sections \ref{subsec:existence-evaluation} and
\ref{subsec:equivalence-evaluation}, respectively.
Before seeing those proofs in detail, we discuss the basic 
ingredients here. We then conclude the section by mentioning
a few important consequences of the developed framework.

\AP
The proof of Proposition \ref{prop:existence-evaluation} resembles the construction 
used by Shelah in his proof of decidability of the "monadic second-order" theory of 
"countable" "linear orderings" \cite{composition_method_shelah}. 
In particular, it uses a theorem of Ramsey \cite{ramsey} and a lemma stating 
that every non-trivial "word" indexed by a "countable" "dense" "linear ordering" 
has an "$\eta $-shuffle" as a "factor". Note that this latter lemma, and hence 
also Proposition \ref{prop:existence-evaluation}, relies on the fact that the 
"domain" of the "word" is "countable".
The proof of the proposition also makes use of 
Zorn's Lemma (or equally, the Axiom of Choice), 
so it is a proof in ZFC.
On the other hand, we observe that it does not make any use of 
Axioms~\refaxiom{axiom:concatenation}-\refaxiom{axiom:shuffle}.

\AP
Proposition \ref{prop:equivalence-evaluation} can be regarded as 
the core contribution of the paper, and its proof technique is quite 
original. For example. as opposed to Proposition~\ref{prop:existence-evaluation}, 
one cannot find any ingredient of the proof of Proposition~\ref{prop:equivalence-evaluation} 
in \cite{composition_method_shelah}.
The proof heavily relies on the use of Axioms~\refaxiom{axiom:concatenation}-\refaxiom{axiom:shuffle}.
As a matter of fact, each axiom can be seen as an instance of 
Proposition~\ref{prop:equivalence-evaluation} in some special cases 
of "evaluation trees" of small height. 
The proof also depends on Proposition~\ref{prop:existence-evaluation}, 
in the sense that is exploits in several places the existence of 
"evaluation trees" over arbitrary ("countable") "words".

\AP 
Another key ingredient for the proof of Proposition \ref{prop:equivalence-evaluation},
which is also reused in other proofs, is the formalization of a suitable 
induction principle on "condensation@condensation trees" and "evaluation trees". 
More precisely, by exploiting the fact that all branches of a "condensation tree" 
are finite, one can associate with any "condensation tree" $T$ a "countable" ordinal 
$\rank(T)$, called the ""rank"" of $T$.
Intuitively, this is the smallest ordinal $\beta$ that enables a labelling of 
the "nodes" of $T$ by ordinals less than or equal to $\beta$ in such a way that 
the label of each "node" is strictly greater than the labels of its "children". 

\begin{lemma}\label{lem:rank}
It is possible to associate with each "condensation tree" $T$ a "countable" ordinal 
$\rank(T)$ in such a way that $\rank(T') < \rank(T)$ for all "subtrees" $T'$ of $T$
rooted at "proper descendants" of the "root".
\end{lemma}

\begin{proof}
We associate with each "node" $I\in T$ a "countable" ordinal $\beta_I$ as follows. 
For every "leaf" $I$ of $T$, let $\beta_I=0$. Then, given an "internal node" $I$ of $T$, 
we assume that $\beta_J$ is defined for every "child" $J$ of $I$, and we define 
$\beta_I$ as the ordinal $\sup\{\beta_J+1 \:\mid\: J\in\children(I)\}$ 
(note that this is either a successor ordinal or a limit ordinal, depending 
on whether the set $\{\beta_J+1 \:\mid\: J\in\children(I)\}$ has a maximum 
element or not). 
Since $T$ has no infinite branch, it follows that $\beta_I$ is defined for every 
"node" of $T$. We thus let $\rank(T)=\beta_I$, where $I$ is the "root" of $T$. 
By construction, the function $\rank$ that maps any "condensation tree" 
$T$ to its rank $\rank(T)$ satisfies the properties stated in the lemma.
\end{proof}

\medskip
\AP
Now, assuming that Propositions \ref{prop:existence-evaluation} and 
\ref{prop:equivalence-evaluation} hold, we can prove the desired
correspondence between "$\scountable$-semigroups" and "$\scountable$-algebras":

\begin{theorem}\label{th:algebra-to-semigroup}
Every finite "$\scountable$-algebra" $({S,\cdot,\tau,\tauop,\kappa})$
is "induced" by a unique "product" $\pi$ from $S^\scountable$ to $S$.
\end{theorem}

\begin{proof}
Given a "word" $w$ with domain $\alpha$, one defines $\pi(w)$ to be the "value" of some 
"evaluation tree" over $w$ (the "evaluation tree" exists by Proposition~\ref{prop:existence-evaluation} 
and the "value" $\pi(w)$ is unique by Proposition~\ref{prop:equivalence-evaluation}). 

\AP
We prove that $\pi$ satisfies the "generalized associativity" property. 
Let $\sim$ be a "condensation" of the domain $\alpha$. 
For all classes $I\in\alpha\quotient{\sim}$, let $\cT_I$ be some "evaluation tree"
over $w\suborder{I}$. Let also $\cT'$ be some "evaluation tree" over the 
"word" $w'=\prod_{I\in\alpha\quotient{\sim}}\pi(w\suborder{I})$.
One constructs an "evaluation tree" $\cT$ over $w$ by 
first lifting $\cT'$ from the "linear ordering" $\alpha\quotient{\sim}$ to $\alpha$
(this is done by replacing each node $J$ in $\cT'$ by $\bigcup J$)
and then substituting each "leaf" of $\cT'$ corresponding to 
some class $I\in\alpha\quotient{\sim}$ with the "evaluation tree" $\cT_I$. 
The last step is possible (namely, respects the definition of "evaluation tree") 
because the "value" of each "evaluation tree" $\cT_I$ is $\pi(w\suborder{I})$, which coincides with 
the "value" $w'(I)$ at the "leaf" $I$ of $\cT'$. By Proposition \ref{prop:equivalence-evaluation}, 
the resulting "evaluation tree" $\cT$ has the same "value" as $\cT'$ and this proves that 
$\pi(w) = \pi\left(\prod_{I\in\alpha\quotient{\sim}} \pi(w\suborder{I})\right)$.

\AP
It remains to prove that the above choice of $\pi$ indeed "induces" the operators 
$\cdot,\tau,\tauop,\kappa$. This is done by a straightforward case analysis.
\end{proof}

\smallskip
\AP
The result that we just proved immediately implies an analogous 
correspondence between "$\countable$-monoids" and "$\countable$-algebras":

\begin{corollary}\label{cor:algebra-to-monoid}
Every finite "$\countable$-algebra" $(M,1,\cdot,\tau,\tauop,\kappa)$
is "induced" by a unique "product" $\pi$ from $M^\countable$ to $M$.
\end{corollary}

\medskip
\AP
Finally, we discuss the algorithmic implications of the above results.
In the same way as we talked of "languages" "recognized" by finite 
$\countable$-monoids, we can equally talk of "languages" 
""recognized by finite $\countable$-algebras"".
Moreover, because finite "$\countable$-algebras" are finite objects,
this enables the possibility of manipulating and reasoning on  
"recognized@recognized by finite $\countable$-algebras" "languages"
by means of algorithms. 
An example of such a possibility is given just below,
in a theorem that shows the decidability of the emptiness problem 
for "languages" "recognized by finite $\countable$-algebras".
The theorem also gives effective witnesses of non-empty "languages",
in the same spirit as some results of La\"uchli and 
Leonard for models of first-order logic and weak "monadic second-order" 
logic \cite{lauchli1966,lauchli1968}.
Other examples of algorithmic manipulation of "languages"
can be found in Section \ref{sec:logic-to-algebra}, where we will
prove some closure properties of "languages" "recognized by 
finite $\countable$-algebras".

\begin{theorem}\label{thm:emptiness}
The problem of testing whether $L\neq\emptyset$ for any "language" 
$L\subseteq A^\countable$ 
"recognized@recognized by finite $\countable$-algebras" 
by a given finite "$\countable$-algebra" is decidable.
Moreover, if $L\neq\emptyset$, 
then a finite expression can be effectively constructed 
that represents some "word" in $L$ and is generated by 
the following grammar:
$$
  \mathtt{w} ~::=~~ \emptystr ~|~~ a ~|~~ \mathtt{w}\cdot\mathtt{w} ~|~~ 
                         \mathtt{w}\omegapow ~|~~ \mathtt{w}\omegaoppow ~|~~ 
                         \{\mathtt{w},\ldots,\mathtt{w}\}\etapow
  \qquad\qquad \text{for $a\in A$.}
$$
\end{theorem}

\begin{proof}
Recall that a language $L\subseteq A^\countable$ is "recognized by a $\countable$-algebra"
$(M,1,\cdot,\tau,\tauop,\kappa)$ if there is $F\subseteq M$ and a "morphism"
$h: "(A^\countable,\prod)" \rightarrow (M,\pi)$ such that $L=h^{-1}(F)$, 
where $(M,\pi)$ is the "$\countable$-monoid" "induced" by $(M,1,\cdot,\tau,\tauop,\kappa)$ 
(Corollary~\ref{cor:algebra-to-monoid}). 
To decide the emptiness problem, it is sufficient to describe an algorithm that, 
given $(M,1,\cdot,\tau,\tauop,\kappa)$ and $h:A\rightarrow M$ (which uniquely 
extends to a function from $"(A^\countable,\prod)"$ to $(M,\pi)$), computes the set
$$
  h(A^\countable) =~ \big\{1\big\} \cup~ \big\{h(u) ~\big\mid~ u\in A^\countable\big\}
$$
(note that $L=h^{-1}(F)\neq\emptyset$ iff $h(A^\countable) \cap F \neq \emptyset$).

\AP
To compute the set $h(A^\countable)$, one can simply saturate the subset 
$\{1\}\cup h(A)$ of $M$ under the operations $\cdot,\tau,\tau^*,\kappa$. 
Formally, given $S\subseteq M$, we define the set ""generated"" by $S$ 
in $(M,1,\cdot,\tau,\tauop,\kappa)$ as the least set $\langle S\rangle$ 
that contains $S$ and satisfies the following closure properties:
\begin{itemize}
  \item if $a,b\in \langle S\rangle$, then $a\cdot b\in \langle S\rangle$,
  \item if $a\in \langle S\rangle$, then $a^\tau\in \langle S\rangle$,
  \item if $a\in \langle S\rangle$, then $a^\tauop\in \langle S\rangle$,
  \item if $\emptyset\neq P\subseteq \langle S\rangle$, then $P^\kappa\in \langle S\rangle$.
\end{itemize}
Clearly, the set $\langle S\rangle$ can be easily computed from $S$.

\AP
Below we prove that the set "generated" by $\{1\}\cup h(A)$,
denoted $\big\langle\{1\}\cup h(A)\big\rangle$, coincide with
$h(A^\countable)$.
First, it is easy to see 
that $\big\langle\{1\}\cup h(A)\big\rangle \subseteq h(A^\countable)$, 
since $\{1\}\cup h(A)\subseteq h(A^\countable)$ and containments in 
$h(A^\countable)$ are preserved under all operations of the saturation. 
The opposite containment $h(A^\countable) \subseteq \big\langle\{1\}\cup h(A)\big\rangle$ 
follows by Proposition~\ref{prop:existence-evaluation} and some inductive argument.
More precisely, one first observes that the value $h(w)$ of any "word"
$w\in A^\countable$ is the same as witnessed by some "evaluation tree" $\cT_w$.
Then, one exploits a simple induction on $\cT_w$ -- in fact, on the 
"rank" of the underlying "condensation tree" -- to verify that the set
$\big\langle\{1\}\cup h(A)\big\rangle$ contains the "value" of $\cT_w$. 

\AP
The above arguments show that the set 
$h(A^\countable) = \big\langle\{1\}\cup h(A)\big\rangle$
can be effectively constructed by a saturation procedure.
To conclude, we observe that this procedure implicitly associates with each 
element of $\big\langle\{1\}\cup h(A)\big\rangle$ a corresponding finite expression,
as generated by the grammar of the claim.
\end{proof}
 
\medskip
\subsection{Existence of "evaluation trees"}\label{subsec:existence-evaluation}

We introduce a few additional ingredients for the proof of 
Proposition \ref{prop:existence-evaluation}, namely, for 
showing the existence of "evaluation trees" over any "word". 
We begin with a variant of Ramsey's theorem for "additive labellings". 
Recall that $(S,\cdot,\tau,\tauop,\kappa)$ is a finite 
"$\scountable$-algebra" and, in particular, $(S,\cdot)$ is 
a finite semigroup. 

\begin{definition}\label{def:additive-labelling}
Let $(S,\cdot)$ be a semigroup. An ""additive labelling"" 
is a function $f$ that maps any two of points $x<y$ in a 
"linear ordering" $\alpha$ to an element $f(x,y)$ in $S$ in 
such a way that, for all $x<y<z$, $f(x,y)\cdot f(y,z) = f(x,z)$.
\end{definition}

\begin{lemma}[Ramsey \cite{ramsey}]\label{lemma:ramsey}
Given a "linear ordering" $\alpha$ with a minimum element $\bot$ and no maximum 
element, and given an "additive labelling" $f:\alpha\times\alpha\then(S,\cdot)$, there exist 
an $\omega$-sequence $\bot<x_1<x_2<\ldots$ of points in $\alpha$ and two elements $a,e\in S$ 
such that:
\begin{itemize}
  \item for all $y\in\alpha$, there is $x_i>y$,
  \item for all $i>0$, $f(\bot,x_i)=a$,
  \item for all $j>i>0$, $f(x_i,x_j)=e$.
\end{itemize}
\end{lemma}

\AP
Note that the conditions in the above lemma imply that $e$ is an "idempotent": 
indeed, we have $e\cdot e=f(x_i,x_{i+1})\cdot f(x_{i+1},x_{i+2})=f(x_i,x_{i+2})=e$.

\AP
In the same spirit of Lemma~\ref{lemma:ramsey}, the following lemma shows that every "countable" 
"dense" "word" contains an "$\eta$-shuffle" as a "factor". Even though this result appears 
already in \cite{composition_method_shelah}, we give a proof of it for the sake of self-containment.

\begin{lemma}[Shelah \cite{composition_method_shelah}]\label{lemma:densetoperfectshuffle}
Every "word" indexed by a non-empty non-singleton "countable" "dense" "linear ordering" 
contains a "factor" that is an "$\eta$-shuffle".
\end{lemma}

\begin{proof}
Let $\alpha$ be a non-empty non-singleton "countable" "dense" "linear ordering", 
let $A=\{a_1,\ldots,a_n\}$ be a generic alphabet, and let $w$ be a "word" over 
$A$ with domain $\alpha$. For the sake of brevity, given a symbol $a\in A$, 
we denote by $w^{-1}(a)$ the set of all points $x\in\alpha$ 
such that $w(x)=a$. We then define $w_0=w$ and $A_0=\emptyset$,
and we recursively apply the following construction for each index $1\le i\le n$:
$$
\begin{array}{rcl}
  A_i &=& \begin{cases}
            A_{i-1}\cup\{a_i\} & \text{if $w_i^{-1}(a_i)$ is "dense in" $\dom(w_i)$}, \\
            A_{i-1}            & \text{otherwise},
          \end{cases}   
  \\[4ex]
  w_i &=& \begin{cases}
            w_{i-1} \phantom{\cup\{s_i\}}\;\, 
                               & \text{if $w_i^{-1}(a_i)$ is "dense in" $\dom(w_i)$}, \\[2ex]
            w_{i-1}\suborder{I}& \text{otherwise, where $I$ is any open non-empty} \\ 
                               & \text{"convex" subset of $\alpha$ 
                                       such that $w^{-1}(a_i)\cap I=\emptyset$}.
          \end{cases}
\end{array}
$$
By construction, the domain of the "factor" $w_n$ is non-empty, non-singleton,
"countable", and "dense". Moreover, for all symbols $a\in A$, either
$w_n^{-1}(a_j)$ is "dense in" $\dom(w_n)$ or empty, depending on whether
$a\in A_n$ or not.This shows that $w_n$ is an "$\eta$-shuffle" of the set 
$A_n$.
\end{proof}

\AP
We are now ready to prove Proposition \ref{prop:existence-evaluation}:

\begin{proof}[Proof of Proposition \ref{prop:existence-evaluation}]
Let $u$ be a "word" with "countable" domain $\alpha$. 
We say that a "convex" subset $I$ of $\alpha$ 
is ""definable"" if there is an "evaluation tree" over the "factor" $u\suborder{I}$. 
Similarly, we say that $I$ is ""strongly definable"" if every non-empty "convex" subset 
$J$ of $I$ is "definable". We first establish the following claim: 

\begin{claim}
For every ascending chain $I_0\subseteq I_1\subseteq\ldots$ of "strongly definable convex subsets" 
of $\alpha$, the limit $I=\bigcup_{i\in\bbN}I_i$ is "strongly definable".
\end{claim}

\begin{proof}[Proof of claim]
Let $J$ be a non-empty "convex" subset of $I$ and let $J_i=I_i\cap J$ for all $i\in\bbN$. 
We prove that $J = \bigcup_{i\in\bbN}J_i$ is "definable", namely, we show how to construct 
an "evaluation tree" over the "factor" $u\suborder{J}$. Without loss of generality we assume that 
the $J_i$'s are non-empty. Note that all the $J_i$'s are "strongly definable". 
Of course, if the sequence of the $J_i$'s is ultimately constant, then $J=J_i$ for a sufficiently 
large $i\in\bbN$ and the existence of an "evaluation tree" over $u\suborder{J}$ follows trivially from 
the fact that $J_i$ is "strongly definable".
We now consider the case when all the $J_i$'s coincide on the left. We can 
partition $J$ into a sequence of "convex" subsets $K_0<K_1<\ldots$, where $K_0=J_0$ 
and $K_{i+1}=J_{i+1}\setminus J_i$ for all $i\ge1$. The "convex" subsets $K_i$ form a 
"condensation" of $J$ such that $J_i=K_0\cup\ldots\cup K_i$ for all $i\in\bbN$. 
For every $i<j$ in $\bbN$, we define $K_{i,j}=K_i\cup\ldots\cup K_{j-1}$. We recall that 
every "convex" $J_j$, as well as every "convex" subset $K_{i,j}$ of it, is 
"strongly definable". We can thus associate with each $K_{i,j}$ an "evaluation tree" 
$\cT_{i,j}$ over $u\suborder{K_{i,j}}$. We denote by $c_{i,j}$ the "value" of $\cT_{i,j}$. 
Using Lemma \ref{lemma:ramsey} (i.e., Ramsey's Theorem), one can extract a sequence 
$0<i_1<i_2<\ldots$ in $\omega$ such that $c_{i_1,i_2}=c_{i_2,i_3}=\ldots$ (and moreover, 
this element is an "idempotent"). We can then construct an "evaluation tree" over $u\suborder{J}$ 
that has root $J$ and the "convex" subsets $K_{0,i_1}$, $K_{i_1,i_2}$, \dots for "children", 
with the associated "evaluation@evaluation tree" "subtrees" $\cT_{0,i_1}$, $\cT_{i_1,i_2}$, \dots. 
This allows us to conclude that $J$ is a "definable convex" when the $J_i$'s coincide 
on the left.
The case where the $J_i$'s coincide on the right is symmetric. 
Finally, in the general case, we can partition each set $J_i$ into two 
subsets $J'_i$ and $J''_i$ such that (i) $J'_i<J''_i$, (ii) the sequence 
of the $J'_i$'s coincide on the right, and (iii) the sequence of the $J''_i$'s 
coincide on the left. Let $J'=\bigcup_{i\in\bbN}J'_i$ and $J''=\bigcup_{i\in\bbN}I''_i$. 
One knows by the cases above that there exist "evaluation trees" over $u\suborder{J'}$ and 
over $u\suborder{J''}$. Finally, one can easily construct an "evaluation tree"
over $u\suborder{J} = u\suborder{J'\cup J''}$ out of the "evaluation trees" for $J'$ and $J''$. 
This proves that $J$ is "definable" and hence $I$ is "strongly definable".
\end{proof}

\AP
Turning back to the main proof, let us now consider the set $\cC$ of all 
"condensations" $C$ of $\alpha$ such that every class is "strongly definable". 
"Condensations" in $\cC$ are naturally ordered by the `finer than' relation. 
Let us consider a chain $(C_i)_{i\in\beta}$ of "condensations" in $\cC$ ordered 
by the finer than relation, i.e., for all $j<i$ in $\beta$, $C_j$ is finer than $C_j$. 
Since $\alpha$ is "countable", one can assume that $\beta$ is "countable", 
or even better that $\beta=\omega$. Let us consider the ""limit condensation"" $C$, 
i.e., the finest "condensation" that is coarser than every $C_i$. 
Each class $I\in C$ is the union of a sequence of convex subsets $I_i$, with $I_i\in C_i$ 
for all $i\in\bbN$. From the assumption that every "condensation" $C_i$ belongs 
to $\cC$, we get that $I_i$ is "strongly definable" and from the claim above, we 
conclude that $I$ is "strongly definable" as well. This shows that the 
"limit condensation" $C$ belongs to $\cC$ and hence every chain of $\cC$ 
has an upper bound in $\cC$.

\AP
It follows that we can apply Zorn's Lemma and deduce that $\cC$ contains a maximal 
element, say $C$. If $C$ is a "condensation" with single class, this means that 
there exists an "evaluation tree" over $u$ and the proposition is established. 
Otherwise, we shall head toward a contradiction. Consider the "condensed ordering" 
induced by $C$ (by a slight abuse of notation, we denote it also by $C$). Two cases 
can happen: either $C$ contains two consecutive classes or $C$ is a "dense" "linear order". 

\AP
In the former case, we fix two consecutive classes $I,I'\in C$, with $I<I'$.
We observe that each class of $C$ is a limit of "strongly definable convexes" and hence,
by the previous claim, it is also "strongly definable". It is then easy to see that
the union $I\cup I'$ of the two consecutive "strongly definable convexes" $I$ and $I'$
is also "strongly definable", which contradicts the definition of $C$. 

\AP
In the second case we have that the "linear ordering" $C$ is "dense in itself". As before, we 
recall that each class of $C$ is "strongly definable" and we prove that there exist non-trivial 
unions of classes of $C$ that are "strongly definable" (a contradiction). 
We begin by associating with each "convex" subset $J$ of a class $I$ of $C$ an 
"evaluation tree" $\cT_J$ over $u\suborder{J}$ and we denote by $c_J$ the "value" induced by it. 
We then consider the "word" $v = \prod_{I\in C} c_I$. 
We know from Lemma~\ref{lemma:densetoperfectshuffle} that that $v$ contains a "factor"
that is an "$\eta$-shuffle", say, $v' = v\suborder{C'}$ for some "convex" $C'\subseteq C$. 
Let $J=\bigcup_{I\in C'}I$. To prove that $J$ is "strongly definable" we consider a 
"convex" $K\subseteq J$ and we construct an "evaluation tree" $\cT_K$ over $u\suborder{K}$ as follows. 
First we observe that $K$ is the union of all non-empty "convexes" of the form $I\cap K$, 
for $I\in C'$, and that each set $I\cap K$ is contained in a class of $C$, hence it is 
"definable" and has "value" $c_{I\cap K}$. 
Now, one needs to distinguish some cases depending on whether $C'$ contains minimal/maximal 
"convexes" $I$ intersecting $K$. For the sake of simplicity, we only consider the case 
where $C'$ contains a minimal "convex" $I_0$ such that $I_0\cap K\neq\emptyset$, but no 
maximal convex $I$ such that $I\cap K\neq\emptyset$.
In this case, we recall that $v'$ is an "$\eta$-shuffle" and that its restriction to 
the non-empty "convexes" $I\cap K$, with $I\in C'$, is the juxtaposition of the singleton 
$c_{I_0\cap K}$ and the "$\eta$-shuffle" $\prod_{I\in C''} (c_{I\cap K})$, 
where $C''=\{I\in C' \:\mid\: I\cap K\neq\emptyset,~I\neq I_0\}$.
An "evaluation tree" $\cT_K$ over $u\suborder{K}$ can be constructed by appending to the
"root" $K$ two subtrees: the "evaluation tree" $\cT_{I_0\cap K}$ associated with 
the "definable convex" $I_0\cap K$, and the "evaluation tree" $\cT_{K\setminus I_0}$ 
that consists of the node $K\setminus I_0$ and the direct subtrees $\cT_{I\cap K}$, 
for all $I\in C''$. This shows that there is a non-trivial union $J$ of 
classes of $C$ that is "strongly definable", which contradicts the definition of $C$.
\end{proof}
 
\smallskip
\subsection{Equivalence of "evaluation trees"}\label{subsec:equivalence-evaluation}

We now turn towards proving Proposition \ref{prop:equivalence-evaluation},
namely, the equivalence of "evaluation trees" with respect to the "induced values". 
As we already mentioned, the proof is rather long and requires a series of technical 
lemmas.

\AP
For reasons that will be clear in the sequel, it is convenient to extend 
slightly the domain of the partial function $\pis$ that computes "values" of 
`simple "words"' (cf.~Definition \ref{def:base-mapping}). 
Intuitively, such an extension adds prefixes and suffixes 
of finite length to the elements of the original domain of $\pis$. 

\newcounter{definitionback}
\setcounter{definitionback}{\value{theorem}}   \setcounter{theorem}{\value{definitionref}}    \let\thedefinitionback\thetheorem              \def\thetheorem{\thedefinitionback bis}        \begin{definition}\label{def:base-mapping-bis}
We \emph{""extend the partial function $\pis$""} in such a way that:
\begin{itemize}
  \item $\epis(a_1\ldots a_n) = a_1\cdot\ldots\cdot a_n$ for all $n\ge1$ and all $a_1,\ldots,a_n\in S$,
  \item $\epis(a\:b\omegapow) = a\cdot b^\tau$ for all $a,b\in S$,
  \item $\epis(a\omegaoppow\:b) = a^\tauop\cdot b$ for all $a,b\in S$,
  \item $\epis(a\: P\etapow\: b) = a \cdot P^\kappa\cdot b$ for all $a,b\in S\uplus\{\emptystr\}$ 
        \par\noindent
        (by a slight abuse of notation, we let $\emptystr\cdot s = s\cdot \emptystr = s$
         for all $s\in S$),
  \item in all remaining cases, $\epis$ remains undefined.
\end{itemize}
\end{definition}
\let\thetheorem\thedefinitionback              \setcounter{theorem}{\value{definitionback}}   
\noindent
The new definition of $\epis$ results in a more general notion of "evaluation tree".
Note that the definition of "rank" of an "evaluation tree" still applies to this 
generalized notion, since the "rank" was in fact defined on "condensation trees"
independently of $\pis$.
The generalized notion of "evaluation tree", together with the associated "rank",
will give a strong enough invariant for having a proof by induction of the 
equivalence of "evaluation trees".\footnote{The extended definition of $\epis$ could have been introduced straight 
          at the beginning, in place of Definition \ref{def:base-mapping}. 
          Of course, all the results in the paper would still hold, 
          but some proofs would become slightly more involved (in particular,
          those that show the correspondence between "recognizability" and
          "MSO" definability). This explains why we prefer to adopt
          a more restrictive definition of "evaluation tree", and use
          the extended version only here for convenient.}

\smallskip
\AP
The lemma below basically shows that if the (extended) partial mapping $\epis$ is defined 
over a "word", then it is also defined over all its "factors". It is convenient here to 
allow also some change of letters at the extremities of the "word" and make some case 
distinctions for dealing with the empty "word" $\emptystr$. This makes the statement 
of the following lemma a bit more technical. 

\begin{lemma}\label{lemma:restriction-pibase}
If $\epis$ is defined over a non-empty "word" of the form $u\:c\:v$, with 
$u,v\in S^\scountable\uplus\{\emptystr\}$ and $c\in S\uplus\{\emptystr\}$, 
then it is also defined over the "words" $u\:a$ and $b\:v$, for all 
$a,b\in S\uplus\{\emptystr\}$ such that $a=b=\emptystr$ implies $c=\emptystr$.
In addition, if $a=b=c=\emptystr$ or $a\cdot b=c$, then 
$\epis(u\:c\:v)=\epis(u\:a)\cdot\epis(b\:v)$.
\end{lemma}

\noindent
The proof of the lemma is straightforward by a case distinction, 
and thus omitted.

\smallskip
\AP
The next step consists in showing how to restrict a "condensation tree" to an
arbitrary "convex" subset (further along, we will lift this operation to 
the generalized notion of "evaluation tree"):

\begin{definition}\label{def:generalized-subtree}
Given a "condensation tree" $T$ over a "linear ordering" $\alpha$ and a 
"convex" subset $I$ of $\alpha$ (not necessarily an element of $T$), 
define the ""generalized subtree"" of $T$ rooted at $I$ as follows:
$$
  T\subtree{I} \eqdef~ \{I\cap J \:\mid\: J\in T,~I\cap J\neq\emptyset\}.
$$
\end{definition}

\AP
The above operation can be seen as a further generalization of the notion 
of "subtree" that was given just after Definition \ref{def:condensation-tree}.
Below we prove that, not only $T\subtree{I}$ is a valid "condensation tree",
but also that this operation does not increase the "rank".

\begin{lemma}\label{lemma:restriction-condensation}
If $T$ is a "condensation tree" over $\alpha$ and $I$ is a "convex" subset of
$\alpha$, then $T\subtree{I}$ is a "condensation tree" over $\alpha\cap I$. 
Furthermore, we have $\rank(T\subtree{I}) \le \rank(T)$ and
$(T\subtree{I})\subtree{J}=T\subtree{J}$
for all "convex" subsets $J$ of $I$. 
\end{lemma}

\begin{proof}
We only prove that $T\subtree{I}$ is a "condensation tree".
The remaining claims follow easily from our definitions.
The property stated in the first item of Definition \ref{def:condensation-tree}
follows from the fact that $\alpha\in T$ and $\alpha\cap I = I \in T\subtree{I}$.
To prove the property in the second item, consider two "convexes" $J,K$ in $T$. 
We have that either $J\subseteq K$ or $K\subseteq J$ or $J\cap K=\emptyset$.
As a consequence, either $J\cap I\subseteq K\cap I$ or $K\cap I\subset J\cap I$
or $(J\cap I)\cap (K\cap I)=\emptyset$.
Now, for the third item, consider two "convexes" $J,K$ in $T$ such that 
$K\cap I \in T\subtree{I}$ (or, equally, $K\cap I\neq\emptyset$) and 
$(K\cap I)\subsetneq(J\cap I)$. Since $K\cap I$ is non-empty, this means 
that $J \cap K$ is non-empty too. Thus, either $K\subseteq J$ or $J\subseteq K$.
If $J\subseteq K$ held, then we would have $(J\cap I)\subseteq (K\cap I)$,
which would contradict $(K\cap I)\subsetneq(J\cap I)$. We thus conclude that 
$K\subseteq J$. 
It remains to verify the property in the fourth item, namely, the fact that any 
subset of $T\subtree{I}$ that is totally ordered by inclusion is finite. Consider such a 
subset $C$. For each $J\in C$, define $T_{\supseteq J}=\{K\in T \:\mid\: K\cap I\supseteq J\}$. 
By construction, $T_{\supseteq J}$ is a subset 
of $T$ that is totally ordered by inclusion. In particular, $T_{\supseteq J}$ is 
finite and has a minimal element, denoted $\min(T_{\supseteq J})$. We define $C'$ 
as the set of all "convexes" of the form $\min(T_{\supseteq J})$, with $J\in C$. 
Since $J\subseteq J'$ implies $\min(T_{\supseteq J}) \subseteq \min(T_{\supseteq J'})$,
we have that $C'$ is a subset of $T$ totally ordered by inclusion, and hence $C'$
is finite. Moreover, since each "convex" $J\in C$ can be written as $\min(T_{\supseteq J})\cap I$,
we have that, for all $J,J'\in C$, $J\neq J'$ iff 
$\min(T_{\supseteq J}) \neq \min(T_{\supseteq J'})$.
Since $C'$ is finite, we conclude that $C$ is finite too.
\end{proof}

\smallskip
\AP
Putting together all the previous definitions and lemmas, we can show that 
an "evaluation tree" $\cT=(T,\pit)$ over a "word" $u$ provides not only a 
"value" for $u$, but also, via restrictions to "generalized subtrees", 
"values" for all the "factors" of $u$. Intuitively, this means that the 
mapping $\pit$ of $\cT$ can be extended to all "convex" subsets $I$ of $\alpha$: 

\begin{lemma}\label{lemma:restriction-evaluation}
For every "evaluation tree" $\cT=(T,\pit)$ over a "word" $u$ with domain $\alpha$ and 
every "convex" subset $I$ of $\alpha$, there is an "evaluation tree" 
$\cT\subtree{I}=(T\subtree{I},\pit_I)$ 
such that $\pit_I$ and $\pit$ coincide over 
$(T\subtree{I})\cap T = \{J\in T \:\mid\: J\subseteq I\}$. 
\par
In particular, by a slight abuse of notation, one can denote 
by $\pit(I)$ the "value" associated 
with the "convex" $I$ in the "evaluation tree" $\cT\subtree{I}$ 
(this notation is consistent with the 
"value" associated with $I$ in the "evaluation tree" $\cT$, when $I\in T$). 
\end{lemma}

\begin{proof}
Let us first assume that $I$ is an initial segment of $\alpha$, namely, for every $y\in I$ 
and every $x\leq y$, $x\in I$. The proof is by induction on $T$, namely, on the "rank" of 
the underlying "condensation tree". Let $C=\children(\alpha)$ be the top-level 
"condensation" of $T$. We distinguish between two subcases. 

\AP
If the "condensation" $\{I,~\alpha\setminus I\}$ is coarser than $C$, 
then for all $K\in T\subtree{I}$, with $K\neq I$, we have $K\in T$. Hence 
it makes sense to define $\pit_I(K)=\pit(K)$. We complete the definition by 
letting $\pit_I(I)=\epis(\pit(C\suborder{I}))$, where $\pit(C\suborder{I})$ is the "word" 
with domain $C\suborder{I} = \{K\in C \:\mid\: K\subseteq I\}$ and with each position $J$ 
labelled by the "value" $\pit(J)$ 
(thanks to Lemma~\ref{lemma:restriction-pibase} the function $\epis$ is 
defined on the "word" $\pit(C\suborder{I})$). It is easy to check that the $(T\subtree{I},\pit_I)$ 
thus defined is an "evaluation tree" over the "factor" $u\suborder{I}$ and 
that $\pit_I$ and $\pit$ coincide over $(T\subtree{I})\cap T = \{K\in T \:\mid\: K\subseteq I\}$. 

\AP
Otherwise, if the condensation $\{I,~\alpha\setminus I\}$ is not coarser than $C$, 
then there exist three convex subsets $J_1<J_2<J_3$ of $\alpha$ such that 
(i) $\{J_1,J_2,J_3\}$ forms a partition coarser than $C$, 
(ii) $J_1\subseteq I$, 
(iii) $J_3\subseteq\alpha\setminus I$, and 
(iv) $J_2\in C$ with $J_2\cap I\neq\emptyset$ and $J_2\setminus I\neq\emptyset$. 
In particular, we have that the "convex" $J_2\cap I$ is included in a proper 
descendant of the "root" of $T$, and hence by
Lemmas \ref{lem:rank} and \ref{lemma:restriction-condensation}
$\rank(T\subtree{J_2\cap I}) < \rank(T)$.
We can thus apply the induction hypothesis to construct the 
"evaluation tree" $\cT\subtree{J_2\cap I} = (T\subtree{J_2\cap I},\pit_{J_2\cap I})$. 
Note that for every $K\in T\subtree{J_1}$ with $K\neq J_1$, we have $K\in T$. 
Hence it makes sense to define $\pit_I(K)=\pit(K)$. For every $K\in T\subtree{J_2}$, 
we define $\pit_I(K)=\pit_{J_2\cap I}(K)$. Finally, we let
$\pit_I(I) = \epis\big(\pit_I(C\suborder{J_1})~\pit_{J_2\cap I}(J_2\cap I)\big)$
(again this is well defined thanks to Lemma~\ref{lemma:restriction-pibase}). 
It is easy to check that the $(T\subtree{I},\pit_I)$ thus defined is an 
"evaluation tree" over $u\suborder{I}$ and that $\pit_I$ and $\pit$ 
coincide over $(T\subtree{I})\cap T = \{K\in T \:\mid\: K\subseteq I\}$.

\AP
The proof for the symmetric case, where $I$ is a final segment of $\alpha$, is analogous.

\AP
Finally, we consider the case where $I$ is not an initial segment, nor a final 
segment of $\alpha$. In this case it is possible to write $I$ as $I_1\cap I_2$, 
where $I_1$ is an initial segment and $I_2$ is a final segment of $\alpha$. 
By Lemma \ref{lemma:restriction-condensation} we have
$T\subtree{I} = (T\subtree{I_1})\subtree{I_2}$, and hence it suffices 
to apply twice the cases for the initial/final segment discussed above.
\end{proof}

\medskip
\AP
Now that we have set up the basic tools for reasoning on "evaluation trees"
and their restrictions, we begin to exploit the "axioms of $\scountable$-algebras" 
to prove a series of equivalence results.
The first of these results can be seen as a form of "associativity" rule for 
the function $\epis$, but for which equality is required to hold only when 
every expression is defined:

\begin{lemma}\label{lemma:pibase-substitution}
For every "word" $u$ of the form $\prod_{i\in\alpha}u_i$, with $\alpha$ 
"countable" "linear ordering" and $u_i\in S^\scountable$ for all $i\in\alpha$, 
if both $\epis(u)$ and $\epis\big(\prod_{i\in \alpha}\epis(u_i)\big)$ 
are defined, then the two values are equal.
\end{lemma}

\begin{proof}
We prove the lemma by a case analysis, namely, by distinguishing the "order type" of $u$ 
(recall that, since $\epis(u)$ is defined, the "order type" of $u$ must be either finite, 
$\omega$, $\omegaop$, $\eta$, $1+\eta$, $\eta+1$, or $1+\eta+1$). 
For the sake of brevity, we let $v=\prod_{i\in\alpha}\epis(u_i)$.

\AP
If $u=a_1\ldots a_n$ for some $a_1,\ldots,a_n\in S$, then $v$ has to be of the form $b_1\ldots b_m$, 
for some $m\ge 1$ and some $b_1,\ldots,b_n\in S$. Since $\cdot$ is associative 
(see Axiom \refaxiom{axiom:concatenation}), 
we obtain $\epis(u)=a_1\cdot\ldots\cdot a_n=b_1\cdot\ldots\cdot b_m=\epis(v)$. 

\AP
If $u=a\:e\omegapow$ for some $a,e\in S$, with $e$ "idempotent", then $v$ can be 
either of the form $c_1\ldots c_m$, for some $m\ge 1$ and some $c_1,\ldots,c_m\in S$, 
or of the form $b\:f\omegapow$, for some $b,f\in S$, with $f$ "idempotent". 
If $v=c_1\ldots c_m$, say with $m\geq 2$ (the case $m=1$ is trivial), 
then we necessarily have $c_1=a\cdot e^{n_1}=a\cdot e$ for some $n_1\ge0$, 
$c_i=e^{n_i}=e$ for all $2\le i<m-1$ and some $n_2,\ldots,n_{m-1}\ge 1$,
and $c_m=e^\tau$. Axioms \refaxiom{axiom:concatenation} and \refaxiom{axiom:omega} 
together imply $e\cdot e^\tau=e\cdot(e\cdot e)^\tau=(e\cdot e)^\tau=e^\tau$. 
We thus have $\epis(u) = a\cdot e^\tau = c_1\cdot\ldots\cdot c_m = \epis(v)$. 
Otherwise, if $v=b\:d\omegapow$, then, as above, we get $b=a\cdot e^{n_1}=a\cdot e$,
for some $n_1\ge 0$, and $f=e^{n_2}=e^{n_3}=\ldots=e$, for some $n_2,n_3,\ldots\ge 1$. 
Using Axioms \refaxiom{axiom:concatenation} and \refaxiom{axiom:omega}
we finally derive $\epis(u) = a\cdot e^\tau = b\cdot f^\tau = \epis(v)$. 

\AP
The case $u=e\omegaoppow\: a$ is just symmetric to the previous case and 
uses Axiom \refaxiom{axiom:omegaop} instead of Axiom \refaxiom{axiom:omega}. 

\AP
Finally, the most interesting case is when $u=a\: P\etapow \:b$ for some 
non-empty set $P\subseteq S$ and some empty or singleton "words" $a,b\in S\uplus\{\emptystr\}$. 
We further distinguish some cases depending on the form of $v$:
\begin{itemize}
  \item If $v=c_1\ldots c_m$, then the proof goes by induction on $m$. 
        The interesting base case is $m=2$ (for $m=1$ the claim holds trivially).
        We further distinguish between five subcases. 
        If the first "factor" $u_1$ has no last letter and 
        the last "factor" $u_2$ has no first letter, then we 
        have $c_1=\epis(u_1)=a\cdot P^\kappa$ and $c_2=\epis(u_2)=P^\kappa\cdot b$. 
        Using Axiom~\refaxiom{axiom:shuffle}, we get 
        $\epis(u)=a\cdot P^\kappa\cdot b=(a\cdot P^\kappa)\cdot(P^\kappa\cdot b)=c_1\cdot c_2=\epis(v)$.
        If $u_1$ consists of a single letter, then this letter must be $a\neq\emptystr$. 
        Moreover, $u_2$ cannot have a first letter and hence, as above, 
        we have $\epis(u_2)=P^\kappa\cdot b$. We thus derive
        $\epis(u)=a\cdot P^\kappa\cdot b=c_1\cdot c_2=\epis(v)$.
        If $u_1$ has a last letter, say $p$, but length greater than $1$,
        then $p$ must belong to $P$ and $u_2$ has no first letter. 
        We thus have 
        $\epis(u)=a\cdot P^\kappa\cdot b=(a\cdot P^\kappa\cdot p)\cdot (P^\kappa\cdot b)=\epis(v)$.
        The cases where $u_2$ has length $1$ and where $u_2$ has a first letter and length greater
        than $1$ are symmetric. Finally, the induction for $m>2$ is straightforward.
  \item If $v=c\:e\omegapow$, then, by distinguishing some subcases as 
        above, one verifies that $c=\epis(u_1)$ is either $a$ or $a\cdot P^\kappa\cdot p$, 
        for some $p\in P\uplus\{\emptystr\}$, and that $e=\epis(u_2)=\epis(u_3)=\ldots$ 
        is either $P^\kappa\cdot q$ or $q\cdot P^\kappa$, for some $q\in P\uplus\{\emptystr\}$, 
        depending on whether $u_1$ has a first letter or not. 
        Depending on the various subcases, and using Axiom~\refaxiom{axiom:shuffle}, 
        we derive either $\epis(u)=a\cdot P^\kappa=a\cdot (P^\kappa\cdot q)^\tau=\epis(v)$, 
        or $\epis(u)=a\cdot P^\kappa=(a\cdot P^\kappa\cdot p)\cdot (P^\kappa\cdot q)^\tau=\epis(v)$,
        or $\epis(u)=a\cdot P^\kappa=(a\cdot P^\kappa)\cdot (q\cdot P^\kappa)=\epis(v)$.
  \item If $v=e\omegaoppow\:c$, then the claim holds by symmetry with the previous case.
  \item If $v=c\:R\etapow\:d$ for some non-empty set $R\subseteq S$ and
        some empty or singleton "words" $c,d\in S\uplus\{\emptystr\}$,
        then we prove that $R$ is included in 
        $P \;\cup\; (P\uplus\{\emptystr\})\cdot P^\kappa\cdot(P\uplus\{\emptystr\})$.
        Let us treat first the case $c=d=\emptystr$. Since $v$ has no first nor 
        final letter, this implies $a=b=\emptystr$. Let us consider an element $r\in R$ 
        and a corresponding factor $u_i$ of $u$, with $i\in\alpha$, such that $\epis(u_i)=r$.
        If $u_i$ consists of the single letter $r$, then we clearly have $r\in P$. 
        Otherwise $u_i$ has more than one letter and, depending on the existence of 
        a first/last letter in $u_i$, we get one of the four possibilities 
        $r=P^\kappa$, $r=p\cdot P^\kappa$, $r=P^\kappa\cdot q$ and $r=p\cdot P^\kappa\cdot q$,
        for suitable $p,q\in P$. 
        This proves that $R$ is included in 
        $P \;\cup\; (P\uplus\{\emptystr\})\cdot P^\kappa\cdot(P\uplus\{\emptystr\})$.
        Using Axiom~\refaxiom{axiom:shuffle} we immediately obtain 
        $\epis(u)=P^\kappa=R^\kappa=\epis(v)$.
        The general case where $c,d\in S\uplus\{\emptystr\}$, 
        can be dealt with by using similar arguments plus 
        Axiom~\refaxiom{axiom:concatenation}.
\end{itemize}
\vspace{-2ex}
\hfill\ 
\end{proof}

\begin{corollary}\label{cor:evaluation-pibase}
Let $u$ be a "word" with domain $\alpha$ such that $\epis(u)$ is defined and 
let $\cT=({T,\pit})$ be an "evaluation tree" over $u$. Then $\epis(u)=\pit(\alpha)$.
\end{corollary}

\begin{proof}
We prove the claim by induction on $\cT$. If $\cT$ consists of a single "node", then
this "node" must be a "leaf" and $\alpha$ must be a singleton "leaf", and hence
the claim follows immediately by definition of "evaluation tree". Otherwise, 
let $C=\children(\alpha)$ be the top-level "condensation". By Lemma~\ref{lemma:restriction-pibase}, 
we know that $\epis(u\suborder{I})$ is defined for all $I\in C$. We can then use the induction 
hypothesis on the "evaluation tree" $\cT\subtree{I}$ and obtain $\epis(u\suborder{I})=\pit(I)$. 
Finally, using Lemma~\ref{lemma:pibase-substitution}, we get 
$\epis(u)=\epis(\pit(C))=\pit(\alpha)$.
\end{proof}

\AP
The following series of lemmas prove equalities between the "value" at the "root"
of an "evaluation tree" and the "values" induced by $\epis$ under different 
"condensations" of the "root". We first consider finite condensations, then 
$\omega$-"condensations" 
(and, by symmetry, $\omegaop$-"condensations"), and finally $\eta$-"condensations". 
The gathering of those results will naturally entail that two "evaluation trees" over 
the same "word" have the same "value" (see Corollary~\ref{cor:equivalence-evaluation}).

\begin{lemma}\label{lemma:equivalence-evaluation-finite}
Given a "word" $u$ with domain $\alpha$, an "evaluation tree" $\cT=({T,\pit})$ over $u$, 
and a finite "condensation" $I_1<\ldots<I_n$ of $\alpha$, we 
have $\pit(\alpha)=\pit(I_1)\cdot\ldots\cdot\pit(I_n)$.
\end{lemma}

\begin{proof}
The proof is by induction on $\cT$. If $\cT$ consists of a single "leaf", then 
$\alpha$ must be a singleton and hence $n=1$ and the claim follows trivially. 
Let us now consider the case where $\cT$ has more than one "node". 
We only prove the claim for $n=2$ (for $n=1$ it is obvious and for $n>2$ it follows 
from a simple induction). Let $C=\children(\alpha)$ be the top-level "condensation" 
and let $J$ be the unique "convex" subset in $C$ that intersects both $I_1$ and $I_2$ (if $C$ 
does not contain such an element, then we let $J=\emptyset$). For the sake of brevity, we define, 
for both $i=1$ and $i=2$, $C_i = \{K\in C \:\mid\: K\subseteq I_i\}$, $u_i=\prod_{K\in C_i}\pit(K)$, 
and $a_i=\pit(J\cap I_i)$ (with the convention that $\pit(J\cap I_i)=\emptystr$ if $J=\emptyset$). 
Note that $C=C_1\cup\{J\}\cup C_2$ if $J\neq\emptyset$ (resp., $C=C_1\cup C_2$ if $J=\emptyset$)
and hence $\pit(\alpha)=\epis\big(u_1\:\pit(J)\:u_2\big)$ (we assume that $\pit(J)=\emptystr$ if 
$J=\emptyset$). 
Let us consider the case where $J$ is not empty (the case $J=\emptyset$ is similar). 
Since $J\in C$ and $C=\children(\alpha)$, we have $\rank(\cT\subtree{J}) < \rank(\cT)$ 
and hence we can apply the induction hypothesis to the "evaluation tree" $\cT\subtree{J}$ and 
the "condensation" $\{J\cap I_1,J\cap I_2\}$ of $\alpha\suborder{J}$. 
We thus obtain $\pit(J)=\pit(J\cap I_1)\cdot\pit(J\cap I_2)=a_1\cdot a_2$ 
and hence $\pit(\alpha)=\epis\big(u_1\:(a_1\cdot a_2)\:u_2\big)$. 
Lemma~\ref{lemma:restriction-pibase} then implies 
$\epis\big(u_1\:(a_1\cdot a_2)\:u_2\big)=\epis(u_1\:a_1)\cdot\epis(u_2\:a_2)$.
Similarly, Lemma~\ref{lemma:pibase-substitution} implies $\epis(u_1\:a_1)=\pit(I_1)$ 
and $\epis(u_2\:a_2)=\pit(I_2)$. Overall, we get $\epis(\alpha)=\pit(I_1)\cdot\pit(I_2)$.
\end{proof}

\smallskip
\begin{lemma}\label{lemma:equivalence-evaluation-omega}
Given a "word" $u$ with domain $\alpha$, an "evaluation tree" $\cT=(T,\pit)$ over $u$, and an 
$\omega$-"condensation" $I_0<I_1<I_2<\ldots$ of $\alpha$ such that $\pit(I_1)=\pit(I_2)=\ldots$
is an "idempotent", we have $\pit(\alpha)=\pit(I_0)\cdot\pit(I_1)^\tau$.
\end{lemma}

\begin{proof}
The proof is again by induction on $\cT$. Note that the case of $\cT$ consisting of a 
single "leaf" cannot happen. Let $C=\children(\alpha)$ be the top-level "condensation". 
We distinguish two cases depending on whether $C$ has a maximal element or not.

\AP
Suppose that $C$ has a maximal element, say $J_\max$, and $C\neq\{J_\max\}$
(the case where $C=\{J_\max\}$ can be considered as a degenerate case, which can
be dealt with by similar arguments). We can find a "condensation" $K_1<K_2$ of $\alpha$ 
that is coarser than $I_0<I_1<I_2<\ldots$ and such that $K_2\subseteq J_\max$. 
By Lemma~\ref{lemma:equivalence-evaluation-finite}, we have $\pit(\alpha)=\pit(K_1)\cdot\pit(K_2)$.
Moreover, since $K_1$ is the union of a finite sequence of "convex" subsets $I_0,I_1,\ldots,I_k$, 
by repeatedly applying Lemma~\ref{lemma:pibase-substitution}, we obtain 
$\pit(K_1) = \pit(I_0)\cdot\pit(I_1)\cdot\ldots\cdot\pit(I_k)
           = \pit(I_0)\cdot\pit(I_1)$ 
(the last equality follows from the fact that $\pit(I_1)=\pit(I_2)=\ldots$ 
is an "idempotent"). 
Finally, from the induction hypothesis (note that $\rank(\cT\subtree{K_2}) < \rank(\cT)$), 
we get $\pit(K_2)=\pit(I_1)^\tau$. We thus conclude that 
$\pit(\alpha) = \big(\pit(I_0)\cdot\pit(I_1)\big)\cdot\big(\pit(I_1)^\tau\big)
              = \pit(I_0)\cdot\pit(I_1)^\tau$.

\AP
If $C$ has no maximal element, then, using standard techniques and 
Ramsey's Theorem (Lemma~\ref{lemma:ramsey}), one can construct an 
$\omega$-"condensation" $J_0<K_1<J_1<K_2<J_2<\ldots$ of $\alpha$ such that:
\begin{itemize}
  \item $\{J_0\cup K_1,J_1\cup K_2,\ldots\}$ is coarser than $\{I_0,I_1,I_2,\ldots\}$,
  \item $\{J_0,K_1\cup J_1,K_2\cup J_2,\ldots\}$ is coarser than $C$,
  \item $\pit(K_1\cup J_1)=\pit(K_2\cup J_2)=\ldots$ is an "idempotent".
\end{itemize}
Let $\pit(C)$ be the "word" with domain $C$ where each position $H\in C$ is labelled 
by the "value" $\pit(H)$. By construction, we have $\pit(\alpha)=\epis(\pit(C))$. 
Moreover, since the "condensation" $\{J_0,K_1\cup J_1,K_2\cup J_2,\ldots\}$ is 
coarser than $C$, by repeatedly applying Lemma~\ref{lemma:pibase-substitution}, 
we obtain 
$\epis(\pit(C)) = \epis\big(\pit(J_0)~\pit(K_1\cup J_1)~\pit(K_2\cup J_2)~\ldots\big) 
               = \pit(J_0)\cdot\pit(K_1\cup J_1)^\tau$.
Similarly, since $\{J_0\cup K_1,J_1\cup K_2,\ldots\}$ is coarser than $\{I_0,I_1,I_2,\ldots\}$ 
and $\pit(I_1)=\pit(I_2)=\ldots$ is an idempotent, we have 
$\pit(J_0\cup K_1) = \pit(I_0)\cdot\pit(I_1)$ and 
$\pit(J_1\cup K_2) = \pit(J_2\cup K_3) = \ldots = \pit(I_1)$.
Thus, by Axioms \refaxiom{axiom:concatenation} and \refaxiom{axiom:omega}, we 
obtain $\pit(J_0)\cdot\pit(K_1\cup J_1)^\tau = \pit(I_0)\cdot\pit(I_1)^\tau$.
\end{proof}

\smallskip
\AP
We can gather all the results seen so far and prove the following corollary
(recall that an ordering is "scattered" if all "dense" "suborderings" of 
it are empty or singletons):

\begin{corollary}\label{cor:equivalence-evaluation-scattered}
Given a "word" $u$ with domain $\alpha$, an "evaluation tree" $\cT=(T,\pit)$ over $u$, 
a "scattered" "condensation" $C$ of $\alpha$, and an "evaluation tree" $\cT'=(T',\pit')$
over the "word" $\pit(C)=\prod_{I\in C}\pit(I)$ with domain $C$, we have $\pit(\alpha)=\pit'(C)$.
\end{corollary}

\begin{proof}
As a preliminary remark, note that since the "condensation" $C$ is "scattered", 
we have that, for every node $J$ in the "evaluation tree" $\cT'=(T',\pit')$, the 
"condensation" of $J$ induced by $\cT'$ is "scattered" as well. The proof is by 
induction on $\cT'$. 
If $\cT'$ consists of a single "node", then $\pit(C)$ is a singleton "word" of "value"
$\pit(\alpha)$ and hence the statement boils down to $\pit(\alpha)=\pit(\alpha)$.
Otherwise, let $D$ be the "childhood" of the "root" $C$ of $\cT'$.
From the induction hypothesis, we know that for every $J\in D$, 
$\pit'(J)=\pit\big(\bigcup J\big)$, where $\bigcup J$ denotes the union 
of all "convex" subsets of $J$ (recall that $J\subseteq C$). Moreover, 
if we denote by $\bigcup D$ the "condensation" of $\alpha$ obtained from 
the substitution of each element $J\in D$ by $\bigcup J$, we have
$$
  \pit'(C) =~ \epis\left(\prod_{J\in D}\pit'(J)\right) 
           =~ \epis\left(\prod_{J\in D}\pit\big(\bigcup J\big)\right) 
           =~ \epis\Big(\pit\big(\bigcup D\big)\Big).
$$ 
Note that the "condensation" $\bigcup D$ of $\alpha$ has the same "order type"
of the "condensation" $D$ of $C$, namely, it is either a finite "condensation",
an $\omega$-"condensation", or an $\omegaop$-"condensation". Therefore, 
using either Lemma~\ref{lemma:equivalence-evaluation-finite} or 
Lemma~\ref{lemma:equivalence-evaluation-omega} (or its symmetric variant), 
we obtain $\epis\big(\pit(\bigcup D)\big)=\pit(\alpha)$.
\end{proof}

\smallskip
\AP
It remains to consider the case of "dense" "condensations", which give rise to $\eta$-shuffles:

\begin{lemma}\label{lemma:equivalence-evaluation-shuffle}
Given a "word" $u$ with domain $\alpha$, an "evaluation tree" $\cT=(T,\pit)$ over $u$, and a 
"dense" "condensation" $C$ of $\alpha$ such that $\pit(C)=\prod_{I\in C}\pit(I)$ 
is isomorphic to a "word" of the form $a\:P\etapow\:b$, for some elements $a,b\in S\uplus\{\emptystr\}$ 
and some non-empty set $P\subseteq S$, we have $\pit(\alpha)=a\cdot P^\kappa\cdot b$.
\end{lemma}

\begin{proof}
We remark here that the proof works for any "condensation" $C$, independently of the 
form of the "word" $\pit(C)$. However, the use of the following technical arguments 
does only make sense when $C$ is a "dense" "condensation". We prove the lemma by 
induction on $\cT$. As in the proof of Lemma~\ref{lemma:equivalence-evaluation-omega}, 
the case of $\cT$ consisting of a single "node" cannot happen. Let $D=\children(\alpha)$ 
be the top-level "condensation" and let $E$ be the finest "condensation" that is coarser 
than or equal to both $C$ and $D$ (note that $E$ exists since "condensations" form a 
lattice structure with respect to the `coarser than' relation). Moreover, let $\sim$ 
be the "condensation" over the "condensed ordering" $C$ such that, for every $I,I'\in C$, 
$I\sim I'$ holds iff either $I=I'$ or there is $J\in D$ with $I\subseteq J$ and 
$I'\subseteq J$. This can naturally be seen as a "condensation" $C'$ over $\alpha$ 
which is at least as coarse as $C$: the classes of $C'$ are either the single 
classes of $C$ that are not contained in any class of $D$, or the unions of the 
classes of $C$ that are contained in the same class of $D$. Furthermore, it is 
easy to see that $E$ is at least as coarse as $C'$. Below, we disclose further
properties of the "condensations" $C$, $D$, $E$, and $C'$.

\AP
Let us consider a class $I\in C'$. Two cases can happen: either $I$ is included 
in some $J\in D$, and in this case $\pit(I)=\epis\big(\pit(C\suborder{I})\big)$ holds thanks
to the induction hypothesis, or $I$ belongs to $C$, and hence $\pit(I)=\epis\big(\pit(C\suborder{I})\big)$
follows trivially. We have just proved that
\begin{equation}\label{eq:shuffle-property1}
  \forall I\in C' \quad \pit(I) \;=\; \epis\big(\pit(C\suborder{I})\big).
\end{equation}

\AP
Now, let $I,I'$ be two distinct classes in $C'$. We claim that there exist $x\in I$ and $x'\in I'$ 
that are not equivalent for $D$, namely,
\begin{equation}\label{eq:shuffle-property2}
  \exists x\in I  \exists x'\in I'  \forall J\in D \quad x\nin J \vel x'\nin J.
\end{equation}
The proof of this property is by case distinction. If $I$ is contained in some $J\in D$ and $I'$ 
is contained in some $J'\in D$, then we necessarily have $J\neq J'$ (otherwise, we would 
have $I=I'$ by definition of $C'$) and hence Property~\refeq{eq:shuffle-property2} holds.
Otherwise, either $I$ is not contained in any class $J\in D$, or $I'$ is not contained
in any class $J\in D$. Without loss of generality, we assume that $I$ is not contained 
in any class $J\in D$. This means that there exists $J\in D$ such that $I\cap J\neq\emptyset$ 
and $I\setminus J\neq\emptyset$. Let us pick some $x'\in I'$. Clearly, $x'$ belongs to 
some $J'\in D$. Then either $J\cap J'=\emptyset$ or $J=J'$. In the first case, one 
chooses $x\in I\cap J$, while in the second case one chooses $x\in I\setminus J$.
This completes the proof of Property~\refeq{eq:shuffle-property2}.

\AP
From the above property, we can deduce the following:
\begin{equation}\label{eq:shuffle-property3}
\begin{array}{l}
  \text{\emph{If $I,I'\in C'$, $I<I'$, and $I,I'\subseteq K$ for some $K\in E$, then}} \\
  \text{\emph{there are only finitely many classes $I''\in C'$ between $I$ and $I'$.}}
\end{array}
\end{equation}
Indeed, suppose that the above property does not hold, namely, that there are
infinitely many classes $I''\in C'$ between $I$ and $I'$. In particular, we can 
find an $\omega$-sequence of classes $I_1,I_2,\ldots$ such that $I=I_1<I_2<\ldots<I'$ 
or $I<\ldots<I_2<I_1=I'$. We only consider the first case (the second case is symmetric). 
By applying Property~\refeq{eq:shuffle-property2} to the classes $I_1,I_2,\ldots$,
we can find some points $x_1\in I_1$, $x'_1\in I_2$, $x_2\in I_3$, $x'_2\in I_4$, \dots
such that, for all $i\ge1$, $x_i$ and $x'_i$ are not equivalent for $D$
(i.e., for all $J\in D$, $x_i\nin J$ or $x'_i\nin J$). Let $X$ be the set of
all points $x\in\alpha$, with $x<I_i$ for some $i\ge1$, and let $X'$ be 
the set of all points $x'\in\alpha$, with $x'>I_j$ for all $j\ge1$.
Since $D$ is a "condensation", we have that for all $x\in X$ and all $x'\in X'$,
$x$ and $x'$ are not equivalent for $D$. Moreover, by construction, all such 
points $x$ and $x'$ are not equivalent for $C'$, and hence neither for $C$
(recall that $C$ is finer than $C'$). Since $E$ is the defined as the finest 
"condensation" that is coarser than or equal to both $C$ and $D$ and since
$X\cup X'=\alpha$, it follows that there is no class $K\in E$ that intersects 
both $X$ and $X'$. In particular, since $I\subseteq X$ and $I'\subseteq X'$, 
it follows that there is no class $K\in E$ such that $I\subseteq K$ and 
$I'\subseteq K$, which is a contradiction. This completes the proof of 
Property~\refeq{eq:shuffle-property3}.

\AP
We prove the following last property:
\begin{equation}\label{eq:shuffle-property4}
  \forall K\in E \quad \pit(K) \:=\: \epis\big(\pit(C\suborder{K})\big).
\end{equation}
Let $K\in E$ and let $\cT'=({T',\pit'})$ be an "evaluation tree" over the "word" $\pit(C'\suborder{K})$ 
(such a tree exists according to Proposition~\ref{prop:existence-evaluation}). 
From Property~\refeq{eq:shuffle-property3} we know that the "condensation" of $C'\suborder{K}$ 
induced by the "evaluation tree" $\cT'$ is "scattered".
We can thus apply Corollary~\ref{cor:equivalence-evaluation-scattered} and obtain $\pit(K)=\pit'(C'\suborder{K})$.
Moreover, the value $\epis\big(\pit(C'\suborder{K})\big)$ is defined and hence, by 
Corollary~\ref{cor:evaluation-pibase}, $\pit'(C'\suborder{K})=\epis\big(\pit(C'\suborder{K})\big)$. 
By Property~\refeq{eq:shuffle-property1}, we obtain
$\pit(C'\suborder{K}) = \prod_{I\in C'\suborder{K}}\pit(I) 
             = \prod_{I\in C'\suborder{K}}\epis\big(\pit(C\suborder{I})\big)$.
Finally, from the properties of "condensation trees", we derive
$\epis\big(\pit(C'\suborder{K})\big) = \epis\big(\prod_{I\in C'\suborder{K}}\epis(\pit(C\suborder{I}))\big)
                                 = \pit(C\suborder{K})$.
This completes the proof of Property~\refeq{eq:shuffle-property4}.

\AP
Towards a conclusion, we consider an "evaluation tree" $\cT''=({T'',\pit''})$ over the 
"word" $\pit(E)$ (such a tree exists thanks to Proposition~\ref{prop:existence-evaluation}).
From Property~\refeq{eq:shuffle-property4} we know 
that $\pit(E) = \prod_{K\in E}\pit(K) = \prod_{K\in E}\epis\big(\pit(C\suborder{K})\big)$. 
Moreover, By Corollary~\ref{cor:evaluation-pibase}, we know that $\epis\big(\pit(C\suborder{K})\big)=\pit''(K)$ 
and hence $\prod_{K\in E}\epis\big(\pit(C\suborder{K})\big) = \pit''(E)$. 
Similarly, since $E$ is at least as coarse as $D$, Corollary~\ref{cor:evaluation-pibase} 
implies $\pit''(E)=\epis\big(\pit(D)\big)=\pit(\alpha)$. This completes the proof of the lemma.
\end{proof}

\smallskip
\begin{corollary}\label{cor:equivalence-evaluation}
Given a "word" $u$ with domain $\alpha$, an "evaluation tree" $\cT=({T,\pit})$ over $u$, 
a "condensation" $C$ of $\alpha$, and an "evaluation tree" $\cT'=({T',\pit'})$ over 
the "word" $\pit(C)=\prod_{I\in C}\pit(I)$ with domain $C$, we have $\pit(\alpha)=\pit'(C)$.
\end{corollary}

\begin{proof}
The proof is exactly the same as for Corollary~\ref{cor:equivalence-evaluation-scattered}, 
with the only difference that we do not use the assumption that the "condensation" $C$ 
is "scattered" and we use Lemma~\ref{lemma:equivalence-evaluation-shuffle} for treating 
the "nodes" $I'$ of $\cT'$ for which the "condensation" $\children(I')$ is "dense".
\end{proof}

\medskip
\AP
Finally, Proposition \ref{prop:equivalence-evaluation} follows easily from the previous corollary.

\begin{proof}[Proof of Proposition \ref{prop:equivalence-evaluation}]
Let $\cT=(T,\pit)$ and $\cT'=(T',\pit')$ be two "evaluation trees" over the same "word" $u$ 
with domain $\alpha$ and let $C$ be the finest "condensation" of $\alpha$, whose classes
are the singleton sets. Clearly, the "evaluation tree" $\cT'$ is isomorphic to an 
"evaluation tree" $\cT''=(T'',\pit'')$ over the "word" $\pit(C)=\prod_{I\in C}\pit(I)$ 
with domain $C$. Using Corollary \ref{cor:equivalence-evaluation} we immediately obtain that 
$\pit(\alpha)=\pit''(C)=\pit'(\alpha)$.
\end{proof}
  
\smallskip
\section{From monadic second-order logic 
         to \texorpdfstring{$\ostar$-algebras}{algebras}}\label{sec:logic-to-algebra}

Let us recall that ""monadic second-order"" (MSO) logic is the extension of first-order 
logic with set quantifiers. We assume the reader to have some familiarity with this 
logic, as well as with the technique used by B\"uchi to translate "MSO" formulas into 
equivalent automata. A good survey can be found in~\cite{languages_automata_logic}. 

\AP
Here, we show a relatively direct consequence of the results obtained in the previous
section, namely, that "MSO" formulas can be effectively translated to "$\countable$-algebras":

\begin{theorem}\label{theorem:mso-to-algebra}
The "MSO" definable $\countable$-languages are effectively "recognizable".
\end{theorem}

\noindent
Before turning to the proof of the above result, let us remark that we could 
have equally well used the composition method of Shelah for establishing 
Theorem~\ref{theorem:mso-to-algebra}. Indeed, given any "MSO" sentence $\psi$, 
one can construct effectively a "$\countable$-algebra" "recognizing" the 
"language" defined by $\psi$ \cite{composition_method_shelah}.

\AP
Our proof of Theorem~\ref{theorem:mso-to-algebra} follows B\"uchi's approach,
namely, we establish a number of closure properties for "recognizable"
$\countable$-languages. Then, each construction of the logic will be translated into 
an operation on languages. To disjunction corresponds union, to conjunction corresponds 
intersection, to negation corresponds complementation, etc. We assume the reader to be 
familiar with this approach (in particular the coding of the valuations of free variables).

\AP
The $\countable$-languages corresponding to the atomic predicates are easily shown to be 
"recognizable". Similarly, the language operations of intersection, union, and 
complementation can be implemented easily by means of classical algebraic operations:

\begin{lemma}\label{lemma:closure-under-boolean-operations}
The "recognizable" $\countable$-languages are effectively closed under intersection,
union, and complementation.
\end{lemma}

\begin{proof}
Given two "$\countable$-monoids" $(M_1,\pi_1)$ and $(M_2,\pi_2)$
"recognizing" the languages $L_1=h_1^{-1}(F_1)$ and $L_2=h_2^{-1}(F_2)$, 
respectively, with $F_1\subseteq M_1$, $F_2\subseteq M_2$, and $h_1$ and $h_2$
"morphisms" to $(M_1,\pi_1)$ and $(M_2,\pi_2)$, respectively, we have that 
$A^\countable\setminus L_1 = h_1^{-1}(M_1\setminus F_1)$,
$L_1\cap L_2 = (h_1\times h_2)^{-1}(F_1\times F_2)$, and
$L_1\cup L_2 = (h_1\times h_2)^{-1}\big((M_1\times M_2) \setminus (F_1\times F_2)\big)$.
In particular, the complement of $L_1$ is recognized by $(M_1,\pi_1)$, while the
union and the intersection of $L_1$ and $L_2$ are recognized by the product
"$\countable$-monoid" $(M_1\times M_2,\pi_1\times\pi_2)$.
Moreover, the latter product can be easily implemented at the level 
of "$\countable$-algebras": the operators of a "$\countable$-algebra" that 
corresponds to $(M_1\times M_2,\pi_1\times\pi_2)$ can be obtained by applying 
component-wise the operators of some "$\countable$-algebras" that correspond to 
$(M_1,\pi_1)$ and $(M_2,\pi_2)$.
\end{proof}

\smallskip
\AP
What remains to be proved is the closure under projection.
Formally, given a language~$L$ over some alphabet~$A$, and a mapping~$f$ from~$A$ 
to another alphabet~$B$, the ""projection"" of~$L$ via~$f$ is the language~$f(L)$, 
where $f$ is extended in a pointwise manner to "words" and "languages".
The logical operation of existential quantification corresponds, at the level of 
the defined languages, to a "projection". Hence, it remains to prove the following:

\begin{lemma}\label{lemma:closure-under-projection}
The "recognizable" $\countable$-languages are effectively closed under "projections".
\end{lemma}

\begin{proof}
We first describe the construction for a given "$\countable$-monoid" $(M,\pi)$,
and then show how to adapt the construction at the level of $\countable$-algebras.
The "projection" is implemented, as usual, by a powerset construction, namely, 
by providing the definition of a "generalized product" over~$\sP(M)$. 
Given two words~$u$ and~$U$ over~$M$ and~$\sP(M)$, respectively, we 
write~$u\in U$ if $\dom(u)=\dom(U)$ and $u(x)\in U(x)$ for all~$x\in\dom(U)$. 
We then define the mapping~$\tilde\pi$ from~$\sP(M)^\countable$ to~$\sP(M)$
by letting
\begin{align*}
  \tilde\pi(U) ~\eqdef~ \big\{\pi(u) \:\big\mid\: u\in  U\big\} 
  \qquad 
  \text{for all~$U\in\sP(M)^\countable$} \ .
\end{align*}
Let us show that~$\tilde\pi$ is "associative". 
Consider a word~$U$ over~$\sP(M)$ and a "condensation" $\sim$ of its domain. 
Then, 
$$
\begin{array}{rl}
  \tilde\pi(U) &=~ \big\{\pi(u) \:\big\mid\: u\in U\big\} \\[1ex]
               &=~ \big\{\pi\big(\prod\nolimits_{I\in\alpha\quotient{\sim}} \pi(u\suborder{I})\big)
                         ~\big\mid~ u\in U\big\} \\[1ex]
               &=~ \big\{\pi\big(\prod\nolimits_{I\in\alpha\quotient{\sim}}
                                 a_I\big)
                         ~\big\mid~ a_I\in\tilde\pi(U\suborder{I}) ~
                                    \text{for all~$I\in\alpha\quotient{\sim}$}\big\} \\[1ex]
               &=~ \tilde\pi\big(\prod\nolimits_{I\in\alpha\quotient{\sim}}
                                 \tilde\pi(U\suborder{I})\big)\ ,
\end{array}
$$
where the second equality is derived from the "associativity" of~$\pi$.
Hence~$(\sP(M),\tilde\pi)$ is a "$\countable$-monoid".

\AP
Next, we show that $(\sP(M),\tilde\pi)$ "recognizes" any "projection" 
of a "language" "recognized" by~$(M,\pi)$. Let 
let $L\subseteq A^\countable$ be a "language" "recognized" by~$(M,\pi)$
via some morphism $h:(A,\prod)\rightarrow(M,\pi)$,, namely, $L=h^{-1}(h(L))$, 
and let $f:A\rightarrow B$ be a "projection". 
We claim that the "projected@projection" "language" $L' = f(L)$ 
is "recognized" by $(\sP(M),\tilde\pi)$ via the morphism
$g = h\circ f^{-1} :(B,\prod)\rightarrow(\sP(M),\tilde\pi)$.
Clearly, we have $g^{-1}(g(L')) \supseteq L'$. For the 
opposite containment, consider a word $v\in g^{-1}(g(L'))$. 
By construction, there is a word $v'\in L'$ such that
$g(v')=g(v)$. Since $v'\in L'=f(L)$, there is $w'\in L$
such that $v'=f(w')$. Moreover, since $g(v')=g(v)$, there
is $w$ such that $f(w)=v$ and $h(w')=h(w)$. Finally, 
since $L=h^{-1}(h(L))$, we conclude that $w\in L$, and
hence $v=f(w)\in L'$.

\AP
Thanks to Lemma~\ref{lemma:semigroup-to-algebra} and Corollary~\ref{cor:algebra-to-monoid},
the construction of $(\sP(M),\tilde\pi)$ can be performed 
at the level of "$\countable$-algebras".
More precisely, any "$\countable$-algebra" $(M,1,\cdot,\tau,\tauop,\kappa)$
uniquely determines a "$\countable$-monoid" $(M,\pi)$, and from this, using the
powerset construction, one defines the "$\countable$-monoid" $(\sP(M),\tilde\pi)$,
and finally the "induced" "$\countable$-algebra" 
$(\sP(M),\{1\},\mathbin{\tilde{\cdot}},\tilde{\tau},\tilde{\tau}^\rev,\tilde{\kappa})$.
The crux in this line of arguments
is that the correspondence between the original "$\countable$-algebra" 
$(M,1,\cdot,\tau,\tauop,\kappa)$ and the final "$\countable$-algebra" 
$(\sP(M),\{1\},\mathbin{\tilde{\cdot}},\tilde{\tau},\tilde{\tau}^\rev,\tilde{\kappa})$
may be, a priori, not effective.
Below we explain why, in fact, this correspondence is effective, 
namely, we explain how each operator of the "$\countable$-algebra" 
$(\sP(M),\{1\},\mathbin{\tilde{\cdot}},\tilde{\tau},\tilde{\tau}^\rev,\tilde{\kappa})$
can be computed using the initial "$\countable$-algebra" 
$(M,1,\cdot,\tau,\tauop,\kappa)$ and some saturation process.

\AP
We give the intuition for constructing the most difficult and interesting 
operator $\tilde\kappa$, that is, for computing 
$P^{\tilde{\kappa}}=\tilde{\pi}(P\etapow)$ for any given non-empty subset 
$P=\{A_1,\ldots,A_k\}$ of $\sP(M)$, using the operators of the 
"$\countable$-algebra" $(M,1,\cdot,\tau,\tauop,\kappa)$.
We recall that
$P^{\tilde{\kappa}}=\{1\}$ if $A_1=\ldots=A_k=\{1\}$, otherwise 
$P^{\tilde{\kappa}}=\big(P\setminus\{1\}\big)^{\tilde{\kappa}}$.
We also recall that $P^{\tilde{\kappa}}$ must represent the set 
$\{\pi(u) \:\mid\: u\in U,~U\in P\etapow\}$
and hence the computation of $P^{\tilde{\kappa}}$ is very similar to 
that of~$\{\pi(u) \:\mid\: u\in A^\countable\}$, which was done in the proof 
of Theorem~\ref{thm:emptiness}. The difference here is that one needs to 
relativise $u$ to the words that belong to~$U$, for some $U\in P\etapow$.
This can be achieved by performing a product of the "$\countable$-algebra"
$(M,1,\cdot,\tau,\tauop,\kappa)$ with a "$\countable$-algebra" that "recognizes"
the single-word language $\{P\etapow\}$, and then applying the saturation process 
of Theorem~\ref{thm:emptiness} on the resulting "$\countable$-algebra".
\end{proof}
 
\smallskip
\section{From \texorpdfstring{$\ostar$-algebras}{algebras} 
         to monadic second-order logic}\label{sec:algebra-to-logic}

We have seen in the previous section that every "MSO" formula defines 
a "recognizable" $\countable$-language. In this section, we prove the converse.
Hereafter, we refer to the ""$\forall$-fragment"" (resp., ""$\exists$-fragment"")
of "MSO" logic as the set of formulas that start with a block of universal (resp., existential)
set quantifiers, followed by a first-order formula. Similarly, the 
""$\exists\forall$-fragment"" consists of formulas starting with a block
of existential set quantifiers followed by a formula of the "$\forall$-fragment".

\begin{theorem}\label{theorem:rec-to-mso}
The "recognizable" $\countable$-languages are effectively "MSO" definable.
Furthermore, such languages are definable in the "$\exists\forall$-fragment"
of "MSO" logic.
\end{theorem}

\AP
We fix for the remaining of the section a "morphism" $h$ from $"(A^\countable,\prod)"$
to a "$\countable$-monoid" $(M,\pi)$, with $M$ finite, and a subset~$F$ of~$M$.
Let also $1,\cdot,\tau,\tauop,\kappa$ be defined from~$\pi$. 
Our goal is to show that~$L = h^{-1}(F)$ is "MSO" definable. 
It is sufficient for this to show that for every $a\in M$, the language
$$
  \pi^{-1}(a) ~=~ \big\{w\in M^\countable \::\: \pi(w)=a\big\}\ ,
$$
can be defined by a suitable "MSO" sentence $\varphival{a}$. From this it will 
follow that that~$L=\bigcup_{a\in F}h^{-1}(a)$ is defined by the disjunction 
$\bigvee_{a\in F}\hatvarphival{a}$, where $\hatvarphival{a}$ is obtained from 
$\varphival{a}$ by replacing every occurrence of an atom $b(x)$, with $b\in M$, 
by $\bigvee_{c\,\in\, h^{-1}(b)\cap A}c(x)$.

\AP
A reasonable approach for defining $\pi^{-1}(a)$ is to use a formula which, 
given~$u\in M^\countable$, guesses some object that `witnesses' $\pi(u)=a$. 
The only objects that we have seen so far and that are able to ``witness''~$\pi(u)=a$ 
are "evaluation trees". Unfortunately, there is no way an "MSO" formula can guess 
an "evaluation tree", since their "height" cannot be bounded uniformly. That is why we 
use another kind of object for witnessing~$\pi(u)=a$: the so-called "Ramseian" "split", 
which is introduced just below.

\medskip
\subsection{Ramseian splits}\label{subsec:ramsey}

"Ramseian" "splits" are not directly applied to words, but to "additive labellings". 
Recall that an "additive labelling" $\sigma$ from a linear ordering $\alpha$ to a 
semigroup $(M,\cdot)$ (which, in our case, will be induced by the $\countable$-monoid $(M,\pi)$) 
is a function that maps any pair of elements $x<y$ from $\alpha$ to an element 
$\sigma(x,y)\in M$ in such a way that $\sigma(x,y)\cdot\sigma(y,z)=\sigma(x,z)$ 
for all $x<y<z$ in $\alpha$.

\AP
Given two positions~$x<y$ in a word~$u$, 
denote by~$[x,y)$ the interval $\{z \:\mid\: x\leq z< y\}$. 
Given a word~$u$ and two positions~$x<y$ in it, we define 
$""\additive{u}@\additive""(x,y)$ 
to be the element~$\pi\big(u\suborder{[x,y)}\big)$ 
of the "$\countable$-monoid" $(M,\pi)$.
Quite naturally, $\additive{u}$ is an "additive labelling", 
since for all~$x<y<z$, we have 
$\additive{u}(x,y) \cdot \additive{u}(y,z) 
 = \pi\big(u\suborder{[x,y)}\big)\cdot\pi\big(u\suborder{[y,z)}\big)
 = \pi\big(u\suborder{[x,y)}\:w\suborder{[y,z)}\big)
 = \pi\big(u\suborder{[x,z)}\big)
 = \additive{u}(x,z)$.

\begin{definition}\label{def:split}
A ""split"" of height $n$ of a "linear ordering" $\alpha$ is a function 
$g:\alpha\then\{1,\ldots,n\}$. Two elements $x,y\in\alpha$ are called 
""($k$-)neighbours"" iff $g(x)=g(y)=k$ and $g(z)\leq k$ for all 
$z\in \alpha\suborder{[x,y]\cup[y,x]}$ (note that the "neighbourhood" relation is an equivalence). 
The split $g$ is said to be ""Ramseian"" for an "additive labelling" $\sigma:\alpha\then M$ 
iff for all equivalence classes~$X\subseteq\alpha$ of the "neighbourhood" relation, 
there is an "idempotent" $e\in M$ such that~$\sigma(x,y)=e$ for all~$x<y$ in~$X$.
\end{definition}

\begin{theorem}[Colcombet \cite{factorization_forests_for_words_journal}]\label{th:split}
For every finite semigroup $(M,\cdot)$, every linear ordering $\alpha$, and every 
"additive labelling" $\sigma$ from $\alpha$ to $(M,\cdot)$, there is a "split" 
of $\alpha$ which is "Ramseian" for $\sigma$ and which has height at most $2|M|$.
\end{theorem}

\medskip
\subsection{Inductive construction of formulas}\label{subsec:inductiveconstruction}

Below we construct a formula that, given a word~$u$ of domain $\alpha$, 
guesses a "split" of $\alpha$ of height at most~$2|M|$, and uses it for 
representing the function that associates with each "convex" subset $I$ 
of~$\alpha$ the value~$\pi(u\suborder{I})$
in $M$. For the sake of simplicity, we fix a word~$u$ of domain~$\alpha$ and the corresponding
"additive labelling"~$\additive{u}$ over~$\alpha$ that is induced by~$u$. We remark, 
however, that all constructions that follow are uniform and do not depend on the chosen 
word~$u$.

\AP
In the following, we make extensive use of properties, functions, sets that 
are first-order definable from other parameters. For instance, when we say that 
a set~$X$ is ""first-order definable"" from some variables $\bar{Y}$, we mean
that there exists a first-order formula~$\xi(x,\bar{Y})$ that describes the
membership of $x$ in $X$ on the basis of $\bar{Y}$, that is, $x\in X$ iff 
$\xi(x,\bar{Y})$ holds on the given interpretation of $x$ and $\bar{Y}$.
In practice, this means that it is never necessary to quantify over~$X$ 
for defining properties concerning~$X$: it is sufficient to replace each 
predicate~$x\in X$ by the corresponding formula~$\xi(x,\bar{Y})$. This 
remark is crucial for understanding why the construction we provide 
yields a formula in the "$\exists\forall$-fragment" of "MSO" logic.

\medskip
\AP
Recall that we aim at constructing, for each $a\in M$, a sentence $\varphival{a}$
that holds over the word~$u$ iff $\pi(u)=a$. The starting point is to guess: 
\begin{enumerate}
  \item a "split" $g$ of~$\alpha$ of height at most~$2|M|$, and;
  \item a function~$f$ mapping each position~$x\in\alpha$ to an "idempotent" $f(x)\in M$.
\end{enumerate}
The intention is that a choice of $g$ and $f$ is good when $g$ is a "Ramseian" "split"
for~$\additive{u}$ and the function~$f$ maps each position~$x$ to the "idempotent" 
$f(x)$ that arises when the "neighbourhood" class of~$x$ is considered 
(cf.~Definition~\ref{def:split}). 
In this a case, by a slight abuse of terminology, we say that 
$(g,f)$ is a ""Ramseian pair"". 

\AP
Observe that neither~$g$ nor~$f$ can be represented by a single monadic 
variable. However, since both~$g$ and~$f$ are functions from~$\alpha$ to
sets of bounded size ($2|M|$ for~$g$, and $|M|$ for~$f$), one can guess 
them using a fixed number of monadic variables. This kind of encoding is 
quite standard, and from now on we shall use explicitly the mappings~$g$ 
and~$f$ in our formulas, rather than their encodings.

\AP
Knowing a "Ramseian pair" $(g,f)$ is an advance towards computing the value of a word. 
Indeed, "Ramseian" "splits" can be used as ``accelerating structures'' in the sense 
that every computation of~$\pi(u\suborder{I})$ for some convex~$I$ becomes significantly 
easier when a "Ramseian" "split" is known, namely, it becomes "first-order definable" 
in terms of the "Ramseian" "split". This is formalized by the following lemma.

\begin{lemma}\label{lemma:eval-from-split}
Given~$a\in M$, one can construct a first-order formula~$""\eval_a""(g,f,X)$ 
such that for every "convex" subset $I$ of $\alpha$:
\begin{itemize}
  \item if $(g,f)$ is "Ramseian", then $\eval_a(g,f,I)$ holds iff~$\pi(u\suborder{I})=a$,
  \item if both $\eval_a(g,f,I)$ and~$\eval_b(g,f,I)$ hold, then~$a=b$.
\end{itemize}
\end{lemma}

\begin{proof}
As already mentioned, we encode both functions~$g$ and~$f$ by tuples of monadic predicates.
This allows us to use shorthands such as $g(x)=k$, where $x$ is a first-order variable
and $1\le k\le 2|M|$, for claiming that the point~$x$ of the underlying word~$u$ is 
mapped via $g$ to the number $k$. Similarly, we encode the convex subset~$I$ of~$\alpha$ 
by a monadic predicate and we write $x\in I$ as a shorthand for a formula that states 
that the point~$x$ belongs to~$I$. 

\AP
We assume from now that $(g,f)$ is "Ramseian". Under this assumption, it will be clear 
that the constructed formulas will satisfy the desired properties. We remark, however, 
that the following definitions make sense also in the case when~$(g,f)$ is not "Ramseian",
in which case only the second condition of the lemma will be guaranteed. 

\AP
Given a "convex" $I$, we denote by~$""\level""(g,I)$ the maximal value of~$g(x)$
for~$x$ ranging over~$I$. Of course, the properties~$\level(g,I)=k$ and~$\level(g,I)\leq k$
are "first-order definable" in terms of~$g$ and~$I$.

\AP
We will construct by induction on $k\in\{0,1,\ldots,2|M|\}$ a {\sl partial}
function~$\eval^k$ that maps some triples $(g,f,I)$ to elements $\eval^k(g,f,I)\in M$
in such a way that the following properties hold:
\begin{itemize}
  \item $\eval^k(g,f,I)=a$ is "definable by a first-order formula", 
        say~$\eval^k_a(g,f,I)$, for each~$a\in M$,
  \item $\eval^k(g,f,I)$ is defined iff $\level(g,I)\leq k$, and in 
        this case it coincides with $\pi(u\suborder{I})$ 
        (provided $(g,f)$ is "Ramseian").
\end{itemize} 
The base case is when~$k=0$. In this case, we define $\eval^k(g,f,I)$ to 
be the "neutral element" $1$ when $I=\emptyset$, and we let $\eval^k(g,f,I)$ 
be undefined when $I\neq\emptyset$. Of course, this is "first-order definable"
and satisfies the expected induction hypothesis.

\AP
Let us now construct the partial function~$\eval^k(g,f,I)$ for any~$k\geq 1$.
First, if~$\level(g,I)<k$, then one simply outputs~$\eval^{k-1}(g,f,I)$.
Otherwise, the "convex" subset $I$ can be uniquely partitioned
into~$X<J<Y$ in such a way that~$X\cup J\cup Y=I$ and~$J$ 
is the minimal "convex" subset containing~$I\cap g^{-1}(k)$. Note that 
the sets~$X$, $J$, and $Y$ are "first-order definable" in the parameters~$I$ and $g$,
that is, membership of any point~$x$ in $X$ (resp., $J$, $Y$) is characterized 
by a first-order formula in the variables $x$, $I$, and $g$.
Furthermore, fix~$e$ to be~$f(x)$ for some~$x\in I\cap g^{-1}(k)$.
From the assumption that~$I$ has "level"~$k$ for~$g$, we know that 
all elements in $I\cap g^{-1}(k)$ are "neighbours". In particular,
the fact that $g$ is a "Ramseian" "split" for $\additive{u}$ means 
that~$\additive{u}(x,y)=e$ for all~$x<y$ chosen in $I\cap g^{-1}(k)$.
The mapping~$\eval^k(g,f,I)$ is defined below by a case distinction 
(we remark that the following definitions are not symmetric with respect 
to the underlying order, and this reflects the asymmetry occurring in the 
definition of~$\additive{u}$, that is, $\additive{u}(x,y)=\pi(u\suborder{[x,y)})$ 
for all $x<y\in\alpha$):
\begin{enumerate}
  \item if $J$ is a singleton $\{x\}$, then 
        $$
          \eval^k(g,f,I) ~=~ \eval^{k-1}(g,f,X)~\cdot u(x)~\cdot~\eval^{k-1}(g,f,Y)\ ,
        $$
  \item if $J$ has distinct minimal and maximal elements 
        and~$y=\max(J)$, then
        $$
          \eval^k(g,f,I) ~=~ \eval^{k-1}(g,f,X)~\cdot~e~\cdot~u(y)~\cdot~\eval^{k-1}(g,f,Y)\ ,
        $$
  \item if $J$ has no minimal element but has a maximal element~$y$, then
        $$
          \eval^k(g,f,I) ~=~ \eval^{k-1}(g,f,X)~\cdot~e^\tauop~\cdot~u(y)~\cdot~\eval^{k-1}(g,f,Y)\ ,
          \mspace{-16mu}
        $$
  \item if $J$ has a minimal element but no maximal element, then
        $$
          \eval^k(g,f,I) ~=~ \eval^{k-1}(g,f,X)~\cdot~e^\tau~\cdot~\eval^{k-1}(g,f,Y)\ ,
        $$
  \item if $J$ has no minimal element and no maximal element, then
        $$
          \eval^k(g,f,I) ~=~ \eval^{k-1}(g,f,X)~\cdot~e^\tauop\cdot e^\tau~\cdot~\eval^{k-1}(g,f,Y)\ .
        $$
\end{enumerate}
One easily checks that the function $\eval^k$ can be defined by first-order formulas
of the form $\eval^k_a(g,f,I)$, with $a\in M$. It is also easy to see that 
if~$(g,f)$ is "Ramseian@Ramseian pair" and~$\level(g,I)\leq k$, then $\eval^k(g,f,I)$ 
coincides with~$\pi(u\suborder{I})$.

\AP
At this step, the first conclusion of the lemma is already satisfied by the 
first-order formulas~$\eval_a^{2|M|}(g,f,I)$.
The second point, however, is false in general. Indeed, we did not pay attention 
so far on what the formulas compute in the case where~$(g,f)$ is not 
"Ramseian@Ramseian pair". 
In particular, it can happen that both formulas $\eval_a^{2|M|}(g,f,I)$ and 
$\eval_b^{2|M|}(g,f,I)$ hold for distinct elements $a,b\in M$. However, this 
can be easily fixed using the following formula:
$$
  \eval_a(g,f,I) ~\eqdef~ \eval_a^{2|M|}(g,f,I) ~\et~ 
                          \bigwedge\limits_{b\neq a} \neg \eval_b^{2|M|}(g,f,I)\ .
$$
This formula ensures the second property of the lemma by construction, and 
behaves like~$\eval_s$ whenever~$(g,f)$ is "Ramseian@Ramseian pair".
\end{proof}

\smallskip
\AP
The formulas constructed in Lemma \ref{lemma:eval-from-split} can be seen as defining a 
partial function~$\eval$ that maps~$g,f,I$ to some element~$a\in M$ (the second item 
in the lemma enforces that there is no ambiguity about the value, namely, that this is a 
function and not a relation). Hereafter, we simply use the notation~$\eval(g,f,I)$ as 
if it were a function.

\medskip
\AP
One needs now to enforce that~$\eval(g,f,I)$ coincides with~$\pi(u\suborder{I})$, 
even without assuming that $(g,f)$ is "Ramseian@Ramseian pair". 
For this, one uses "condensations". 
A priori, a "condensation" is not representable by monadic variables, since it 
is a binary relation. However, any set~$X\subseteq\alpha$ naturally defines 
the relation~$\approx_X$ such that~$x\approx_X y$ iff either~$[x,y]\subseteq X$, 
or~$[x,y]\cap X=\emptyset$. It is easy to check that this relation is a "condensation". 
A form of converse result also holds:

\begin{lemma}\label{lemma:condensation-encoding}
For every "condensation" $\sim$, there is~$X$ such that~$\sim$ and $\approx_X$ coincide.
\end{lemma}

\begin{proof}
It is easy to see that, given a "linear ordering"~$\beta$, there exists a subset~$Y$ 
of~$\beta$ such that for all~$x<y$ in~$\beta$, $[x,y]$ intersects both~$Y$ and 
its complement~$\beta\setminus Y$: indeed, one can first prove this for 
"scattered" "linear orderings" and for "dense" "linear orderings", and then combine 
the results for these subcases using the fact that every "linear ordering" is 
a "dense" "sum" of non-empty "scattered" "linear orderings" \cite{linear_orderings}.

\AP
The lemma follows easily from the above argument: consider~$Y$ obtained from 
the claim above applied to the "condensed ordering" $\beta=\alpha\quotient{\sim}$. 
We construct the desired set~$X$ in such a way that it contains the elements of 
the equivalence classes of~$\sim$ that belong to~$Y$, i.e., 
$X=\{x \:\mid\: [x]_\sim\in Y\}$.
It is easy to see that~$x\sim y$ iff~$x\approx_X y$.
\end{proof}

\smallskip
\AP
Lemma~\ref{lemma:condensation-encoding} tells us that it is possible to work with "condensations"
as if they were monadic variables. In particular, in the sequel we use variables for "condensations"
and we tacitly assume that they are encoded by the sets obtained from Lemma~\ref{lemma:condensation-encoding}.

\AP
Given a "convex" subset $I$ of $\alpha$ and some "condensation" $\sim$ of $\alpha\suborder{I}$, 
we denote by $""u[I,\sim]""$ the word with domain~$\beta=(\alpha\suborder{I})\quotient{\sim}$ 
in which every $\sim$-equivalence class~$J$ is labelled by~$\eval(g,f,J)$.
One can easily define a formula $""\consistency""(g,f)$ that checks that, 
for all "convex" subsets $I$ and all "condensations" $\sim$ of $\alpha\suborder{I}$, 
the following conditions hold:
\begin{condlist}
  \item if~$I$ is a singleton~$\{x\}$, 
        then $\eval(g,f,I)=u(x)$,\label{consistency1}
  \item if~$"u[I,\sim]"=a\:b$ for some~$a,b\in M$, 
        then $\eval(g,f,I)=a\cdot b$,\label{consistency2}
  \item if~$"u[I,\sim]"=e\omegapow$ for some idempotent~$e\in M$, 
        then $\eval(g,f,I)=e^\tau$,\label{consistency3}
  \item if~$"u[I,\sim]"=e\omegaoppow$ for some idempotent~$e\in M$, 
        then $\eval(g,f,I)=e^\tauop$,\label{consistency4}
  \item if~$"u[I,\sim]"=P\etapow$ for some non-empty set~$P\subseteq M$, 
        then $\eval(g,f,I)=P^\kappa$.\label{consistency5}
\end{condlist}

\AP
For some fixed~$I$ and~$\sim$, the above tests require access to the elements 
$"u[I,\sim]"(J)$, where~$J$ is a~$\sim$-equivalence class of~$\alpha\suborder{I}$. 
Since the property of $\sim$-equivalence for two positions $x,y\in\alpha\suborder{I}$ 
is "first-order definable", we know that for every position $x\in\alpha\suborder{I}$, 
the element $\eval(g,f,[x]_\sim)$ 
is "first-order definable" from~$x$. This shows that the above properties 
can be expressed by first-order formulas and hence $\consistency(g,f)$ 
is a formula in the "$\forall$-fragment" of "MSO" logic.

\medskip
\AP
The last key argument is to show how the `local' consistency constraints
\refcond{consistency1}--\refcond{consistency5} imply a `global' consistency 
property. This is done by the following lemma.

\begin{lemma}\label{lemma:algebra-to-mso-correctness}
If $\consistency(g,f)$ holds, then $\eval(g,f,I)=\pi(u\suborder{I})$ 
for all "convex" subsets $I$ of~$\alpha$.
\end{lemma}

\begin{proof}
Recall that, given a "convex" subset $I$ of~$\alpha$ and a "condensation" $\sim$
of $\alpha\suborder{I}$, $"u[I,\sim]"$ is the word with domain~$\beta=(\alpha\suborder{I})\quotient{\sim}$ 
in which every $\sim$-equivalence class~$J$ is labelled by~$\eval(g,f,J)$.
Suppose that $\consistency(g,f)$ holds, namely, that for all "convex" subsets $I$ 
of~$\alpha$ and all "condensations" $\sim$ of~$\alpha\suborder{I}$, the conditions 
\refcond{consistency1}--\refcond{consistency5} are satisfied.

\AP
To show that $\eval(g,f,I)=\pi(u\suborder{I})$ for all "convex" subsets $I$, we use 
again "evaluation trees". Precisely, we fix a "convex" subset $I$ of~$\alpha$ 
and an "evaluation tree" $\cT=(T,\pit)$ over the word $u\suborder{I}$ (the "evaluation tree" 
exists thanks to Proposition \ref{prop:existence-evaluation}), and we prove, by an 
induction on $\cT$, that 
$$
  \eval(g,f,I) ~=~ \pit(I)\ .
$$ 
Since~$\pit(I)=\pi(u\suborder{I})$ (by Proposition~\ref{prop:equivalence-evaluation}), 
it follows that $\eval(g,f,I)=\pi(u\suborder{I})$.

\AP
If~$\cT$ consists of a single "leaf", then~$I$ is a singleton of the form $\{x\}$.
Condition~\refcond{consistency1} then immediately implies $\eval(g,f,I)=u(x)=\pit(I)$.

\AP
If the "root" of~$\cT$ is not a "leaf", then we let~$\sim$ be the "condensation"
of~$\alpha\suborder{I}$ induced by the "children" of the "root" of $\cT$ and we let
$\beta=(\alpha\suborder{I})\quotient{\sim}$ be the corresponding "condensed ordering" 
(formally, $\beta=\children(I)$). Note that for every class~$J\in\beta$, 
$\cT\subtree{J}$ is a "subtree" of $\cT$. 
From the induction hypothesis on the "evaluation tree" $\cT\subtree{J}$, we have 
$\eval(g,f,J)=\pit(J)$ for all~$J\in\beta$. Moreover, we know from the 
definition of $"u[I,\sim]"$ that $"u[I,\sim]"(J)=\eval(g,f,J)$, for all $J\in\beta$, 
and hence $"u[I,\sim]"$ is isomorphic to the word $\prod_{J\in\beta}\pit(J)$. 
We also know from the definition of $\cT$ that the image under $\pis$ of the 
word $\prod_{J\in\beta}\pit(J)$ is defined. From this we derive that 
$\prod_{J\in\beta}\pit(J)$ is isomorphic to one of the following words:
\begin{enumerate}
  \item a word $a\,b$, for some~$a,b\in M$,
  \item an $\omega$-word $e\omegapow$, for some "idempotent" $e\in M$,
  \item an $\omegaop$-word $e\omegaoppow$, for some "idempotent" $e\in M$,
  \item a shuffle $P\etapow$, for some non-empty subset $P$ of $M$.
\end{enumerate}
We only analyse the first two cases (the remaining cases are all similar).

\AP
If the word $\prod_{J\in\beta}\pit(J)$ is of the form $a\,b$, with $a,b\in M$,
then we let $J_1$ and $J_2$, with $J_1<J_2$, be the two positions in it 
(recall that these are $\sim$-equivalence classes for $\alpha\suborder{I}$).
Thanks to the inductive hypothesis, we have
$\eval(g,f,J_1) = "u[I,\sim]"(J_1) = \pit(J_1) = a$ and
$\eval(g,f,J_2) = "u[I,\sim]"(J_2) = \pit(J_2) = b$. 
From Condition \refcond{consistency2}, using the "condensation" $\sim$, 
we derive $\eval(g,f,I) = \eval(g,f,J_1\cup J_2) = a\cdot b$, and from 
this we easily conclude that
$$
  \eval(g,f,I) ~=~ a\cdot b ~=~ \pis(a\,b) ~=~ \pis\big(\pit(J_1)\:\pit(J_2)\big) ~=~ \pit(I).
$$

\AP
Let us now consider the case where $\prod_{J\in\beta}\pit(J)$ is an $\omega$-word of 
the form $e\omegapow$, for some "idempotent" $e\in M$. 
We denote by $J_1<J_2<\ldots$ the positions in
$\prod_{J\in\beta}\pit(J)$ (recall that these are $\sim$-equivalence classes for $\alpha\suborder{I}$). 
As in the previous case, we know from the inductive hypothesis that
$\eval(g,f,J_i) = u{[I,\sim]}(J_i) = \pit(J_i) = e$ for all $i=1,2,\ldots$. 
We know from Condition \refcond{consistency3} that
$\eval(g,f,I) = \eval(g,f,J_1\cup J_2\cup\ldots) = e^\tau$.
Finally, we derive
$$
  \eval(g,f,I) ~=~ e^\tau 
               ~=~ \pis(e\omegapow) 
               ~=~ \pis\Big(\prod_{J\in\beta}\pit(J)\Big) 
               ~=~ \gamma(I) \ .
$$
\end{proof}

\medskip
\AP
We conclude the section by showing how Lemma \ref{lemma:algebra-to-mso-correctness} 
implies Theorem~\ref{theorem:rec-to-mso}. We claim that, given~$a\in M$, the language 
$\pi^{-1}(a)$ is defined by the following sentence in the "$\exists\forall$-fragment"
of "MSO" logic:
$$
  ""\varphival{a}@\varphival"" ~\eqdef~ \exists g.~ \exists f.~
                             \consistency(g,f) ~\et~ \eval(g,f,\alpha)=a \ .
$$
Let~$\pi(u)=a$. One can find a "Ramseian pair" $(g,f)$ using Theorem~\ref{th:split}. 
Lemma~\ref{lemma:eval-from-split} then implies~$\pi(u\suborder{I})=\eval(g,f,I)$ 
for all "convex" subsets $I$. Since $\pi$ is a "product", the constraints
\refcond{consistency1}--\refcond{consistency5} are satisfied 
and $\consistency(g,f)$ holds. This proves that $\varphival{a}$ holds. 
Conversely, if~$\varphival{a}$ holds, then $\consistency(g,f)$ holds 
for some~$(g,f)$. Lemma~\ref{lemma:algebra-to-mso-correctness} then implies
$$
  \pi(u) ~=~ \pi(u\suborder{\alpha}) ~=~ \eval(g,f,\alpha) ~=~ a\ .
$$
 
\smallskip
\section{Applications}\label{sec:applications}

In this section we present consequences of our results.

\medskip
\subsection{Collapse of the quantifier hierarchy}\label{subsec:collapse}

A first consequence of Theorems~\ref{theorem:mso-to-algebra} and~\ref{theorem:rec-to-mso} 
is that the hierarchy of monadic quantifier alternation for "MSO" logic interpreted 
over countable words collapses to its "$\exists\forall$-fragment".
Clearly, since "MSO" logic is closed under complementation, it also collapses 
to its "$\forall\exists$-fragment":

\begin{corollary}\label{cor:collapse}
Every $\countable$-language definable in "MSO" logic can be equally defined 
in the "$\exists\forall$-fragment" and in the "$\forall\exists$-fragment".
\end{corollary}

\noindent
Moreover, the collapse result is optimal, in the sense that there exist 
"MSO" definable languages that are not definable in the "$\exists$-fragment":

\begin{proposition}\label{prop:separation}
The language $L_\forall$ of countable "scattered" words over the singleton 
alphabet $\{a\}$ cannot be defined in the "$\exists$-fragment" of "MSO" logic.
\end{proposition}

\begin{proof}
We first recall a folklore result that shows that the language $L_\forall$ cannot 
be defined in first-order logic. 
The argument is based on Ehrenfeucht-Fra\"iss\'e games (we refer the reader 
to \cite{ehrenfeucht_fraisse_games,hodges1993model,linear_orderings} 
for basic knowledge on these games). 
One begins by fixing a number $n\in\bbN$ and suitable words $w\in L_\forall$ 
and $w'\nin L_\forall$, which may depend on $n$.
One then considers $n$ rounds of the Ehrenfeucht-Fra\"iss\'e game over $w$
and $w'$, where two players, called Spoiler and Duplicator, alternatively 
mark positions in $w$ and $w'$ inducing partial isomorphisms.
More precisely, at each round $k=1,\ldots,n$, Spoiler marks a position
in one of the two words, say either $x_k\in \dom(w)$ or $y_k\in\dom(w')$ 
-- intuitively this corresponds to quantifying existentially or universally 
over $w$. Duplicator responds by choosing a corresponding position in the other 
structure, say either $y_k\in\dom(w')$ or $x_k\in\dom(w)$. The responses of 
Duplicator must enforce an isomorphism between the induced substructures 
$w|_{\{x_1,\ldots,x_k\}}$ and $w'|_{\{y_1,\ldots,y_k\}}$. If Duplicator
cannot move while preserving the invariant, he loses the game. 
If he survives $n$ rounds, he wins. We know from Fra\"iss\'e's Theorem 
that Duplicator can win the $n$-round game if, and only if, $w$ and $w'$ 
cannot be distinguished by any formula of first-order logic with $n$ nested 
quantifiers -- in particular, if this happens for arbitrarily large
$n\in\bbN$, then $L_\forall$ cannot be defined in first-order logic. 

\AP
Below, we show that, for all $n\in\bbN$, Duplicator has a strategy 
to survive $n$ rounds of the Ehrenfeucht-Fra\"iss\'e game induced 
by the words 
$$
  w \:\eqdef\: a\omegapow \,\in L_\forall 
  \qquad\text{and}\qquad
  w' \:\eqdef\: a\omegapow\,(a\omegaoppow\,a\omegapow)\etapow \:\nin L_\forall \ .
$$
Without loss of generality, we can assume that during the first round 
of the game the left endpoints of $w$ and $w'$ are marked. For the 
subsequent rounds, the strategy of Duplicator will enforce the following invariant: 
if the distance between two positions $x_i,x_j$ that are marked in $w$ at rounds 
$j<i$ is less than $2^{n-i}$, then so is the distance between the corresponding 
positions $y_i,y_j$ that are marked in $w'$, and vice versa. 
On the other hand, if at round $i$ Spoiler picks a position $x_i$ in $w$ that 
is at distance at least $2^{n-i}$ from all previously marked positions, then, 
Duplicator can responds by picking a position $y_i$ inside a factor 
$a\omegaoppow\,a\omegapow$ of $w'$ that has no marked positions, thus 
guaranteeing that $y_i$ is at distance at least $2^{n-i}$ from all 
other marked positions. This strategy guarantees that Duplicator  
survives at least $n$ rounds of the game. The fact that winning
strategies for Duplicator exist for all $n\in\bbN$, proves that 
$L_\forall$ is not definable in first-order logic.

\AP
Now, it is straightforward to generalize the above argument to show that, 
for every first-order formula $\varphi$ and every pair of finite words 
$u,v$ over a finite alphabet, the following implication holds:
\begin{align*}
  u\,v\omegapow \sat \varphi
  \qquad\text{implies}\qquad
  u\,v\omegapow\,(v\omegaoppow\,v\omegapow)\etapow \sat \varphi \ .
  \tag{$\star$}
\end{align*}
We can use this result to show that the language $L_\forall$ cannot be defined in 
the "$\exists$-fragment" of "MSO" logic. Suppose, by way of contradiction, that there is a sentence 
$\psi = \exists \bar X \: \varphi(\bar X)$ that defines $L_\forall$, where $\varphi$ 
is a first-order formula with free variables among $\bar X=X_1,\ldots,X_m$. Since 
$a\omegapow\in L_\forall$, we know that $\varphi$ is satisfied by an interpretation of 
the free variables $\bar X$, and that this interpretation can be encoded by an 
$\omega$-word $w$ over the alphabet $\{a\}\times\{0,1\}^m$. 
By B\"uchi's result (or, equally, by Theorem \ref{thm:emptiness}), we can assume, 
again without loss of generality, that $w$ is ultimately periodic, namely, of the 
form $u\,v\omegapow$, for some finite words $u,v$. 
By the indistinguishability result in ($\star$), we know that $\varphi$ is 
also satisfied by $u\,v\omegapow\,(v\omegaoppow\,v\omegapow)\etapow$.
It follows that $a\omegapow\,(a\omegaoppow\,a\omegapow)\etapow$ is a model of $\psi$.
However, the latter word does not belong to $L_\forall$, and this 
contradicts the fact that $\psi$ defines $L_\forall$.
\end{proof}

\medskip
\subsection{Definability with the cuts at the background}\label{subsec:cuts}

In \cite{definability_with_reals_in_the_background} Gurevich and Rabinovich raised
and left open the following question: given any "MSO" formula $\varphi(X_1,\ldots,X_m)$, 
does there exist another "MSO" formula $\tilde\varphi(X_1,\ldots,X_m)$ such that, for all 
sets of rational numbers $A_1,\ldots,A_m$,
$$
  \phantom{?} \qquad 
  (\bbR,<) \sat \varphi(A_1,\ldots,A_m)
  \qquad\text{iff}\qquad
  (\bbQ,<) \sat \tilde\varphi(A_1,\ldots,A_m) 
  \qquad  ?
$$
In other words, they considered question of whether the ability to use all points of 
the real line does give more expressive power for stating properties of predicates 
over the rational line -- Gurevich and Rabinovich use the suggestive terminology that 
the formula $\varphi$ has access to the reals `\emph{at the background}'.
Note that here we implicitly use the fact that there is a fixed natural 
embedding of $(\bbQ,<)$ into $(\bbR,<)$.

\AP
Gurevich and Rabinovich answered positively the analogous question where the
rational line is replaced by the order of the natural numbers:

\begin{theorem}[\cite{definability_with_reals_in_the_background}]\label{theorem:definability-naturals-background}
For every "MSO" formula $\varphi(X_1,\ldots,X_m)$, there is an 
"MSO" formula $\tilde\varphi(X_1,\ldots,X_m)$ such that, for all 
sets $A_1,\ldots,A_m\subseteq\bbN$,
$$
  (\bbR,<) \sat \varphi(A_1,\ldots,A_m)
  \qquad\text{iff}\qquad
  (\bbN,<) \sat \tilde\varphi(A_1,\ldots,A_m) \ .
$$
\end{theorem}

\AP
We will not enter the details of this result, which is superseded by what follows. 
However, already in this case an interesting phenomenon occurs: the existence of 
the formula $\tilde\varphi$ is inherently non-effective, and this holds even if 
$\tilde\varphi$ is allowed to use extra predicates with a decidable "MSO" theory:

\begin{theorem}[\cite{definability_with_reals_in_the_background}]\label{theorem:undecidability-reals}
Let $\bar B = B_1,\ldots,B_n\subseteq\bbN$ be a tuple of monadic predicates such that 
$(\bbN,<,\bar B)$ has a decidable "MSO" theory. There is no algorithm that transforms 
an "MSO" formula $\varphi(X_1,\ldots,X_m)$ to an "MSO" formula $\tilde\varphi(X_1,\ldots,X_m)$ 
such that
$$
  (\bbR,<) \sat \varphi(A_1,\ldots,A_m)
  \qquad\text{iff}\qquad
  (\bbN,<,\bar B) \sat \tilde\varphi(A_1,\ldots,A_m) \ .
$$
\end{theorem}

\begin{proof}
Assume that such an algorithm exists, and consider a generic "MSO" sentence $\varphi$.
We can apply the algorithm to $\varphi$ to obtain a sentence $\tilde\varphi$ such
that $(\bbR,<)\sat\varphi$ iff $(\bbN,<,\bar B)\sat\tilde\varphi$.
Since the "MSO" theory of $(\bbN,<,\bar B)$ is decidable, we could decide
the "MSO" theory of $(\bbR,<)$. 
However, in \cite{composition_method_shelah,monadic_theory_of_order_and_topology} 
it has been shown that "MSO" theory of the real line is undecidable.
\end{proof}

\smallskip
\AP
Despite the inherent difficulty due to the non-effectiveness of the transformation,
we are able to answer positively the question raised by Gurevich and Rabinovich.

\AP
We begin by describing more precisely the relationship between the rational line 
and the real line. In fact, for technical reasons, it is convenient to work, rather
than on the real line, on a larger structure that is obtained by completing
the rational line with all "Dedekind cuts".

\begin{definition}\label{def:cut}
A ""(Dedekind) cut"" of a linear ordering $\alpha$ is a 
subset $E$ of $\alpha$ such that $\alpha|_E$ is a prefix of $\alpha$. 
\end{definition}

\AP
The "cuts" of $\alpha$ are naturally order by the containment relation,
that is, for all "cuts" $E,F$, we have $E<F$ iff $E\subsetneq F$. 
A "cut" is ""extremal"" if it is empty or contains all elements 
of the linear order $\alpha$.
"Cuts" can also be compared with the elements of $\alpha$ as follows:
for all $x\in\alpha$ and all "cuts" $E$ of $\alpha$, we have
$x<E$ (resp., $E<x$) iff $x\in E$ (resp., $x\nin E$).
Note that every element $x$ of $\alpha$ has two adjacent "cuts":
$x^- = \{y\in\alpha \mid y<x\}$ and $x^+ = \{y\in\alpha \mid y\le x\}$.
"Cuts" that are not of the form $x^-$ or $x^+$ are called ""natural"".

\begin{definition}\label{def:completion}
The ""completion"" of a linear order $\alpha$, denoted
$\completion{\alpha}$, is obtained from the disjoint union of the 
elements of $\alpha$ and the non-"extremal" "cuts" of $\alpha$,
and it is equipped with the extended ordering defined above.
\end{definition}

\AP
Note that the real line is obtained from the rational line using a similar 
notion of "completion" that only adds the non-"extremal" {\sl "natural"} "cuts". 
However, the difference between the real line and the "completion", as defined above, 
of the rational line is negligible, especially as far as "MSO" definability 
of rational sets is concerned. In particular, since the "natural cuts" in 
$\completion{\bbQ}$ are definable by first-order formulas, one can easily transform 
any "MSO" formula $\varphi(X_1,\ldots,X_m)$ to an "MSO" formula $\varphi'(X_1,\ldots,X_m)$ 
such that, for all sets $A_1,\ldots,A_m\subseteq\bbQ$, 
$(\bbR,<) \sat \varphi(A_1,\ldots,A_m)$ iff 
$(\completion{\bbQ},<)\sat \varphi'(A_1,\ldots,A_m)$.
As a consequence, to answer the question raised by Gurevich and Rabinovich, 
it is sufficient to prove the following result:

\begin{theorem}\label{theorem:cuts-at-the-background}
For every "MSO" formula $\varphi(X_1,\ldots,X_m)$, there is an "MSO" formula
$\tilde\varphi(X_1,\ldots,X_m)$ such that, for all countable linear orderings
$\alpha$ and all sets $A_1,\ldots,A_m\subseteq\alpha$,
$$
  \completion{\alpha} \sat \varphi(A_1,\ldots,A_m)
  \qquad\text{iff}\qquad
  \alpha \sat \tilde\varphi(A_1,\ldots,A_m) \ .
$$
\end{theorem}

\AP
Next, we generalize the notion of "completion" to words. 
We fix a dummy letter $""\cutlabel""$ that is intended to label the "cuts". 
The ""completion@word completion"" of a word $w:\alpha\then A$
is the word $\wcompletion{w}:\completion{\alpha}\then A\uplus\{\cutlabel\}$ 
defined by $\wcompletion{w}(x)=w(x)$, for all elements $x\in\alpha$, 
and $\wcompletion{w}(E)=\cutlabel$ 
for all "cuts" $E\in\completion{\alpha}\setminus\alpha$.

\AP
A simple, yet important, property is the relationship between the operation of
"completion" of a word and that of "product" of words, which is formalized 
in the following lemma (proof omitted). Intuitively, the "completion" of the 
"product" of a series of words is equivalent to a "variant of the product" 
on the "completions" of the words, where the "variant of the product" 
`fills the missing "cuts"'.

\begin{lemma}\label{lemma:completion}
For all linear orderings $\alpha$ and all words $(u_i)_{i\in\alpha}$, we have
$$
  \reallywidehat{\!\Big( \prod\nolimits_{i\in\alpha} u_i \Big)\!} 
  \:\:\:\:=~ \eprod\nolimits_{i\in\alpha} \wcompletion{u}_i
$$
where the product variant $""\eprod""$ is defined by 
$\eprod_{i\in\alpha} \wcompletion{v}_i = \prod\nolimits_{i\in\completion{\alpha}} v'_i$,
with $v'_i = v_i$ if $i\in\alpha$ and $v'_i = \cutlabel$ 
if $i\in\completion{\alpha}\setminus\alpha$.
\end{lemma}

\AP
A language $L$ of countable words is said to be 
""MSO definable with the cuts at the background""
if there exists an "MSO" sentence $\varphi$ such that 
$u\in L$ iff $\wcompletion{u} \sat \varphi$.
The following proposition is similar to the claim of Theorem \ref{theorem:mso-to-algebra}
(note that here we omit the part about effectiveness).

\begin{proposition}\label{prop:from-MSO-cuts-to-rec}
Languages of countable words that are "MSO definable with the cuts at the background" 
are "recognizable by $\countable$-monoids".
\end{proposition}

\begin{proof}
Recall that the proof of Theorem \ref{theorem:mso-to-algebra} was based on 
closure properties of recognizable $\countable$-languages under boolean operations 
and projections, which could be easily implemented at the level of the "$\countable$-algebras". 
Because in this proof we do not have to deal with effectiveness, it is convenient to
work directly at the level of "$\countable$-monoids". In particular, the monoids 
recognizing the considered languages will be defined using "logical types" and 
Shelah's composition method \cite{composition_method_shelah}. 
We shall consider "MSO" formulas up to syntactic equivalence, that is,
up to associativity, commutativity, idempotency, and distributivity of conjunctions
and disjunctions, commutativity of conjunctions with universal quantifications
and disjunctions with existential quantifications, and renamings of quantified 
variables.
Recall that, over a fixed finite signature with only relational symbols,
there exist only finitely many sentences up to syntactic equivalence.

Let $\varphi$ be an "MSO" sentence defining, "with the cuts at the background", 
a language $L\subseteq A^\countable$. Let $k$ be the ""quantifier rank of $\varphi$"", 
that is, the maximum number of nested quantifiers in $\varphi$.
Given a word $u$ of possibly uncountable domain, we define its ""$k$-type""
$\type{k}(u)$ as the (finite) set of all sentences of "quantifier rank" at most $k$.
We recall a simplified version of the composition theorem of Shelah, which shows
that the "type@$k$-type" of a "product" of words is uniquely determined by the 
"types@$k$-types" of the words:

\begin{claim}[Shelah's composition theorem \cite{composition_method_shelah}]\label{claim:composition-theorem}
Let $\alpha$ be a (possibly uncountable) linear ordering and, for every 
$i\in\alpha$, let $u_i,v_i$ be words (of possibly uncountable domains).
We have
$$
  \forall i\in\alpha \quad \type{k}(u_i) = \type{k}(v_i)
  \qquad\text{implies}\qquad
  \type{k}\Big( \prod\nolimits_{i\in\alpha} u_i \Big) = 
  \type{k}\Big( \prod\nolimits_{i\in\alpha} v_i \Big) \ .
$$
\end{claim}

\AP
To show that $L$ is "recognizable", we need to construct 
a "$\countable$-monoid" $(M,\pi)$ and a "morphism" $h$ from $A$ to $M$ 
such that $L = h^{-1}\big(h(L)\big)$. 
For this, we define the function $""\hattype{k}@\hattype""$ 
that maps any countable word to the "$k$-type" of its "completion", 
that is, $\hattype{k}(w) = \type{k}(\wcompletion{w})$.
The domain $M$ of the "$\countable$-monoid" is precisely the range of the function
$\hattype{k}$, that is,
$$
  M ~\eqdef~ \big\{\hattype{k}(w) ~\big\mid~ w\in A^\countable \big\} \ .
$$
We further let $""\word""$ be a function that maps any element $m\in M$ to
a word $\word(m)\in A^\countable$ such that $\hattype{k}\big( \word(m) \big) = m$.
The product $\pi$ of the "$\countable$-monoid" is defined as follows:
$$
  \pi\Big( \prod\nolimits_{i\in\alpha} m_i \Big)
  ~\eqdef~ \hattype{k} \Big( \prod\nolimits_{i\in\alpha}\word(m_i) \Big) \ .
$$
Even if we do not know yet that $\pi$ is a product (e.g., that it satisfies "generalized associativity"),
we can easily verify that the function $\hattype{k}$ behaves like a "morphism".
Formally, for all countable linear orderings $\alpha$ and all words $u_i\in A^\countable$, we have:
\ifarxiv
\begin{align*}
  \hattype{k} \Big( \,\prod\nolimits_{i\in\alpha} u_i\, \Big)
  &~=~ \type{k} \Big( \mspace{-5mu} 
                          \reallywidehat{ ~\prod\nolimits_{i\in\alpha} u_i~ } 
                          \mspace{-5mu} \Big) 
       \tag{\text{by definition of $\hattype{k}$}} \\
  &~=~ \type{k} \Big( \,\eprod\nolimits_{i\in\alpha} \wcompletion{u}_i \Big) 
       \tag{\text{by Lemma \ref{lemma:completion}}} \\
  &~=~ \type{k} \Big( \,\eprod\nolimits_{i\in\alpha} 
                          \mspace{-18mu} \reallywidehat{ \quad\word\big( \hattype{k}(u_i) \big)~~\quad }
                          \mspace{-22mu} \Big) 
       \tag{\text{by Claim \ref{claim:composition-theorem}}} \\
  &~=~ \hattype{k} \Big( \,\prod\nolimits_{i\in\alpha} 
                                 \;\word\big( \hattype{k}(u_i) \big)\, \Big)
       \tag{\text{by Lemma \ref{lemma:completion}}} \\
  &~=~ \pi \Big( \,\prod\nolimits_{i\in\alpha} u_i\, \Big)
       \tag{\text{by definition of $\pi$}}
\end{align*}
\fi
Moreover, since $\hattype{k}$ is surjective from $A^\countable$ to $M$, 
the property of being a "$\countable$-monoid" is transferred from
$"(A^\countable,\prod)"$ to $(M,\pi)$. 
Hence, $(M,\pi)$ is a "$\countable$-monoid". 
Finally, if we let $h=\hattype{k}$ and we consider two words 
$u,v\in A^\countable$ such that $h(u)=u(v)$, we get
$u\in L$ iff $\wcompletion{u}\sat \varphi$ iff $\varphi \in \hattype{k}(u) = h(u)$ 
iff $\varphi \in \hattype{k}(v)$ iff $v\in L$.
This shows that $L$ is "recognized by the $\countable$-monoid" $(M,\pi)$ 
via the "morphism" $h=\hattype{k}$.
\end{proof}

\smallskip
\AP
Proposition \ref{prop:from-MSO-cuts-to-rec} combined with 
Theorem \ref{theorem:rec-to-mso} shows that the languages
definable in "MSO logic with the cuts at the background" 
are also definable in classical "MSO" logic:

\begin{corollary}\label{cor:MSO-cuts}
Languages of countable words that are "MSO definable with the cuts at the background" 
are "MSO" definable.
\end{corollary}

\AP
Finally, if we restate the above corollary in terms of relational
structures, we get precisely the claim of Theorem \ref{theorem:cuts-at-the-background}.

\medskip
\subsection{Yields of tree languages}\label{subsec:yields}

We conclude the section by considering another open problem related to countable words.
More precisely, we will consider "yields" of trees, that is, words spelled out by 
frontiers of trees following the natural left-to-right order 
\cite{frontiers_of_infinite_trees,generalized_yields}.\footnote{We remark that a different notion of yield was introduced
          in \cite{continuous_monoids_and_yields}, based
          on a specific continuous function that maps trees to 
          finite or $\omega$-words.}
We restrict ourselves to ""labelled binary trees"", namely, 
trees in which every node has an associated label from a finite 
alphabet and every internal node has exactly two (ordered) successors. 
These "trees" may contain leaves as well as infinite paths. 

\begin{definition}\label{def:yield}
The ""yield"" of a tree $t$ is the word $\yield(t)$ whose domain 
is the set of leaves of $t$, ordered by the infix relation, such that 
$\yield(t)(x) = t(x)$ for all leaves $x$.
\end{definition}

\AP
Given two "trees" $t,t'$ and a set $X$ of leaves of $t$, we denote by $t[X/t']$ 
the "tree" resulting from the simultaneous substitution in $t$ of all leaves 
$x\in X$ by $t'$. This substitution operation is compatible with the 
analogous operation of substitution on "yields", that is, for all 
$X \subseteq \dom\big(\yield(t)\big)$, we have
$$ 
  \yield\big(t[X/t']\big) ~=~ \yield(t)\big[X/\yield(t')\big] \ .
$$
By a slight abuse of notation, given a letter $a$ occurring at some leaves of $t$, 
we denote by $t[a/t']$ the result of the simultaneous substitution in $t$ of all 
$a$-labelled leaves by $t'$, and similarly for $\yield(t)\big[a/\yield(t')\big]$.

\AP
Every word of countable domain can be seen as the "yield" of some "tree". 
Indeed, this holds trivially for every word indexed over the rationals.
Moreover, every word $w$ of countable domain can be obtained from a word 
$w'$ over the rationals by removing some positions. This latter operation 
of removing positions can be implemented at the level of "trees" by a 
substitution: if $w = \yield(t)$ and $X \subseteq \dom(w)$, 
then $w[X/\emptystr] = \yield\big(t[X/t_\emptystr]\big)$, where 
$t_\emptystr$ is the infinite complete binary tree, whose "yield" is
the empty word.

\AP
We can also extend the "yield" function to any language $T$ of "trees"
by letting $\yield(T) = \big\{ \yield(t) \:\big\mid\: t\in T \big\}$.
Similarly, given a language $L$ of words, we define the corresponding 
"tree" language as 
$""\invyield""(L) = \big\{ t \:\big\mid\: \yield(t)\in L \big\}$.
We say that a "tree" language $T$ is ""yield-invariant""
if, for all "trees" $t,t'$ such that $\yield(t)=\yield(t')$, 
we have $t\in T$ iff $t'\in T$ (or, equally, if $T = \invyield(\yield(T))$).

\AP
It is known (see, for instance, \cite{yields_of_regular_tree_languages})
that the "yield" of a regular language $T$ of finite "trees" is a 
context-free language, and in general it is not regular. However, when the regular 
"tree" language $T$ is also "yield-invariant", the "yield" language 
$\yield(T)$ is shown to be regular \cite{yields_formal_languages_handbook}.
A converse result also holds: if $L$ is a regular language of 
finite words, then $T = \invyield(L)$ is "yield-invariant" and regular.
The work \cite{generalized_yields} raises the natural question of whether 
analogous properties hold between languages of possibly infinite "trees" 
and languages of words of countable domains.
Below, we answer positively to this question by exploiting again the correspondence
between "MSO" logic and "$\countable$-algebras".

\begin{theorem}\label{theorem:yields}
Let $L$ be a language of countable words and let $T = \invyield(L)$ 
be the corresponding "yield-invariant" language of "trees".
Then, $L$ is "MSO" definable iff $T$ is "MSO" definable.
\end{theorem}

\AP
The proof of the left-to-right direction is straightforward:
if $L$ is defined by an "MSO" sentence $\varphi$, then we can 
construct another "MSO" sentence $\varphi'$ that, when interpreted
on a "tree", checks that the frontier satisfies $\varphi$;
the sentence $\varphi'$ defines precisely the language
$T = \invyield(L)$.

\AP
The proof of the converse direction is not immediate, since,
a priori, checking whether a given word $w$ belongs to $L$
requires guessing some "tree" $t\in T$ such that $\yield(t)=w$.
To show that $L$ is "recognizable" by $\countable$-monoids, and hence 
definable in "MSO" logic, we will construct a "$\countable$-algebra" 
on the basis of a suitable "congruence" defined from $T$.

\begin{definition}\label{def:congruence}
Let $T$ be a "tree" language over the alphabet $A$ and let $""\placeholder""\nin A$ 
be a fresh letter that will be used as a placeholder for substitution. 
We denote by $""\congT""$ the equivalence on "trees" defined 
by $t_1 \congT t_2$ iff, for all "trees" $t$ labelled over the 
alphabet $A\uplus\{\placeholder\}$, we have $t[\placeholder/t_1] \in L ~\iff~ t[\placeholder/t_2] \in L$.
We say that a tree $t_1$ ""inhabits"" a $\congT$-equivalence class 
$[t_2]_{\congT}$ if $t_1 \congT t_2$.
\end{definition}

\AP
We now show some simple but fundamental properties of the relation $\congT$.
The first property is that $\congT$ correctly abstracts "trees" with the same 
"yield", provided that the language $T$ is "yield-invariant". Formally, if $T$ 
is "yield-invariant" and $t_1$ and $t_2$ are two "trees" such that 
$\yield(t_1) = \yield(t_2)$, then we have $t_1 \congT t_2$. 
It is also easy to verify that $\congT$ is a ""congruence"" with 
respect to the substitution operation, that is, $t_1 \congT t_2$ implies 
$t[\placeholder/t_1] \congT t[\placeholder/t_2]$.

\AP
Another crucial property that is used to prove Theorem \ref{theorem:yields} is
based on Rabin's tree theorem \cite{s2s}, which shows that "MSO" definable "tree"
languages can be equivalently described by means of "automata". Below, we recall 
some basic knowledge about "tree automata", their problems, and the translation 
from "MSO" logic. We begin by introducing a variant of parity tree automaton that 
can parse "trees" containing leaves and/or infinite paths:

\begin{definition}\label{def:parity-tree-automaton}
A ""parity tree automaton"" is a tuple $\cA = (A,Q,I,\Delta,\Omega)$, 
where $A$ is a finite set of node labels, $Q$ is a finite set of states,
$I\subseteq Q$ is a set of initial states, 
$\Delta\subseteq (Q \times A) \uplus (Q \times A \times Q \times Q)$ is 
a set of transition rules, and $\Omega:Q\then\bbN$ is a priority function.
A ""successful run of $\cA$ on a tree $t$"" is a "tree" $\rho$ that 
has the same domain as $t$ and satisfies:
\begin{itemize}
  \item $\rho(x_0)\in I$, where $x_0$ is the root of $\rho$;
  \item for all leaves $x$ of $\rho$, $\big(\rho(x),t(x)\big)\in \Delta$;
  \item for all internal nodes $x$ of $\rho$, $\big(\rho(x),t(x),\rho(x_1),\rho(x_2)\big)\in \Delta$,
        where $x_1$ and $x_2$ are the left and right successors of $x$, respectively;
  \item for all infinite paths $\pi$ in $\rho$, $\limsup\big(\Omega(\rho|_\pi)\big)$ is even,
        where $\Omega(\rho|_\pi)$ denotes the sequence of priorities associated with the
        states along the path $\pi$ and $\limsup\big(\Omega(\rho|_\pi)\big)$ returns the 
        maximal priority that occurs infinitely often in the sequence $\Omega(\rho|_\pi)$.
\end{itemize}
The language ""recognized@recognized by an automaton"" by $\cA$ 
is the set $\sL(\cA)$ of all "trees" $t$ that admit a "successful run" of $\cA$. 
\end{definition}

\AP
We recall that the ""emptiness problem"" for "parity tree automata",
that is, the problem of testing whether $\lang(\cA)=\emptyset$ for any given
"parity tree automaton" $\cA$, is decidable.
The ""containment"" and ""equivalence problems"" can be reduced to 
the "emptiness problem" by exploiting effective closures of automata under 
intersection and complementation: indeed, we have 
$\lang(\cA) \subseteq \lang(\cA')$ iff 
$\lang(\cA)\cap\lang(\overline{\cA'}) = \emptyset$,
where $\overline{\cA'}$ denotes the automaton "recognizing@recognized by an automaton"
the complement of the language $\lang(\cA')$.
There is another fundamental problem that is known to be decidable, called
""membership problem"". This amounts at testing whether a given "tree" $t$
belongs to the language "recognized@recognized by an automaton" by a given 
"parity tree automaton" $\cA$.
For this problem to make sense, however, we need to specify how the "tree" $t$
is provided in input. A simple solution is to restrict to ""regular trees"", 
that is, "trees" that contain only finitely many non-isomorphic subtrees. 
It is easy to see that any "regular tree" can be finitely represented by a 
"parity tree automaton" $\cB$ that "recognizes@recognized by an automaton" 
the singleton language $\{t\}$.
The closure of "parity tree automata" under intersection implies that the 
"membership problem" is decidable: for every "regular tree" $t$ represented 
by the singleton language $\lang(\cB)=\{t\}$, we have $t\in\lang(\cA)$ iff 
$\lang(\cB)\cap\lang(\cA) \neq \emptyset$. 

\AP
We recall below the correspondence between "MSO" sentences interpreted on 
"trees" and "parity tree automata". A proof of this correspondence can 
be found in \cite{languages_automata_logic} and is based on closure properties 
of "parity tree automata" under boolean operations and projections 
(originally, this was established by Rabin in \cite{s2s} using a 
different model of automaton).

\begin{theorem}[Translation of MSO to tree automata \cite{languages_automata_logic}]\label{theorem:MSO-to-automata}
One can effectively translate any "MSO" sentence $\varphi$ 
that defines a "tree" language $T$
into a "parity tree automaton" $\cA$ that "recognizes@recognized by an automaton" $T$. 
\end{theorem}

\AP
We are now ready to prove the following key lemma:

\begin{lemma}\label{lemma:finite-index}
For every "MSO" definable "tree" language $T$, $\congT$ has finite index, 
namely, there exist only finitely many $\congT$-equivalence classes.
Moreover, given an "MSO" sentence defining $T$, one can decide whether 
$t_1 \congT t_2$, for any pair of "regular trees" $t_1$ and $t_2$, 
and one can compute a finite set of "regular trees" that "inhabit" all 
$\congT$-equivalence classes.
\end{lemma}

\begin{proof}
Let $\varphi$ be an "MSO" sentence defining the "tree" language $T$ and let 
$\cA=(A,Q,I,\Delta,\Omega)$ be the corresponding "parity tree automaton"
"recognizing@recognized by an automaton" $T$, 
obtained from Theorem \ref{theorem:MSO-to-automata}. 
Given a generic "tree" $t$, we abstract the behaviour of $\cA$ on $t$ 
by introducing the ""$\cA$-type"" of $t$, defined as 
$$
  \atype{\cA}(t) ~\eqdef~ \big\{ q\in Q ~\big\mid~ t\in\lang(\cA^q) \big\} 
$$
where $\cA^q$ is the "automaton" obtained from $\cA$ by replacing the 
set $I$ of initial states with the singleton $\{q\}$. Note that there 
are at most $2^{|Q|}$ different "$\cA$-types" of "trees".
Based on this, we can establish the first claim of the lemma by simply 
showing that the type-equivalence induced by $\cA$ refines the $\congT$-equivalence, 
namely, that for all "trees" $t_1$ and $t_2$, $\atype{\cA}(t_1) = \atype{\cA}(t_2)$ 
implies $t_1 \congT t_2$. Consider two "trees" $t_1,t_2$ such that 
$\atype{\cA}(t_1) = \atype{\cA}(t_2)$ and another "tree" $t$ labelled 
over the extended alphabet $A\uplus\{\placeholder\}$. 

\AP
We first prove that "$\cA$-types" are compatible with tree substitutions, that is, 
knowing that $\atype{\cA}(t_1) = \atype{\cA}(t_2)$, we get
$$
  \atype{\cA}\big(t[\placeholder/t_1]\big) ~=~ \atype{\cA}\big(t[\placeholder/t_2]\big) \ .
$$
Consider a state $q \in \atype{\cA}\big(t[\placeholder/t_1]\big)$, namely, such that
$t[\placeholder/t_1]\in\lang(\cA^q)$. Let $\rho$ be a "successful run" of $\cA^q$ 
on $t[\placeholder/t_1]$ and let $X$ be the set of $c$-labelled leaves of $t$. 
The set $X$ can be equally seen as a set of nodes of $\rho$. 
We partition $X$ into some subsets $X_{q'}$, where
$q'\in Q$ and $X_{q'} = \{ x\in X \:\mid\: \rho(x) = q' \}$, and for every $x\in X_{q'}$, 
we let $\rho_x$ be the subtree of $\rho$ starting at node $x$. Note that each
subrun $\rho_x$, with $x\in X_{q'}$, is a "successful run" of the "automaton" $\cA^{q'}$
on the "tree" $t_1$. This means that $q'\in\atype{\cA}(t_1)$ for all non-empty sets $X_{q'}$.
Since $\atype{\cA}(t_1) = \atype{\cA}(t_2)$, we derive that $q' \in \atype{\cA}(t_2)$ 
for all non-empty sets $X_{q'}$. Thus, there exist "successful runs" $\rho'_x$ of $\cA^{q'}$ 
on $t_2$, for all $x\in X_{q'}$. Next, we define the "tree" $\rho'$ by substituting in $\rho$
every subtree $\rho_x$ starting at node $x\in X$ with the "tree" $\rho'_x$ (note that 
the substitution is performed simultaneously on nodes that may not be leaves, but 
these nodes are still pairwise incomparable with respect to the descendant relation). 
Since $\rho(x) = \rho'(x)$ for all $x\in X$, we deduce that $\rho'$ is a "successful run" 
of $\cA^q$ on $t[\placeholder/t_2]$. This proves that $q \in \atype{\cA}\big(t[\placeholder/t_2]\big)$.
Symmetric arguments show that $q \in \atype{\cA}\big(t[\placeholder/t_2]\big)$ implies 
$q \in \atype{\cA}\big(t[\placeholder/t_1]\big)$.

\AP
Now that we know that 
$\atype{\cA}\big(t[\placeholder/t_1]\big) = \atype{\cA}\big(t[\placeholder/t_2]\big)$,
we can conclude the proof of the first claim by observing that 
$$
  t[\placeholder/t_1] \in L
  \quad\text{iff}\quad
  \atype{\cA}\big(t[\placeholder/t_1]\big) \cap I \neq \emptyset
  \quad\text{iff}\quad
  \atype{\cA}\big(t[\placeholder/t_2]\big) \cap I \neq \emptyset
  \quad\text{iff}\quad
  t[\placeholder/t_2] \in L
$$
and hence $t_1 \congT t_2$. This shows that $\congT$ has finite index.

\smallskip
\AP
We turn to the proof of the second claim.
Consider two "regular trees" $t_1$ and $t_2$ represented by singleton languages $\lang(\cB_1)=\{t_1\}$
and $\lang(\cB_2)=\{t_2\}$, respectively. Recall that $t_1 \congT t_2$ iff for all "trees" $t$ 
labelled over $A\uplus\{\placeholder\}$, either both "trees" 
$t[\placeholder/t_1]$ and $t[\placeholder/t_2]$ are inside $\lang(\cA)$, or neither of them are. 
Further note that $t[\placeholder/t_i]\in \lang(\cA)$ iff there is a state $q\in Q$
such that $t_i\in\lang(\cA^q)$ and $t\in\lang(\cA_q)$, where $\cA_q$ is the 
"automaton" obtained from $\cA$ by replacing the transition relation $\Delta$ with
$$
  \Delta_q ~\eqdef~ \big( \Delta \cap (Q\times A\times Q\times Q) \big) ~\uplus~
                    \big( \Delta \cap (Q\times (A\setminus\{c\})) \big) ~\uplus~
                    \big( \{q,c\} \big)
$$
(intuitively, $\cA_q$ behaves exactly as $\cA$ on all nodes of the "tree" $t$, with the 
only exception of the $c$-labelled leaves, which must be associated with state $q$).
Using the above properties, we can restate the equivalence $t_1 \congT t_2$ as 
a (decidable) boolean combination of "emptiness problems":
$$
\begin{array}{rrrl}
  \mspace{-10mu}
  t_1 \congT t_2 
  &\text{iff}&
  \displaystyle\bigwedge_{q\in Q}\bigvee_{q'\in Q}&
  \bigg(~ 
    \underbrace{\cB_1 \cap \cA^q \neq \emptyset}_{t_1 \,\in\, \cA^q} ~\et~
    \underbrace{\cA_q \cap \cA_{q'} \neq \emptyset}_{\exists \, t \,\in\, \cA_q \,\cap\, \cA_{q'}} ~\et~
    \underbrace{\cB_2 \cap \cA^{q'} \neq \emptyset}_{t_2 \,\in\, \cA^{q'}}
  ~\bigg) \\[1ex]
  &\et&
  \displaystyle\bigwedge_{q'\in Q}\bigvee_{q\in Q}&
  \bigg(~ 
    \underbrace{\cB_1 \cap \cA^q \neq \emptyset}_{t_1 \,\in\, \cA^q} ~\et~
    \underbrace{\cA_q \cap \cA_{q'} \neq \emptyset}_{\exists \, t \,\in\, \cA_q \,\cap\, \cA_{q'}} ~\et~
    \underbrace{\cB_2 \cap \cA^{q'} \neq \emptyset}_{t_2 \,\in\, \cA^{q'}}
  ~\bigg) \\[1ex]
  &\et&
  \displaystyle\bigvee_{q\in Q}\bigwedge_{q'\in Q}&
  \bigg(~ 
    \underbrace{\cB_1 \cap \cA^q = \emptyset}_{t_1 \,\nin\, \cA^q} ~\et~
    \underbrace{\cB_2 \cap \cA^{q'} = \emptyset}_{t_2 \,\nin\, \cA^{q'}}
  ~\bigg)
  ~\vee~
  \bigg(~ 
    \underbrace{\cA_q \cup \cA_{q'} = \emptyset}_{\not\exists \, t \,\in\, \cA_q \,\cap\, \cA_{q'}} ~\et~
  ~\bigg) \\[1ex]
  &\et&
  \displaystyle\bigvee_{q'\in Q}\bigwedge_{q\in Q}&
  \bigg(~ 
    \underbrace{\cB_1 \cap \cA^q = \emptyset}_{t_1 \,\nin\, \cA^q} ~\et~
    \underbrace{\cB_2 \cap \cA^{q'} = \emptyset}_{t_2 \,\nin\, \cA^{q'}}
  ~\bigg)
  ~\vee~
  \bigg(~ 
    \underbrace{\cA_q \cup \cA_{q'} = \emptyset}_{\not\exists \, t \,\in\, \cA_q \,\cap\, \cA_{q'}} ~\et~
  ~\bigg).
\end{array} 
$$
(for simplicity, we identified "automata" and the "recognized@recognized by an automaton"
languages).

\AP
Finally, to compute a set of "regular trees" that inhabit all $\congT$-equivalence classes, 
we consider again "$\cA$-types". We first show how to associate with each "$\cA$-type" $\sigma$ 
a corresponding "regular tree" $t_\sigma$ such that $\atype{\cA}(t_\sigma)=\sigma$.
We do so by solving a series of "emptiness problems". 
Indeed, we recall that an "$\cA$-type" is any set $\sigma$ 
of states of $\cA$ such that the language 
$\displaystyle\bigcap\nolimits_{q\in \sigma}\lang(\cA^q) \:\cap\:
 \displaystyle\bigcap\nolimits_{q\nin \sigma}\lang(\overline{\cA^q})$ is non-empty.
Moreover, if the latter language is non-empty, then it contains a "regular tree"
$t_\sigma$ that can be effectively constructed from $\sigma$. 
Clearly, we have $\atype{\cA}(t_\sigma)=\sigma$ and hence $t_\sigma$ 
can be used as a representant of the "$\cA$-type" $\sigma$.
Towards a conclusion, we can construct a list of "regular trees" $t_1,\ldots,t_n$, 
one for each "$\cA$-type". Since the equivalence $\congT$ is refined by the 
type-equivalence induced by $\cA$, we know that every $\congT$-equivalence class
is "inhabited" by at least one tree among $t_1,\ldots,t_n$.
If needed, we can also exploit the decidability of $\congT$ to select a 
minimal subsequence $t_{i_1},\ldots,t_{i_m}$ of "regular" "inhabitants" 
of all $\congT$-equivalence classes.
\end{proof}

\smallskip
\AP
We can now prove the right-to-left direction of Theorem \ref{theorem:yields}.
Let $T$ be a "yield-invariant language" defined by an "MSO" sentence $\varphi$.
We will exploit Lemma \ref{lemma:finite-index} and the fact that $\congT$ 
is a yield-invariant equivalence compatible with tree substitutions to construct 
a "$\countable$-algebra" $(M,1,\cdot,\tau,\tauop,\kappa)$ 
"recognizing@recognized by a $\countable$-algebra" 
the language $L = \yield(T)$.
Formally, we define $M$ to be the set of all $\congT$-equivalence classes.
We recall that this set is finite and that $\congT$-equivalence classes can 
be effectively manipulated through their ""regular inhabitants"", that is,
by means of representants that have the form of "regular trees". 
We define the operators of the algebra as follows:
\begin{itemize}
  \item $1$ is the $\congT$-equivalence class of the infinite complete tree $t_\emptystr$.
        Note that this "tree" $t_\emptystr$ has no leaves, and hence its "yield" is 
        the empty word.
        Moreover, $t_\emptystr$ is "regular", and hence it can be used as 
        a "regular inhabitant" of its $\congT$-equivalence class.
  \item $\cdot$ is the function that maps any pair of $\congT$-equivalence classes
        $[t_1]_{\congT}$ and $[t_2]_{\congT}$ to the $\congT$-equivalence class
        $$
          [t_1]_{\congT}\cdot[t_2]_{\congT} ~\eqdef~ 
          \big[ t_{a_1 a_2} [a_1/t_1] [a_1/t_2] \big]_{\congT}
        $$ 
        where $t_{a_1 a_2}$ is a fixed "tree" such 
        that $\yield(t) = a_1\,a_2$, and $a_1,a_2$ are distinct 
        fresh letters not occurring in the alphabet of $t_1$ and $t_2$.
        For example, $t_{a_1 a_2}$ can be chosen to be the "tree"
        $$
          t_{a_1 a_2} ~=  \begin{tikzpicture}[baseline=4pt]
                            \tikzstyle{level 1} = [level distance=0.5cm, sibling distance=0.75cm]
                            \tikzstyle{treenode} = [inner sep = 0mm, minimum size = 4.5mm, circle]
                            \draw (5,0.25) node [treenode] {$\bullet$}
                            child {node [treenode] {$a_1$}}
                            child {node [treenode] {$a_2$}};
                          \end{tikzpicture}
        $$
        where the label $\bullet$ of the root is immaterial.
        Note that the $\congT$-equivalence class $[t_1]_{\congT}\cdot[t_2]_{\congT}$ 
        is well defined thanks to the fact that $\congT$ is a "congruence".
        Moreover, because the "tree" $t_{a_1 a_2}$ is "regular", a "regular inhabitant"
        of the class $[t_1]_{\congT}\cdot[t_2]_{\congT}$ can be effectively 
        constructed from some "regular inhabitants" of $[t_1]_{\congT}$ and 
        $[t_2]_{\congT}$.
  \item $\tau$ is the function that maps any $\congT$-equivalence class 
        $[t_1]_{\congT}$ to the $\congT$-equivalence class
        $$
          \qquad\qquad\qquad
          [t_1]_{\congT}^\tau ~\eqdef~ \big[ t_\omega[a/t_1] \big]_{\congT}
          \qquad\qquad\qquad\text{where}\quad
          t_\omega ~=  \begin{tikzpicture}[baseline=-2pt]
                         \tikzstyle{level 1} = [level distance=0.5cm, sibling distance=0.75cm]
                         \tikzstyle{treenode} = [inner sep = 0mm, minimum size = 4.5mm, circle]
                         \draw (0,0) node [treenode] {$\bullet$}
                         child {node [treenode] {$a$}}
                         child {node [treenode] {$\bullet$}
                           child {node [treenode] {$a$}}
                           child {node [treenode] {\raisebox{2.3pt}{.}\hspace{-1pt}.\hspace{-1pt}\raisebox{-2.3pt}{.}}
                           }
                         };
                       \end{tikzpicture}
        $$ 
        Again, since $t_\omega$ is a "regular tree", a "regular inhabitant" 
        of the class $[t_1]_{\congT}^\tau$ can be computed from a 
        "regular inhabitant" of the class $[t_1]_{\congT}$.
  \item $\tauop$ is defined similarly to $\tau$, where $t_\omega$ is replaced
        by the "tree"
        $$
          t_{\omegaop} ~=  \begin{tikzpicture}[baseline=-2pt]
                             \tikzstyle{level 1} = [level distance=0.5cm, sibling distance=0.75cm]
                             \tikzstyle{treenode} = [inner sep = 0mm, minimum size = 4.5mm, circle]
                             \draw (0,0) node [treenode] {$\bullet$}
                             child {node [treenode] {$\bullet$}
                               child {node [treenode]                                            {\raisebox{-2.3pt}{.}\hspace{-1pt}.\hspace{-1pt}\raisebox{2.3pt}{.}}
                               }
                               child {node [treenode] {$a$}}
                             }
                             child {node [treenode] {$a$}};
                           \end{tikzpicture}
        $$ 
  \item $\kappa$ is the function that maps any set $\big\{ [t_1]_{\congT},\ldots,[t_k]_{\congT} \big\}$
        of $\congT$-equivalence classes to the $\congT$-equivalence class
        $$
          \big\{ [t_1]_{\congT},\ldots,[t_k]_{\congT} \big\}^\kappa 
          ~\eqdef~ 
          \big[ t_\eta [a_1/t_1] \ldots [a_k/t_k] \big]_{\congT}
        $$ 
        where $t_\eta$ is a fixed "regular tree" with "yield" $\{a_1,\ldots,a_k\}\etapow$
        and $a_1,\ldots,a_k$ are fresh letters.
        For example, $t_\eta$ can be defined by a "parity tree automaton"
        so as to satisfy the following equation:
        $$
          t_\eta ~=  \begin{tikzpicture}[baseline=-2pt]
                       \tikzstyle{level 1} = [level distance=0.5cm, sibling distance=0.75cm]
                       \tikzstyle{treenode} = [inner sep = 0mm, minimum size = 4.5mm, circle]
                       \draw (0,0) node [treenode] {$\bullet$}
                       child {node [treenode] {$t_\eta$}}
                       child {node [treenode] {$\bullet$}
                         child {node [treenode] {$a_1$}}
                         child {node [treenode] {$\bullet$}
                           child {node [treenode] {$t_\eta$}}
                           child {node [treenode]                                        {\raisebox{2.3pt}{.}\hspace{-1pt}.\hspace{-1pt}\raisebox{-2.3pt}{.}}
                             child [draw=white] {node {}}
                             child {node [treenode] {$\bullet$}
                               child {node [treenode] {$t_\eta$}}
                               child {node [treenode] {$\bullet$}
                                 child {node [treenode] {$a_k$}}
                                 child {node [treenode] {$t_\eta$}}
                               }
                             }
                           }
                         }
                       };
                \end{tikzpicture}
        $$ 
\end{itemize}

\noindent
Below, we verify that the structure $(M,1,\cdot,\tau,\tauop,\kappa)$ obtained from the 
"automaton" $\cA$ is indeed a "$\countable$-algebra", that is, it satisfies 
Axioms~\refaxiom{axiom:concatenation}-\refaxiom{axiom:identity} of Definition \ref{def:o-algebra}.

\begin{lemma}\label{lemma:yield-invariant}
The structure $(M,1,\cdot,\tau,\tauop,\kappa)$ obtained from $\congT$ is a "$\countable$-algebra".
\end{lemma}

\begin{proof}
The fact that the structure $(M,1,\cdot,\tau,\tauop,\kappa)$ satisfies
Axioms~\refaxiom{axiom:concatenation}-\refaxiom{axiom:identity} 
follows almost directly from its definition and from the fact that the 
equivalence $\congT$ is yield-invariant, that is, $t_1 \congT t_2$ 
whenever $\yield(t_1) = \yield(t_2)$.
For example, recall the definition of the binary operator $\cdot$~: 
for all pairs of "trees" $t_1,t_2$, we have 
$$
  [t_1]_{\congT} \cdot [t_2]_{\congT}
  ~=~ 
  \left[
    \begin{tikzpicture}[baseline=-3pt]
      \tikzstyle{level 1} = [level distance=0.5cm, sibling distance=0.75cm]
      \tikzstyle{treenode} = [inner sep = 0mm, minimum size = 4.5mm, circle]
      \draw (5,0.25) node [treenode] {$\bullet$}
      child {node [treenode] {$t_1$}}
      child {node [treenode] {$t_2$}};
    \end{tikzpicture}  
  \right]_{\congT} .
$$
From this, we easily deduce that $\cdot$ satisfies Axiom \refaxiom{axiom:concatenation}:
\vspace{-1mm}
\ifarxiv
\begin{align*}
  \big( [t_1]_{\congT} \cdot [t_2]_{\congT} \big) \cdot [t_3]_{\congT} 
  &~=~
  \left[
    \begin{tikzpicture}[baseline=-18pt]
      \tikzstyle{level 1} = [level distance=0.5cm, sibling distance=0.75cm]
      \tikzstyle{treenode} = [inner sep = 0mm, minimum size = 4.5mm, circle]
      \draw (0,0) node [treenode] {$\bullet$}
      child {node [treenode] {$\bullet$}
        child {node [treenode] {$t_1$}}
        child {node [treenode] {$t_2$}}
      }
      child {node [treenode] {$t_3$}};
    \end{tikzpicture}             
  \right]_{\congT} 
  \tag{by definition} \\
  &~=~
  \left[
    \begin{tikzpicture}[baseline=-18pt]
      \tikzstyle{level 1} = [level distance=0.5cm, sibling distance=0.75cm]
      \tikzstyle{treenode} = [inner sep = 0mm, minimum size = 4.5mm, circle]
      \draw (0,0) node [treenode] {$\bullet$}
      child {node [treenode] {$t_1$}}
        child {node [treenode] {$\bullet$}
          child {node [treenode] {$t_2$}}
          child {node [treenode] {$t_3$}
        }
      };
    \end{tikzpicture}             
  \right]_{\congT} 
  \tag{by yield-invariance} \\
  &~=~
  [t_1]_{\congT} \cdot \big( [t_2]_{\congT} \cdot [t_3]_{\congT} \big) \ .
  \tag{by definition} \\[-2ex]
\end{align*}
\fi
We omit the analogous arguments showing that $1$, $\tau$, $\tauop$, and $\kappa$ 
satisfy the remaining Axioms~\refaxiom{axiom:omega}-\refaxiom{axiom:identity}.
\end{proof}

\smallskip
\AP
Combining the above lemma, Corollary \ref{cor:algebra-to-monoid}, and Theorem \ref{theorem:rec-to-mso}
gives the right-to-left direction of Theorem \ref{theorem:yields}.

\smallskip
\AP
We also remark that, if the "MSO" definable "tree" language $T$ is not known to be 
"yield-invariant", we can still construct the structure $(M,1,\cdot,\tau,\tauop,\kappa)$ 
from $\congT$.
Below, we explain how to use this structure to decide whether $T$ is "yield-invariant". 
We follow the same approach described in Section \ref{sec:algebra-to-logic} and we construct,
using the operators of $(M,1,\cdot,\tau,\tauop,\kappa)$, a family of "MSO" sentences 
of the form $\varphival{\sigma}$, where $\sigma$ ranges over the set of possible 
$\congT$-equivalence classes. Given a word $w$, these sentences can be used to 
derive the $\congT$-equivalence class of some "tree" $t_w$ such that 
$\yield(t_w)=w$. In particular, we can define in "MSO" logic a word language 
of the form $L = \big\{ w \:\big\mid\: t_w \in \lang(\cA) \big\}$.
We can then use the left-to-right implication of Theorem \ref{theorem:yields}
to derive an "MSO" sentence defining the "tree" language $T' = \invyield(L)$.
Now, if $T$ is "yield-invariant", then 
$T' = \invyield(L) = \invyield(\yield(T)) = T$, 
as shown by Theorem \ref{theorem:yields}.
Conversely, if $T' = T$, then $T$ is clearly "yield-invariant".
We thus reduced the problem of deciding whether an "MSO" definable "tree" language $T$ is 
"yield-invariant" to the equivalence problem for "MSO" sentences interpreted on "trees", 
which is known to be decidable.

\begin{theorem}\label{th:yield-invariant-check}
The problem of deciding whether a "tree" language $T$ defined by an "MSO" sentence
is "yield-invariant" is decidable.
\end{theorem}
 
\smallskip
\section{Conclusion}\label{sec:conclusion}

We have introduced an algebraic notion of "recognizability" for "languages" of 
countable "words" and we have shown the correspondence with the family of "languages"
definable in "MSO" logic. 
As a side-product of this result, we obtained that the hierarchy of monadic quantifier 
alternation for "MSO" logic interpreted over countable "words" collapses to its 
"$\exists\forall$-fragment" (or, equally, to its "$\forall\exists$-fragment"). 
The collapse result is optimal in the sense that there are "recognizable"
"languages" that are not definable in the "$\exists$-fragment".
Our techniques are then used to solve an open problem posed by Gurevich and Rabinovich, 
concerning the definability of properties of sets of rationals using "MSO" formulas
interpreted over the real line (definability "with the cuts at the background").
Finally, we exploited the correspondence between logic and "algebras" to solve another 
open problem posed by Bruyère, Carton, and Sénizergues, concerning the characterization 
of properties of "trees" that can be defined in "MSO" logic and that are "yield-invariant".

We conclude by mentioning the possibility of defining models of automata that extend 
those from \cite{automata_on_linear_orderings} and that capture precisely the expressiveness
of "MSO" logic over words of countable domains. However, such automata need to have complicated 
acceptance conditions in order to distinguish between scattered and non-scattered words and, 
more generally, to enjoy closure properties under boolean operations and projections. 
The definition of an automaton model for languages of countable words is thus not as 
natural as that of "$\countable$-monoid".

\ifarxiv
\bibliographystyle{alpha}
\bibliography{biblio}
\fi

\end{document}